\pdfoutput=1
\RequirePackage{ifpdf}
\ifpdf
\documentclass[pdftex]{sigma}
\else
\documentclass{sigma}
\fi

\usepackage[cmtip,arrow]{xy}
\usepackage{pb-diagram,pb-xy}

\DeclareMathOperator{\Ker}{Ker}
\DeclareMathOperator{\corank}{corank}
\DeclareMathOperator{\Lep}{Lep}
\DeclareMathOperator{\id}{id}

\newcommand{\C}{\mathbb{C}}

\newcommand{\R}{{\mathbb R}}

\newcommand{\der}{\partial}

\newcommand{\cA}{\mathcal{A}}

\newcommand{\cC}{\mathcal{C}}

\newcommand{\cE}{\mathcal{E}}

\newcommand{\cI}{\mathcal{I}}

\newcommand{\cL}{\mathcal{L}}

\newcommand{\cP}{\mathcal{P}}

\newcommand{\cR}{\mathcal{R}}

\newcommand{\cV}{\mathcal{V}}

\newcommand{\bG}{\boldsymbol{G}}

\newcommand{\bM}{\boldsymbol{M}}

\newcommand{\bP}{\boldsymbol{P}}

\newcommand{\bU}{\boldsymbol{U}}

\newcommand{\bX}{\boldsymbol{X}}
\newcommand{\bY}{\boldsymbol{Y}}
\newcommand{\bZ}{\boldsymbol{Z}}

\newcommand{\wed}{\wedge}

\newcommand{\alp}{\alpha}
\newcommand{\bet}{\beta}
\newcommand{\gam}{\gamma}
\newcommand{\del}{\delta}
\newcommand{\eps}{\epsilon}
\newcommand{\zet}{\zeta}
\newcommand{\tht}{\theta}
\newcommand{\iot}{\iota}

\newcommand{\lam}{\lambda}
\newcommand{\sig}{\sigma}

\newcommand{\ome}{\omega}

\newcommand{\For}{{\Lambda}}

\newcommand{\Thd}{{\Theta}}

\numberwithin{equation}{section}

\newtheorem{Theorem}{Theorem}[section]
\newtheorem{Corollary}[Theorem]{Corollary}
\newtheorem{Lemma}[Theorem]{Lemma}
\newtheorem{Proposition}[Theorem]{Proposition}
 { \theoremstyle{definition}
\newtheorem{Definition}[Theorem]{Definition}
\newtheorem{Example}[Theorem]{Example}
 }

\begin{document}

\allowdisplaybreaks

\renewcommand{\thefootnote}{$\star$}

\newcommand{\arXivNumber}{1508.01752}

\renewcommand{\PaperNumber}{045}

\FirstPageHeading

\ShortArticleName{Variational Sequences, Representation Sequences and Applications in Physics}

\ArticleName{Variational Sequences, Representation Sequences\\ and Applications in Physics\footnote{This paper is a~contribution to the Special Issue on Analytical Mechanics and Dif\/ferential Geometry in honour of Sergio Benenti.
The full collection is available at \href{http://www.emis.de/journals/SIGMA/Benenti.html}{http://www.emis.de/journals/SIGMA/Benenti.html}}}

\Author{Marcella PALESE~$^{\dag^1}$, Olga ROSSI~$^{\dag^2\dag^3}$, Ekkehart WINTERROTH~$^{\dag^1\dag^4}$ and Jana MUSILOV\'A~$^{\dag^5}$}

\AuthorNameForHeading{M.~Palese, O.~Rossi, E.~Winterroth and J.~Musilov\'a}

\Address{$^{\dag^1}$~Department of Mathematics, University of Torino, via C. Alberto 10, 10123 Torino, Italy}
\EmailDD{\href{mailto:marcella.palese@unito.it}{marcella.palese@unito.it},
 \href{mailto:ekkehart.winterroth@unito.it}{ekkehart.winterroth@unito.it}}

\Address{$^{\dag^2}$~Department of Mathematics, Faculty of Science, University of Ostrava,\\
\hphantom{$^{\dag^2}$}~Ostrava, Czech Republic}
\EmailDD{\href{mailto:olga.rossi@osu.cz}{olga.rossi@osu.cz}}

\Address{$^{\dag^3}$~Department of Mathematics and Statistics, La Trobe University, Melbourne, Australia}

\Address{$^{\dag^4}$~Lepage Research Institute, 783 42 Slatinice, Czech Republic}

\Address{$^{\dag^5}$~Institute of Theoretical Physics and Astrophysics, Masaryk University Brno, Czech Republic}
\EmailDD{\href{mailto:janam@physics.muni.cz}{janam@physics.muni.cz}}

\ArticleDates{Received September 22, 2015, in f\/inal form April 26, 2016; Published online May 02, 2016}

\Abstract{This paper is a review containing new original results on the f\/inite order variational sequence and its dif\/ferent representations with emphasis on applications in the theory of variational symmetries and conservation laws in physics.}


\Classification{55N30; 55R10; 58A12; 58A20; 58E30; 70S10}

\renewcommand{\thefootnote}{\arabic{footnote}}
\setcounter{footnote}{0}

\section{Introduction}\label{intro}

\looseness=1
After \'Elie Cartan, who in 1922 introduced dif\/ferential geometric methods into the calculus of variations using the celebrated Cartan form to formulate the variational principle \cite{Cart}, Th\'eophile Lepage in his pioneer work \cite{Lep36-I,Lep36-II} initiated the study of the calculus of variations for f\/ield theo\-ries (multiple integrals) as a theory of dif\/ferential forms and their exterior dif\/ferential modulo contact forms (congruences). The main advantage of such an approach evidently is a possi\-bi\-li\-ty of framing the calculus of variations within the de Rham sequence of dif\/ferential forms on a~mani\-fold and thus using the de Rham cohomology as an ef\/fective tool to study relations between the image and kernel of various dif\/ferential operators appearing in the calculus of variations.

The attempt to understand global structures of the calculus variations motivated the use of $C$-spectral sequences, and gave rise to the constructions of variational sequences and bicomplexes.
This direction was initiated in 1950's by Dedecker~\cite{Ded53,Ded55,Ded56,Ded57,Ded}, and continued in the work of Horndeski~\cite{Hor}, Gelfand and Dikii~\cite{GeDi}, Olver and Shakiban~\cite{OSh}, Dedecker and Tulczyjew~\cite{DT80}, Tulczyjew~\cite{Tul77,Tu80,Tu88}, Takens~\cite{Tak79}, Anderson~\cite{An}, Anderson and Duchamp~\cite{AnDu80}, Dickey~\cite{Dic}, Vinogradov~\cite{Vin77,Vin84-I,Vin84-II}, and others. A main motivation for this research was
\begin{itemize}\itemsep=0pt
\item \looseness=1 to f\/ind solution of {\em the inverse problem of the calculus of variations} which means to f\/ind local and global constraints for a system of ordinary or partial dif\/ferential equations to come from a variational principle as Euler--Lagrange equations, and to construct a Lag\-rangian,

\item to determine the local and global structure of {\em variationally trivial Lagrangians}, i.e., Lag\-rangians giving rise to identically zero Euler--Lagrange expressions (the kernel of the Euler--Lagrange operator),

\item to understand the form and origin of dif\/ferent ``canonically looking'' expressions appearing in the calculus of variations, as for example
\begin{gather*}
\frac{\partial L}{\partial y^\sigma} - d_j \frac{\partial L}{\partial y^\sigma_j} + d_jd_k \frac{\partial L}{\partial y^\sigma_{jk}} - \cdots ,
\qquad
\frac{\partial L}{\partial y^\sigma_p} - d_j \frac{\partial L}{\partial y^\sigma_{jp}} + d_jd_k \frac{\partial L}{\partial y^\sigma_{jkp}} - \cdots,
\\
\frac{\partial E_\sigma}{\partial q^\nu_j} -
\sum_{k=j}^{s} (-1)^k \binom{k}{j} \frac{d^{k-j}}{dt^{k-j}}\frac{\partial E_\nu}
{\partial q^\sigma_k},
\qquad
\sum_{l=0}^{s-j-1} (-1)^{j+l}\binom{j+l}{l}
 \frac{d^l}{dt^l} \frac{\partial E_\sigma}{\partial q^\nu_{j+l+1}},
\end{gather*}

\item to understand local and global phenomena typical for variational problems given by local data, as, for instance examples of global Euler--Lagrange forms which do not admit a~global Lagrangian; the best known examples are the Galilean relativistic mecha\-nics, or Chern--Simons f\/ield theories (where a~global Lagrangian exists but is not gauge-invariant).
\end{itemize}

The theory was f\/irst formulated on {\em infinite} order jets. After several attempts to obtain an alternative construction using f\/inite order structures (see, e.g.,~\cite{AnDu80}), the {\em finite order variational sequence} was discovered by Krupka 25~years ago~\cite{Kru90}. The inspiration originated from the Lepage's idea of a~`congruence', indicating that there should be a {\em close relationship between the Euler--Lagrange operator and the exterior derivative of differential forms}, and resulted in the concept of a f\/inite order (cohomological) exact sequence of forms, in which the Euler--Lagrange mapping is included as one arrow.

Compared with the other approaches, the concept of Krupka's variational sequence is conceptually very simple due to the following reasons:
\begin{itemize}\itemsep=0pt
\item it uses only f\/inite order jets,
\item it appears as a {\em quotient sequence of the de Rham sequence}.
\end{itemize}

This makes it quite easily accessible to potential users, since factorization is a natural, straightforward, commonly understood, and well-mastered technique. Moreover, the facto\-ri\-za\-tion procedure can be incorporated into intrinsic operators which allow to avoid tedious and technically dif\/f\/icult calculations.

It should be stressed, however, that here is an additional benef\/it of the setting on f\/inite order jets: Namely, it provides precise results on the {\em order} of the local and global objects of interest.

The aim of this paper is to provide the reader a comprehensive source of the current theory of the f\/inite order variational sequences as it stands 25 years after its beginnings. While for
the variational bicomplex theory there is a standard reference the (unpublished) book by Anderson~\cite{An}, the results on the f\/inite order variational sequence up to now have not been collected, revised and completed in a self-contained article. However, our aim is not merely to provide a~revised exposition of known results. Apart from putting known results into a unif\/ied treatment, we add a~number of original results in order to develop the theory and make it complete. Our aim is also to stimulate the use and applications of the variational sequence in physics. As a~motivation we mention a few applications, some of which are interesting and not widely known or even are original here.

\looseness=1
The paper is organized as follows:
Section~\ref{section2} introduces the variational sequence as a quotient sequence of the de Rham sequence, developed by Krupka~\cite{Kru90,Kru97,Kru99,Kru04} (see also \mbox{\cite{Gr99,Gr08, Kru08,Kru15}}).
Section~\ref{section3} is the core of the paper, and deals with the problem of representation of the va\-riatio\-nal sequence by dif\/ferential forms.
In the variational sequence, the objects are classes of local dif\/ferential forms. However (as pointed out already by Krupka) the classes can be represented by (global) dif\/ferential forms: Lagrangians or Lepage $n$-forms (where $n$ is the number of independent variables in the variational problem), Euler--Lagrange forms, Helmholtz forms. This property initiated the search of representation operators and representations with help of dif\/ferential forms. We present two dif\/ferent representation sequences: the {\em Takens representation sequence} based on the so-called {\em interior Euler operator}~\cite{An}
which assigns to every class in the $q$-th column of the variational sequence a source form of degree~$q$ (a Lagrangian if $q=n$, a~dynamical form if $q=n+1$, a Helmholtz-like form if $q = n+2$), and the {\em Lepage representation sequence}. The latter is a sequence of Lepage forms (Lepage equivalents of source forms), such that the morphisms are just the {\em exterior derivatives}. The existence of such a~sequence clarif\/ies the relationship of the variational morphisms to the exterior derivative ope\-ra\-tor. Moreover, it transfers the questions on the structure of the kernel and image of all the variational morphisms just to application of the Poincar\'e lemma. A~complete representation of the
variational sequence by source forms (of any degree) was given by Musilov\'a and Krbek \cite{Kr02,KrMu03,KrMu05} and independently by Krupka and his collaborators \cite{KrSe05,KUV13,VoUr14}. Next, in Section~\ref{section3} we recall the concept of Lepage form (\cite{Kru73,KrSe05,olga86,KS09} and introduce the Lepage representation of the va\-ria\-tio\-nal sequence. As a~new result concerning Lepage $n$-forms, we present an intrinsic formula for the Cartan form (Theo\-rem~\ref{thmCartanform}, formula~(\ref{Cartan})). Remarkably, our formula explicitly shows that (and explains why) the higher-order Cartan form generically is not global (this fact was a surprise when discovered and demonstrated by a counterexample in 1980's~\cite{HoKo}, and up to now has not been well understood). In the subsection on Lepage forms {\em of higher degrees} (and any order) we present for the f\/irst time properties of general Lepage forms, and namely we give a general intrinsic formula for a Lepage $(n+k)$-form of any order. We also f\/ind a~higher degree generalization of the famous Cartan form. The techniques based on the use of the interior Euler operator and the residual operator not only lead to elegant formulas and easy proofs; they also enable to avoid lengthy, tedious, and sometimes even impossible calculations in coordinates. The results then help us to f\/ind a~complete correct Lepage representation of the variational sequence. The last remark in Section~\ref{section3} concerns the variational order of objects in the variational sequence. We explain how due to the f\/inite order setting it is possible to make (local and global) conclusions on a~precise order of Lagrangians and other objects in the variational sequence, distinguishing it from alternative constructions on inf\/inite order jets.

In Section~\ref{section4} we are concerned with the Lie derivative in the variational sequence. We explore the important property that the Lie derivative of dif\/ferential forms with respect to prolongations of projectable vector f\/ields preserves the contact structure. This suggests that a Lie derivative of classes of forms, a {\em variational Lie derivative}, can be correctly def\/ined as the equivalence class of the standard Lie derivative of forms and represented by forms. New results in this section are then contained in subsequent theorems generalizing the f\/irst variation formula to forms of any degree.

\looseness=1 The last section of the paper is devoted to selected applications of the variational sequence theory. We recall the fundamental well-known applications to the inverse variational problem and to variationally trivial Lagrangians. Apart from that we emphasize some other possible applications which have not yet been fully explored~-- like, for instance, properties of Helmholtz forms. We include also examples of variational problems def\/ined by local data, and on symmetries and conservation laws, emphasizing cohomological questions and local and global phe\-nomena.

For a reader interested in a wider context of the theory of variational sequences, we refer to the survey works~\cite{Kru08,Vi08}, and to the book~\cite{Kru15}.

\section{The variational sequence}\label{section2}

\subsection{Jet bundles and contact forms}

Fibred manifolds and their jet prolongations represent a convenient mathematical framework for mechanics and f\/ield theories, suitable for investigation of Lagrangian systems of dif\/ferent orders, and one or many independent variables within a suf\/f\/iciently general and unif\/ied geometric framework. For an excellent exposition of the jet bundle geometry we refer to the book~\cite{Sa89}.

In what follows, we shall consider a smooth f\/ibred manifold $\pi\colon \bY \to \bX$, with $\dim \bX = n$ and $\dim \bY = n+m$, and its $r$-jet prolongations $\pi_r\colon J^r \bY \to \bX$, where $r \geq 1$. For $n =1$ the mani\-fold~$\bY$ is a space of events for mechanical systems of~$m$ degrees of freedom, and local sections of~$\pi$ are graphs of curves, so usually $\bX = \R$ (meaning time, or just a parameter for the curves). If $n > 1$ then $\bX$ has the meaning of a physical space or spacetime, and local sections of~$\pi$ are physical f\/ields over the mani\-fold~$\bX$. Then $J^r\bY$ is a manifold of points $j^r_x\gamma$, the equivalence classes of local sections $\gamma$ of $\pi$ with the same value at~$x$ and the same partial derivatives at~$x$ up to the order~$r$. Thus $J^r\bY$ serves as a domain for functions, dif\/ferential forms, vector f\/ields, and other objects depending on higher derivatives (up to the $r$-th order).

Due to the af\/f\/ine bundle structure of $\pi_{r+1,r} \colon J^{r+1} \bY \to J^{r} \bY$, we have a natural splitting $ J^r\bY\times_{J^{r-1}\bY}T^*J^{r-1}\bY = J^r\bY\times_{J^{r-1}\bY}(T^*\bX \oplus V^*J^{r-1}\bY)$ which induces natural splittings into horizontal and vertical parts of projectable vector f\/ields, forms, and of the exterior dif\/ferential on $J^r\bY$.

A dif\/ferential $q$-form $\omega$ on $J^r \bY$ is called {\em contact} if $J^r\gamma^*\omega = 0$ for every local section~$\gamma$ of~$\pi$.
Contact forms on $J^r \bY$ form an ideal in the exterior algebra, called the {\em contact ideal} of order~$r$. The contact ideal is generated by contact $1$-forms and their exterior derivatives~\cite{Kru95}.
A~dif\/ferential $q$-form $\omega$ on $J^r \bY$ is called {\em horizontal}, or {\em $0$-contact}, if $\xi\rfloor \omega = 0$ for every vector f\/ield~$\xi$ on~$J^r \bY$, vertical with respect to the projection onto~$\bX$.

From the def\/initions it follows that every $q$-form for $q > n$ is contact, and that every form~$\omega$ on $J^r \bY$, if lifted to $J^{r+1} \bY$, is a sum of a unique horizontal and contact form. This splitting can be further ref\/ined if the concept of a {\em $k$-contact form} for $k \geq 1$ is introduced. A~contact form~$\omega$ is called $k$-contact if for every vertical~$\xi$ the contraction $\xi\rfloor \omega$ is $(k-1)$-contact. Then, as proved in~\cite{Kru83}, every $q$-form $\omega$ on $J^r \bY$, $q \geq 1$, admits the unique decomposition
\begin{gather}\label{splitting of forms}
 \quad \pi^*_{r+1,r} \omega = \sum_{k = 0}^q p_k \omega
\end{gather}
into a sum of $k$-contact forms, $0 \leq k \leq q$. For the horizontal ($0$-contact) operator often the notation $h = p_0$ is used. The form $p_k \omega$ is called the {\em $k$-contact component of $\omega$}.
Note that if $q \geq n+1$ then $\omega$ is contact, and it is {\em at least $(q-n)$-contact}, i.e., the contact components $p_0, p_1, \dots, p_{q-n-1}$ of $\omega$ are equal to zero. If, moreover, $p_{q-n}\omega = 0$, we speak about a {\em strongly contact} form~\cite{Kru90}.

An important special case of the decomposition \eqref{splitting of forms} concerns closed forms and their local primitives, obtained due to the Poincar\'e lemma, by means of the Poincar\'e homotopy opera\-tor~$\cP$. Since every form~$\omega$ can be locally expressed as $\omega = \cP d\omega + d\cP \omega$, equation $\omega = d\rho$ has a local solution $\rho = \cP \omega$. Applying, however, the decomposition into contact components we can see that there arises a modif\/ied homotopy operator, denoted by $\cA$, and called the {\em contact homotopy operator}~\cite{Kru83}, which is adapted to the decomposition \eqref{splitting of forms} and def\/ined by
$\cA p_0 \omega = 0$, $\cA p_k \omega = p_{k-1} \cP \omega$, $ k\geq 1$.
It satisf\/ies $\pi_{r+1,r}^*\omega = \cA d(\pi_{r+1,r}^*\omega) + d \cA(\pi_{r+1,r}^*\omega)$, and moreover, if (locally)
$\omega = d\rho$, then
\begin{gather*}%\label{splitting homotopy}
\quad \pi^*_{r+1,r} \rho = \cA(\pi_{r+1,r}^*\omega) = \sum_{k = 0}^q \cA p_{k+1} \omega,
\end{gather*}
hence
$p_k \rho = \cA p_{k+1} \omega$.
Compared to $\cP$, the operator $\cA$ concerns vertical curves (curves in the f\/ibres over~$\bX$) only.

Notice that the decomposition \eqref{splitting of forms} applied to $\omega = d\rho$ induces a splitting of the exterior derivative $d$ to $\pi_{r+1, r}^* d = d_H + d_V$ where the {\em horizontal derivative} $d_H = hd$, and the {\em vertical derivative} $d_V$ is the contact part of $\pi_{r+1, r}^* d$.

The splitting of vector f\/ields then arises with help of the total derivative operator. If we denote $(x^i,y^\sigma)$ local f\/ibred coordinates on $\bY$, and $(x^i,y^\sigma_J)$ where $J$ is a multiindex, $0 \leq |J| \leq r$,
the associated coordinates on $J^r\bY$, then the $i$-th total (formal) derivative takes the form
\begin{gather*}
d_i = \frac{\partial}{\partial x^i} + \sum_{|J|=0}^r y^\sigma_{Ji} \frac{\partial}{\partial y^\sigma_J}.
\end{gather*}
Considered as a vector f\/ield along the projection $\pi_{r+1,r}$, the total derivative induces a splitting of vector f\/ields to a horizontal and vertical part as follows: Letting
\begin{gather*}
\xi = \xi^i \frac{\partial}{\partial x^i} + \sum_{|J|=0}^r \xi^\sigma_{J} \frac{\partial}{\partial y^\sigma_J}
\end{gather*}
be a vector f\/ield on $J^r\bY$, in order to obtain the splitting $\xi = \xi_H + \xi_V$, we put
\begin{gather*}
\xi_H = \xi^i d_i , \qquad
\xi_V = \xi - \xi_H = \sum_{|J|=0}^r \big(\xi^\sigma_{J} - y^\sigma_{Ji}\xi^i\big) \frac{\partial}{\partial y^\sigma_J}.
\end{gather*}
Note that both $\xi_H$ and $\xi_V$ are vector f\/ields along the projection $\pi_{r+1,r}$ rather than `ordinary' vector f\/ields on $J^r\bY$. If $\Xi$ is a projectable vector f\/ield on $\bY$, we write $J^r \Xi_V = (J^r \Xi)_V$ (this can be viewed as a def\/inition of prolongation of $\Xi_V$ which is a vector f\/ield along the projec\-tion~$\pi_{1,0}$).

Recall that ${\cC}(\pi_r)$, the {\em contact distribution} or {\em Cartan distribution} of order $r$, is locally annihilated by contact $1$-forms
\begin{gather*}
\omega^\sigma_J = dy^\sigma_J - y^\sigma_{Jj} dx^j , \qquad 1 \leq \sigma \leq m ,
\end{gather*}
or, equivalently, locally generated by the vector f\/ields $d_i$ along $\pi_{r,r-1}$ and $\partial/\partial y^\sigma_J$ where $|J| = r$.
The Cartan distribution on $J^r \bY$ is not completely integrable, hence the ideal generated by contact $1$-forms of order $r$ is not closed. Its closure is then the {\em contact ideal} of order $r$ (the ideal of all contact forms on $J^r \bY$).

In what follows, we shall need on $J^r \bY$ the ideal generated by lifts of {\em contact $1$-forms of order~$1$}, whose elements take the form $\pi_{r,1}^*\omega \land \eta$ where $\omega \in {\cC}^0(\pi_1)$ (the annihilator of~$\cC (\pi_1)$) and~$\eta$ is of order~$r$. Since locally they are expressed as
$\omega^\sigma \land \eta_\sigma$, where
$\omega^\sigma = dy^\sigma - y^\sigma_j dx^j$, $1 \leq \sigma \leq m$,
are the basic local contact $1$-forms of order~$1$, we say that forms belonging to this ideal are {\em $\omega^\sigma$-generated}.

In the calculus of variations and the theory of dif\/ferential equations on f\/ibred manifolds, the fundamental role is played by the sheaf $\Lambda^r_{n,\bX}$ of horizontal ($0$-contact) $n$-forms on~$J^r \bY$, the elements of which are called {\em Lagrangians} of order $r$, and the sheaf $\Lambda^r_{n+1,1,\bY}$ of $1$-contact $(n+1)$-forms on $J^r \bY$, horizontal with respect to the projection onto~$\bY$, the elements of which are called {\em dynamical forms} of order $r$. By means of variation of the action def\/ined by a Lagran\-gian~$\lambda$ one obtains a distinguished dynamical form $E_\lambda$, called the {\em Euler--Lagrange form} of~$\lambda$~\cite{Kru73}; the components of $E_\lambda$ in every f\/ibred chart are the Euler--Lagrange expressions of the Lagrangian~$\lambda$. Remarkably, if $\lambda$ is of order $r$ then $E_\lambda$ is of order $\leq 2r$.

The mapping assigning to every Lagrangian its Euler--Lagrange form is called the {\em Euler--Lagrange mapping}.

A dynamical form $E$ is called {\em globally variational} if there exists a Lagrangian $\lambda$ such that (possibly up to a jet projection) $E = E_\lambda$. $E$ is called {\em locally variational} if every point in the domain of $E$ has a neighbourhood $U$ where $E$ is variational \cite{Kru83}.

Dynamical forms are often called {\em source forms}, as originally introduced by Takens
\cite{Tak79}. However, we shall follow Krupkov\'a and Prince \cite{KP07} who extended the concept of source form to forms of all degrees $q = n+ k$ where $k \geq 1$:

\begin{Definition} Let $k \geq 1$. By a source form of degree $n+k$ we shall mean a $\omega^\sigma$-generated $k$-contact $(n+k)$-form.
\end{Definition}

For $k=1$, dynamical forms are $\omega^\sigma$-generated $1$-contact $(n+1)$-forms, hence, indeed, source forms of degree $n+1$.
Source forms of degree $n+2$ are called {\em Helmholtz-like forms}~\cite{KrMa10}. For the calculus of variations the most important examples of source forms of degree $n+1$ are {\em Euler--Lagrange forms}, and of source forms of degree $n+2$ the {\em Helmholtz forms}.

\subsection{The variational sequence}

Now we recall the f\/inite order variational sequence as introduced by Krupka in~\cite{Kru90}. The concept uses the sheaf theory (standard references are, e.g., \cite{Bre67,We}). However, for a f\/irst understanding, a concrete use of the variational sequence, and calculations it is not necessary to be familiar with all the rather dif\/f\/icult mathematical theory. In place of a `sheaf' of forms in this context one can roughly imagine a family of modules of local dif\/ferential forms def\/ined on the open subsets of the corresponding manifold.

In what follows, we denote by $ \R_{\bY}$ the constant sheaf over $\R$, and by $\Lambda^r_q$ the sheaf of $q$-forms on $J^r\bY$.
We let $\Lambda^r_{0,c} = \{0\}$, and $\Lambda^r_{q,c}$ be the sheaf of contact $q$-forms, if $q \leq n$, or the sheaf of strongly contact $q$-forms, if $q >n$, on $J^r \bY$. We set
\begin{gather*}
\Theta^r_q = \Lambda^r_{q,c} + \big(d \Lambda^r_{q-1,c}\big) ,
\end{gather*}
where $d \Lambda^r_{q-1,c}$ is the image of $\Lambda^r_{q-1,c}$ by the exterior derivative $d$, and
$(d \Lambda^r_{q-1,c})$ denotes the sheaf generated by the presheaf $d \Lambda^r_{q-1,c}$. We note that
$\Theta^r_q = 0$ for $q > \corank {\cal C}(\pi_r)$.

The subsequence
\begin{gather*}
0 \to \Theta^r_1 \to \Theta^r_2 \to \Theta^r_3 \to \cdots
\end{gather*}
of the de Rham sequence
\begin{gather*}
0 \to \R_{\bY} \to \Lambda^r_0 \to \Lambda^r_1 \to \Lambda^r_2 \to \Lambda^r_3 \to \cdots
\end{gather*}
is an exact sequence of soft sheaves. The quotient sequence
\begin{gather} \label{vseq}
0 \to \R_{\bY} \to \Lambda^r_0 \to \Lambda^r_1 / \Theta^r_1 \to \Lambda^r_2 / \Theta^r_2 \to \Lambda^r_3 / \Theta^r_3 \to \cdots
\end{gather}
is called the {\em variational sequence} of order $r$. The main theorem due to Krupka~\cite{Kru90} then states that {\em the variational sequence is an acyclic resolution of the constant sheaf $\R_{\bY}$} (meaning that the sequence is locally exact with the exception of~$\R_{\bY}$). Hence, due to the abstract de Rham theorem, {\em the cohomology groups of the cochain complex of global sections of the variational sequence are identified with the de Rham cohomology groups $H^{q}_{\rm dR}\bY$ of the manifold~$\bY$}.

It should be stressed that objects in the variational sequence are elements of the quotient sheaves $\Lambda^r_q / \Theta^r_q$, i.e., they are {\em equivalence classes of local $r$-th order differential $q$-forms}. We denote by $[\rho] \in \Lambda^r_q / \Theta^r_q$ the class of $\rho \in \Lambda^r_q$. Morphisms in the variational sequence are then by construction {\em quotients of the exterior derivative $d$}. We denote
\begin{gather*}
{\cal E}_q\colon \ \Lambda^r_q / \Theta^r_q \to \Lambda^r_{q+1} / \Theta^r_{q+1},
\end{gather*}
so that
\begin{gather*}
{\cal E}_q ([\rho]) = [d\rho],
\end{gather*}
and ${\cal E}_q ([\rho]) = [d\rho] = 0$ means that there exists $[\eta] \in \Lambda^r_{q-1} / \Theta^r_{q-1}$ such that $[\rho] = {\cal E}_{q-1} ([\eta]) = [d\eta]$. If, moreover, $H^q_{\rm dR} \bY = \{0\}$ then the class $[\eta]$ has a {\em global} representative.

The meaning of this abstract construction for the calculus of variations becomes clear in view of the following theorem:

\begin{Theorem}[Krupka \cite{Kru90,Kru08}]\label{IsoSh}
The sheaf $\Lambda^r_n/ \Theta^r_n$ is isomorphic with a subsheaf of the sheaf of Lagran\-gians~$\Lambda^{r+1}_{n,\bX}$, the sheaf $\Lambda^r_{n+1}/ \Theta^r_{n+1}$ is isomorphic with a subsheaf of the sheaf of source forms $\Lambda^{2r+1}_{n+1,1,\bY}$, and the quotient mapping ${\cal E}_n\colon \Lambda^r_n / \Theta^r_n \to \Lambda^r_{n+1} / \Theta^r_{n+1}$ is the Euler--Lagrange mapping.
\end{Theorem}

By this theorem, the elements of the quotient sheaf $\Lambda^r_n / \Theta^r_n$ have the meaning of Lagrangians. The kernel of ${\cal E}_n$ then consists of null-Lagrangians (variationally trivial Lagrangians), and its image, which, by exactness is the kernel of the next morphism,
\begin{gather*}
{\cal E}_{n+1}\colon \ \Lambda^r_{n+1} / \Theta^r_{n+1} \to \Lambda^r_{n+2} / \Theta^r_{n+2},
\end{gather*}
represents {\em variational equations}. Therefore ${\cal E}_{n+1}$ is called {\em Helmholtz mapping}. The name refers to the famous ``Helmholtz conditions''~\cite{Hel}, necessary and suf\/f\/icient conditions for dif\/ferential equations to be variational (as they stand).
Remarkably, the Helmholtz morphism, discovered within the variational sequence theory, has not been known in the classical calculus of variations.

\looseness=-1
The above theorem also suggests the idea that all classes in the variational sequence could be representable by {\em differential forms} with a clear place in the calculus of variations. So far there have been noticed two such representation sequences: the ``{\em source form representation}'' which we call {\em Takens representation}, and the ``{\em Lepage form representation}''. In the former the sheaves contain familiar dif\/ferential forms which appear in the calculus of variations, like Lag\-rangians and Euler--Lagrange forms, and their generalizations, like, for example, Helmholtz forms. The morphisms are then the corresponding variational mappings, like, e.g., the Euler--Lagrange mapping sending Lagrangians to Euler--Lagrange forms, or the Helmholtz mapping, sending dynamical forms to Helmholtz forms. The latter representation has as objects Lepage forms, which are extensions of the ``familiar forms'' mentioned above by higher contact components in such a~way that the {\em morphisms are the exterior derivatives},
so that in the Lepage form representation the variational sequence ``morally'' (up to the orders of individual subsheaves) becomes a~subsequence of the de Rham sequence. We shall deal with the representation sequences below.

In what follows, we shall denote
\begin{gather*}
{\cal V}^r_0 = \Lambda^r_0 , \qquad
{\cal V}^r_q = \Lambda^r_q / \Theta^r_q,
\end{gather*}
and the variational sequence of order $r$ will be shortly written in the form
\begin{gather*}
0 \to \R_{\bY} \to {\cal V}^r_* .
\end{gather*}

\section{Representation sequences}\label{section3}

 \subsection{The interior Euler operator}

Let us introduce an operator which plays an essential role in the representation theory for the variational sequence. It was introduced to the calculus of variations within the variational bicomplex theory in \cite{An, Bau},
and adapted to the f\/inite order situation of the variational sequence in \cite{Kr02,KrMu03,KrMu05,KrSe05}.
This operator, called {\em interior Euler operator}, and denoted by~$\cI$, ref\/lects in an intrinsic way the procedure of getting a distinguished representative of a class $[\rho] \in \Lambda^r_{q}/\Theta^r_{q}$ for $q > n$ by applying to $\rho$ the operator $p_{q-n}$ and a factorization by
$\Theta^r_{q}$.

Let $(x^i,y^\sigma)$ denote local f\/ibred coordinates on $\bY$, and $(x^i,y^\sigma_J)$ the associated coordinates on~$J^r\bY$, where $J$ is a multiindex, $0 \leq |J| \leq r$, $J = (j_1, \dots, j_p)$ with $1 \leq j_1 \leq j_2 \leq \cdots \leq j_p \leq n$ and $p \leq r$.

Let $k \geq 1$ and $q = n+k$. Recall that if $\rho \in \Lambda^r_{n+k}$, we have
\begin{gather*}
\pi_{r+1,r}^* \rho = p_k \rho + p_{k+1} \rho + \cdots + p_{k+n}\rho .
\end{gather*}
Following \cite{KrMu05} we set
 \begin{gather}\label{inter}
 \cI(\rho)= \frac{1}{k}\ome^\sigma \wed \sum_{|J|=0}^{r}(-1)^{|J |} d_J \big(\der/\der y^\sigma_{J} \rfloor p_k\rho \big),
\end{gather}
where $d_J = d_{j_1} \dots d_{j_p}$ for $|J| = p$.
It can be shown by means of the partition of unity arguments that this formula def\/ines a global form~$\cI (\rho)$.

The operator $\cI\colon \Lambda^r_{n+k} \to \Lambda^{2r+1}_{n+k}$ has the following properties:
\begin{enumerate}\itemsep=0pt
\item[(1)] for every $\rho \in \Lambda^r_{n+k}$, $\cI(\rho)$ is a source form of degree $n+k$,
\item[(2)] $\cI(p_k\rho) = \cI(\rho)$,
\item[(3)] $\cI (\rho)$ belongs to the same class as $\pi_{2r+1,r}^* \rho$, i.e., $\pi_{{2r+1},{r}}^*\rho -\cI(\rho)\in \Thd_{n+k}^{2r+1}$,
\item[(4)] $\cI^2 = \cI$, up to a canonical projection; precisely, $\cI^2(\rho) = \pi_{{4r+3},{2r+1}}^*\cI(\rho)$,
\item[(5)] $\Ker \cI = \Thd_{n+k}^{r}$.
\end{enumerate}

By construction, $\cI(\rho)$ is a $k$-contact form, so there is a question about the dif\/ference bet\-ween~$\cI(\rho)$ and~$p_k\rho$. It turns out that the dif\/ference can be expressed intrinsically by means of an operator $\cR$ which we shall call {\em residual operator}, as follows
\begin{gather}\label{pir}
p_k\rho=\cI(\rho)+p_kdp_k\cR(\rho)
\end{gather}
(up to appropriate canonical projections).
By the properties (2) and (4) of $\cI$ we can see that
\begin{gather}\label{pdr}
\cI(p_kdp_k\cR(\rho))=0 .
\end{gather}
 ${\cal{R}}(\rho)$ is by def\/inition a local strongly contact $(n+k-1)$-form.
Note that, contrary to ${\cal{I}}(\rho)$, the form ${\cal{R}}(\rho)$ need not be unique and need not be global.

Moreover, $\cI(\rho)$ is a source form, so there is a question about its relationship with other source forms in the same class~$[\rho]$. To clarify this it is useful to introduce the concept of a {\em canonical source form} as a source form $\rho$ such that (up to projection) $\rho = \cI(\rho)$~\cite{KP07}.

\begin{Proposition} \label{propsource}
Every source form is equivalent with a canonical source form $($of possibly higher order$)$. More precisely, letting $\rho = \omega^\sigma \land \eta_\sigma$ be a source $(n+k)$-form it holds
$($up to projection$)$
\begin{gather*}
\rho = k \cI(\rho) - (k-1) \omega^\sigma \land \cI(\eta_\sigma)
\end{gather*}
and
\begin{gather*}
\rho - \cI(\rho) = p_k d p_k \cR (\rho) = \frac{k-1}{k} \omega^\sigma \land p_{k-1} d p_{k-1} \cR(\eta_\sigma) .
\end{gather*}
As a consequence, for $k \geq 2$,
\begin{gather*}
\cI(\omega^\sigma \land \eta_\sigma) = \cI( \omega^\sigma \land \cI(\eta_\sigma)) .
\end{gather*}
\end{Proposition}

\begin{proof}
The f\/irst assertion is just the property (3) of $\cI$. Next two formulas easily follow from the def\/initions of $\cI$ and $\cR$ and the fact
that $\rho$ is the source form, i.e., $\pi_{r+1,r}^* \rho = p_k\rho$. Indeed, for $\rho = \omega^\sigma \land \eta_\sigma$ we have
\begin{gather*}
 \cI(\omega^\sigma \land \eta_\sigma) = \frac{1}{k} \omega^\sigma \land
\sum_{|J|=0}^{r}(-1)^{|J |} d_J (\der/\der y^\sigma_{J} \rfloor (\omega^\nu \land \eta_\nu)) \\
\hphantom{\cI(\omega^\sigma \land \eta_\sigma)}{} = \frac{1}{k} \omega^\sigma\! \land \! \bigg( \eta_\sigma - \omega^\nu
\land \!\sum_{|J|=0}^{r}\! (-1)^{|J |} d_J (\der/\der y^\sigma_{J} \rfloor \eta_\nu)
\bigg) =
\frac{1}{k} \omega^\sigma \land \eta_\sigma + \frac{k\!-\!1}{k} \omega^\sigma \land \cI(\eta_\sigma) .
\end{gather*}
This formula then yields
\begin{gather*}
p_k d p_k \cR (\rho) = \rho - \cI(\rho) = \frac{k-1}{k} \omega^\sigma \land \bigl( \eta_\sigma - \cI(\eta_\sigma)
\bigr) = \frac{k-1}{k} \omega^\sigma \land p_{k-1} d p_{k-1} \cR(\eta_\sigma) ,
\end{gather*}
as desired.

Finally, applying $\cI$ to the f\/irst formula we get
$(k-1) \cI(\rho) = (k-1) \cI(\omega^\sigma \land \cI(\eta_\sigma))$, which provides the last formula for $k\geq 2$.
\end{proof}

For {\em dynamical forms} Proposition~\ref{propsource} gives the following important uniqueness result~\cite{KP07}.

\begin{Corollary}
For $k=1$, every source form is a canonical source form.
\end{Corollary}

Note that this can be seen also directly from the def\/inition of $\cI$: Indeed, if $k=1$ and $\rho$ is a~$\omega^\sigma$-generated $1$-contact $(n+1)$-form then in \eqref{inter} the contractions give horizontal $n$-forms with the same components as those of~$\rho$, and the summation over~$|J|$ reduces to one term, $|J|=0$. Hence,
$\cI(\rho) = \omega^\sigma \land (\der/\der y^\sigma \rfloor \rho) = \rho$.

For the description of the image of $\cI$ we refer also to~\cite{KUV13}.

\begin{Example}
We shall show the computations of $\cI(\rho)$ and $\cR(\rho)$ for a~second-order
1-contact 2-form $\rho$ in mechanics ($n=1$, $k=1$, $r=2$).

Consider local f\/ibred coordinates $(t, q^\sigma)$
on $\bY$ and the corresponding associated coordinates
$(t,q^\sigma, q^\sigma_j)$, $1\leq j\leq r$, $1\leq\sigma\leq m$, on
$J^r\bY$, and denote $d_t$ the total derivative operator and $\omega^\sigma_j = d{q}_j^\sigma-{q}_{j+1}^\sigma dt$ the local basic contact $1$-forms, and put $\omega^\sigma = \omega^\sigma_0$. The calculation of the image of $\rho$ by the interior Euler operator is direct, using~(\ref{inter}).
We have
\begin{gather*}
p_1\rho=A^0_\sigma\omega^\sigma\wedge dt+A^1_\sigma{\omega}^\sigma_1\wedge dt+
A^2_\sigma\omega^\sigma_2\wedge dt,
\end{gather*}
hence
\begin{gather*}
{\cal{I}}(\rho) =\omega^\sigma\wedge\sum_{j=0}^{r=2}(-1)^jd_t^j\left(
\frac{\partial}{\partial q^\sigma_j}\rfloor p_1\rho\right)
=\big(A^0_\sigma-d_tA^1_\sigma+d_t^2A^2_\sigma\big)\omega^\sigma\wedge dt.
\end{gather*}

The residual term is then
\begin{gather*}
p_1dp_1{\cal{R}}(\rho)=\pi^{*}_{5,3} p_1\rho-{\cal{I}}(\rho)=
\big(d_tA_\sigma^1-d_t^2A^2_\sigma\big)\omega^\sigma\wedge
dt+A^1_\sigma\omega^\sigma_1\wedge
dt+A_\sigma^2\omega^\sigma_2\wedge dt .
\end{gather*}

Let us compute $\cR(\rho)$. Since $\cR(\rho)$ is a contact $1$-form of order two, it holds $\cR(\rho) = C_\sigma^0 \omega^\sigma + C_\sigma^1 \omega^\sigma_1 + C_\sigma^2 \omega^\sigma_2 = p_1 \cR(\rho)$. The $1$-contact part of the exterior derivative is
$p_1 d\cR(\rho) = - d_t C_\sigma^0 \omega^\sigma \land dt - (d_t C_\sigma^1 + C^0_\sigma) \omega^\sigma_1 \land dt - (d_t C_\sigma^2 + C^1_\sigma) \omega^\sigma_2 \land dt - C_\sigma^2 \omega^\sigma_3 \land dt$.
Comparing this with the formula above yields
\begin{gather*}
C^2_\sigma = 0, \qquad
d_t C_\sigma^2 + C^1_\sigma = - A^2_\sigma, \qquad
d_t C_\sigma^1 + C^0_\sigma = - A^1_\sigma, \qquad
- d_t C_\sigma^0 = d_tA_\sigma^1-d_t^2A^2_\sigma.
\end{gather*}
From here we immediately see that $C^1_\sigma = - A^2_\sigma$ and $C^0_\sigma = - A^1_\sigma + d_t A_\sigma^2$, and that the equation $- d_t C_\sigma^0 = d_tA_\sigma^1-d_t^2A^2_\sigma$ is satisf\/ied identically. So f\/inally,
\begin{gather*}
\cR(\rho) = \big(d_t A_\sigma^2- A^1_\sigma\big) \omega^\sigma - A_\sigma^2 \omega^\sigma_1.
\end{gather*}
We can easily verify that
$\pi^*_{5,3} p_1\rho={\cal{I}}(\rho)+p_1dp_1{\cal{R}}(\rho)$.
\end{Example}

\subsection{Takens representation of the variational sequence}

Once we have the interior Euler operator, we can obtain a sequence of sheaves of {\em differential forms} (rather than of classes of dif\/ferential forms), such that both the objects and the morphisms have a straightforward interpretation in the calculus of variations. Such a representation of the variational sequence, introduced by Krbek and Musilov\'a (see~\cite{Kr02,KrMu03,KrMu05})
basically concerns {\em source forms $($of all degrees$)$}, therefore we call it in honour of Takens the {\em Takens representation}.

The motivation for the representation comes from the following observations:

$\bullet$ If $q \leq n$ then for every $q$-form $\rho$ of order $r$, the contact decomposition \eqref{splitting of forms} takes the form
\begin{gather*}
\pi_{r+1,r}^*\rho = h\rho + p_1\rho + \dots + p_q \rho,
\end{gather*}
and $\Theta^r_q = \Lambda^r_{q,c} + (d \Lambda^r_{q-1,c})$, where we have sheaves of contact forms. This means that given a~class $[\rho] \in \cV^r_q = \Lambda^r_q/\Theta^r_q$, one has
$\rho_1 \sim \rho_2$ if\/f $\rho_1 - \rho_2 = \vartheta + d\eta$ where $\vartheta$ is a contact $q$-form and~$\eta$ is a contact $(q-1)$-form, i.e.,
\begin{gather*}
h\rho_1 = h\rho_2.
\end{gather*}
Hence, every class $[\rho]$ is completely determined by a {\em unique horizontal form}, and, in particular, if $q=n$, by a {\em Lagrangian}.

Summarizing, {\em for $q \leq n$ the operator of horizontalization},
\begin{gather*}
h\colon \ \Lambda^r_q \to \Lambda^{r+1}_{q,\bX} \subset \Lambda^{r+1}_q, \qquad \rho \to h\rho ,
\end{gather*}
{\em induces a representation mapping},
\begin{gather*}
R_q\colon \ {\cV}^r_q \to \Lambda^{r+1}_{q}, \qquad R_q([\rho]) = h\rho .
\end{gather*}

$\bullet$ If $q = n + k$ for $k>1$ then for every $q$-form $\rho$ of order~$r$, the contact decomposition~\eqref{splitting of forms} takes the form
\begin{gather*}
\pi_{r+1,r}^*\rho = p_k\rho + p_{k+1}\rho + \dots + p_q \rho,
\end{gather*}
and $\Theta^r_q = \Lambda^r_{q,c} + (d \Lambda^r_{q-1,c})$, where we have sheaves of strongly contact forms (with the only exception of $\Lambda^r_{n,c}$ which is a sheaf of contact forms). This means that given a class $[\rho] \in \cV^r_{n+k} = \Lambda^r_{n+k}/\Theta^r_{n+k}$,
one has $\rho_1 \sim \rho_2$ if\/f $\rho_1 - \rho_2 = \vartheta + d\eta$ where $\vartheta$ is a strongly contact $(n+k)$-form (i.e., such that $p_{k}\vartheta = 0$) and $\eta$ is a strongly contact $(n+k-1)$-form ($p_{k-1}\eta = 0$). Thus $p_k\rho_1 = p_k\rho_2 + p_k dp_{k} \eta \ne p_k \rho_2$.
However, $\rho_1 - \rho_2 \in \Theta^r_{n+k}$, so that
\begin{gather*}
\cI(\rho_1) = \cI(\rho_2).
\end{gather*}
Hence, every class $[\rho] \in \cV^r_{n+k}$ is completely determined by {\em a unique canonical source form}.

If we denote by $\cI \Lambda^{r}_{{n+k}}$ the image of $\Lambda^r_{n+k}$ by $\cI$, we can see that {\em the interior Euler operator}
\begin{gather*}
\cI\colon \ \Lambda^r_{n+k} \to \cI \Lambda^{r}_{n+k} \subset \Lambda^{2r+1}_{n+k}, \qquad \rho \to \cI (\rho) ,
\end{gather*}
{\em induces a representation mapping},
\begin{gather*}
R_{n+k}\colon \ {\cV}^r_{n+k} \to \Lambda^{2r+1}_{n+k} , \qquad R_{n+k}([\rho]) = \cI (\rho) .
\end{gather*}

For $k=1$ there is no other source form representing a class $[\rho] \in \cV^r_{n+1}$, so that every class is completely determined by
{\em a unique dynamical form}. If $[\rho] = \cE_{n}(\mu) = [d\mu]$ then
$\cI (d\mu) = \cI (dh\mu)$ is the Euler--Lagrange form of the Lagrangian $h\mu$. Note that this argument proves the uniqueness of the Euler--Lagrange form. At the same time Proposition~\ref{propsource} shows us that for $k\geq 2$ the representing canonical source form is no longer a unique source form in the class $[\rho]$. An explicit important example of nonuniquness for $k = 2$ is realized by Helmholtz forms, and will be discussed below in Section~\ref{Hel}.

$\bullet$ If $q = 0$ we def\/ine $R_0$ be the identity mapping.

Now, by means of the representation mappings $R_q$ we construct to the variational sequence $0 \to \R_{\bY} \to \cV^r_{*}$ a representation sequence
\begin{gather*}
0 \to \R_{\bY} \to R_0(\cV^{r}_0) \to R_1(\cV^{r}_1) \to R_2(\cV^{r}_2) \to \cdots
\end{gather*}
shortly denoted by $0\to \R_{\bY} \to R_{*}(\cV^{r}_{*})$, and called the {\em Takens representation} of $0 \to \R_{\bY} \to \cV^r_{*}$.
We also denote by
\begin{gather*}
E_q\colon \ R_q(\cV^r_q) \to R_{q+1}(\cV^r_{q+1}), \qquad q \geq 1,
\end{gather*}
the morphisms in the representation sequence. Hence, the def\/inition of the $E_q$'s follows from the commutativity of the diagrams
\begin{gather} \label{diagram1}
\begin{CD}
\cdots @>\mathcal{E}_{q-1}>>\cV_{q}^r @>\mathcal{E}_{q}>> \cV_{q+1}^r @>\mathcal{E}_{q+1}>> \cdots
\\
@. @V{R_q} VV@V{R_{q+1}} VV
\\
\cdots @>E_{q-1} >> R_{q}(\cV^r_{q})@>E_{q}>> R_{q+1}(\cV_{q+1}^r) @> E_{q+1}>> \cdots
\end{CD}
\end{gather}
and we can immediately see that the following theorem holds true:

\begin{Theorem}
The representation sequence $0\to \R_{\bY} \to R_{*}(\cV^{r}_{*})$ is exact.
\end{Theorem}

Indeed, by def\/inition of $E_q$ and $\cE_q$,
\begin{gather*}
E_q\circ R_q ([\rho])=R_{q+1}\circ \cE_q ([\rho])= R_{q+1}
([d\rho]) ,
\end{gather*}
and for every $\eta \in R_{q}(\cV_{q}^r)$ we have $\eta = R_q([\rho])$ for some $[\rho] \in \cV^r_{q}$. Hence
\begin{gather*}
 (E_{q+1} \circ E_{q}) (\eta) = (E_{q+1} \circ E_q \circ R_q)([\rho]) =
(E_{q+1} \circ R_{q+1} \circ \cE_q)([\rho]) \\
\hphantom{(E_{q+1} \circ E_{q}) (\eta)}{} =(E_{q+1} \circ R_{q+1})([d\rho]) = (R_{q+2} \circ \cE_{q+1}([d\rho]) =
R_{q+2} ([dd\rho]) = 0 .
\end{gather*}

The representation sequence can be easily understood, and explicit formulas for the morphisms $E_q$ and the forms $E_q(\eta) \in R_{q+1}(\cV^r_{q+1})$ can be easily obtained, if we write the diag\-ram~\eqref{diagram1} in the explicit form
\begin{gather*}%\label{diagram2}
\begin{CD}
\cdots @>\mathcal{E}_{n-1}>> \Lambda_n^r / \Theta_n^r @>\mathcal{E}_n>>
\Lambda_{n+1}^r / \Theta_{n+1}^r
@>\mathcal{E}_{n+1}>>
\Lambda_{n+2}^r / \Theta_{n+2}^r @>\mathcal{E}_{n+2}>> \cdots
\\
@. @V{R_n} VV @V{R_{n+1}} VV @V{R_{n+2}} VV
\\
\cdots @>E_{n-1} >> \Lambda_{n,\bX}^{r+1} @> E_n >> \Lambda_{n+1,1,\bY}^{r+1} @> E_{n+1} >>\cI\Lambda_{n+2}^{r+1}
@> >> \cdots
\end{CD}
\end{gather*}
where $R_q$, $1 \leq q \leq n$, has the meaning of horizontalization $h$ on classes of order $r$, and $R_{n+k}$ is the operator $\cI$ acting on $p_k \rho$ where $\rho$ is of order $r$.
Then since elements of the sheaves $R_q(\cV^r_q)$
are functions for $q=0$, horizontal forms for $1 \leq q \leq n$, and canonical source forms for $q \geq n+1$, we have
\begin{gather*}
E_0(f) = hdf,\qquad
E_q(h\rho) = hdh\rho, \qquad 1 \leq q \leq n-1,\\
E_n(h\rho) = \cI (dh\rho),\qquad
E_{n+k}(\cI (\rho)) = \cI (d (\cI(\rho))), \qquad k \geq 1 .
\end{gather*}

In particular, for $q=n$ we can write in coordinates $h\rho={L}\omega_0$,
where $\omega_0=dx^1\wedge\cdots\wedge dx^n$, and then
\begin{gather*}
E_n(h\rho) = R_{n+1}([d\rho])={\cal{I}}(d\rho)={\cal{I}}(dh\rho)=\sum_{|J|=0}^r
(-1)^{|J|} d_J \left(\frac{\partial{L}}{\partial
y^\sigma_J}\right)\omega^\sigma \wedge\omega_0.
\end{gather*}
Hence, $E_n\colon R_n(\cV^r_n) \to R_{n+1}(\cV^r_{n+1})$ is the {\em Euler--Lagrange mapping}, assigning to every Lag\-rangian $\lambda = h\rho$ its Euler--Lagrange form $E_\lambda = \cI (d\lambda)$. The next morphism, $E_{n+1}\colon R_{n+1}(\cV^r_{n+1}) \to R_{n+2}(\cV^r_{n+2})$ is the {\em Helm\-holtz mapping}, assigning to every dynamical form $\varepsilon$ its canonical Helm\-holtz form $H_{\varepsilon} = \cI (d\varepsilon)$. Note that the preceding morphism to the Euler--Lagrange morphism, $E_{n-1}\colon R_{n-1}(\cV^r_{n-1}) \to R_{n}(\cV^r_{n})$, assigns to every horizontal $(n-1)$-form $\varphi$ (resp., if $\dim X = n=1$, to a function $f$) a Lagrangian $\lambda = hd\varphi$ (resp.\ $\lambda = hdf$). This is so-called {\em null-Lagrangian} (also called {\em variationally trivial Lagrangian}), due to the fact that its Euler--Lagrange form is by exactness of the sequence equal to zero.

More generally, for $q \geq 1$, elements of $R_q(\cV^r_q)$ belonging to the kernel of the variational morphism $E_q$ are called {\em variationally trivial}.

We say that two elements $\eta_1$, $\eta_2$ of $R_q(\cV^r_q)$ are {\em $($variationally$)$ equivalent} if their dif\/ference is variationally trivial, i.e., $\eta_2 - \eta_1 \in \Ker E_q$. Thus, two Lagrangians are equivalent if\/f they dif\/fer by a null Lagrangian, $hd\rho$, two dynamical forms are equivalent if\/f they dif\/fer by a locally variational form, and two Helmholtz-like forms are equivalent if\/f they dif\/fer by a Helmholtz form.
Generally, two canonical source forms are equivalent if\/f they dif\/fer by $\cI(d(\cI \rho))$.

\begin{Example} \label{ExQP}
As an illustration of the use of the Takens representation of the variational sequence let us compute the Euler--Lagrange operator with help of the operator $\cI$. Indeed, given a Lagrangian~$\lambda$, the Euler--Lagrange form of~$\lambda$ arises either as the image of $\lambda$ by the mor\-phism~$E_n$, or, by application of $\cI$ to the exterior derivative of $\lambda$ (remind the formula $E_n(\lambda)={\cal{I}}(d\lambda)$ coming from the commutativity relation $E_n\circ{\cal{I}}={\cal{I}}\circ d$).

Consider the Lagrangian of a free quantum particle moving in one dimension. The quantum mechanical state of such a system is, in the coordinate representation, described by a wave function $\phi(t, x)\colon \R \times \R \to \C$, and the Lagrangian is
\begin{gather*}
\lambda=\left(-\frac{\hbar^2}{4m}(\phi_x\phi^*_x)+
\frac{{\rm i}\hbar}{4}(\phi^*\phi_t-\phi\phi^*_t)\right)\omega_0,
\qquad \omega_0=dt\wedge dx.
\end{gather*}
Denoting $v(t, x)=\operatorname{Re} \phi(t, x)$ and $w(t, x)=\operatorname{Im} \phi(t, x)$ we obtain
\begin{gather} \label{QP}
\lambda=\left(-\frac{\hbar^2}{4m}\big(v_x^2+w_x^2\big)-\frac{\hbar}{2}(vw_t-wv_t)\right)\omega_0,
\end{gather}
which is a f\/irst-order Lagrangian for the f\/ibred manifold $\pi\colon \R^2 \times \R^2 \to \R^2$ (hence $n=2$, $m=2$, $r=1$) where $\pi$ is the canonical projection, with f\/ibred coordinates
$(x^1,x^2,y^1,y^2,y^1_1,y^1_2,y^2_1,y^2_2)=(t,x,v,w,v_t,v_x,w_t,w_x)$. In this notation, the generating contact $1$-forms on $J^2(\R^2 \times \R^2)$ read as follows:
$\omega^1 = dv-v_tdt-v_xdx$,
$\omega^2 = dw-w_tdt-w_xdx$,
$\omega^1_t = dv_t-v_{tt}dt-v_{tx}dx$,
$\omega^1_x = dv_x-v_{xt}dt-v_{xx}dx$,
$\omega^2_t = dw_t-w_{tt}dt-w_{tx}dx$,
$\omega^2_x = dw_x-w_{xt}dt-w_{xx}dx$.
Now, $p_1d\lambda$ is the following 1-contact 3-form on $J^2(\R^2 \times \R^2)$:
\begin{gather*}
p_1d\lambda=\left(-\frac{\hbar}{2}w_tdv+\frac{\hbar}{2}v_tdw+\frac{\hbar}{2}wdv_t-
\frac{\hbar}{2}vdw_t-\frac{\hbar^2}{2m}v_xdv_x-\frac{\hbar^2}{2m}w_xdw_x
\right) \omega_0.
\end{gather*}
Calculating ${\cal{I}}(d\lambda)$ using formula (\ref{inter}) we obtain
\begin{gather*}
 {\cal{I}}(d\lambda) =\omega^1\wedge\left[-\frac{\hbar}{2}
w_t-d_t\left(\frac{\hbar}{2}w\right)
-d_x\left(-\frac{\hbar^2}{2m}v_x\right)\right]\omega_0 \\
\hphantom{{\cal{I}}(d\lambda) =}{} +\omega^2\wedge\left[\frac{\hbar}{2}
v_t-d_t\left(-\frac{\hbar}{2}v\right)
-d_x\left(-\frac{\hbar^2}{2m}w_x\right)\right]\omega_0 \\
\hphantom{{\cal{I}}(d\lambda)}{}
=\left[\left(-\hbar
w_t+\frac{\hbar^2}{2m}v_{xx}\right)\omega^1+\left(\hbar
v_t+\frac{\hbar^2}{2m}w_{xx}\right)\omega^2\right]\wedge\omega_0 .
\end{gather*}
The dynamical $3$-form ${\cal{I}}(d\lambda)$ is the Euler--Lagrange form of (\ref{QP}). This means that
\begin{gather*}
\varepsilon_1=\frac{\hbar^2}{2m}v_{xx}-\hbar w_t, \qquad
\varepsilon_2=\frac{\hbar^2}{2m}w_{xx}+\hbar v_t,
\end{gather*}
are the Euler--Lagrange expressions of the Lagrangian (\ref{QP}), and equations
$\varepsilon_1\circ J^2\gamma=0$, \mbox{$\varepsilon_2\circ J^2\gamma=0$} are the
Euler--Lagrange equations for extremals~$\gamma$ of~(\ref{QP}). Denoting
$\varepsilon =\varepsilon_1+{\rm i}\varepsilon_2$ we obtain the Euler--Lagrange equation
\begin{gather*}
\left(\frac{\hbar^2}{2m}\phi_{xx}+{\rm i}\hbar\phi_t\right)\circ
J^2\gamma=0,
\end{gather*}
which, indeed, is the Schr\"{o}dinger equation.
\end{Example}

\subsection[Lepage $n$-forms and the f\/irst variation formula]{Lepage $\boldsymbol{n}$-forms and the f\/irst variation formula} \label{secLep}

Lepage $n$-forms were introduced by Krupka in 1973~\cite{Kru73} (see also~\cite{Kru83}) in order to establish foundations of the higher-order calculus of variations in jet bundles. They are fundamental for a~geometric formulation of the intrinsic f\/irst variation formula, and of coordinate free, global Lagrangian and Hamiltonian mechanics and f\/ield theory. Combined with the concepts of inva\-riant variational functionals they provide geometric formulations of Noether theorems, as well as geometric integration methods based on symmetries.

Here we remind only some basic properties of Lepage $n$-forms, explored in the variational sequence theory. For more details and applications we refer the reader to the survey papers
\cite{Kru08,KrKrSa10,KrKrSa12,olga09,MuLe09}, and the book \cite{Kru15}.

As above, $\pi\colon \bY \to \bX$ is a f\/ibred manifold, $n = \dim \bX$, and $m = \dim \bY - n$.

\begin{Definition} Let $r \geq 0$. A $n$-form $\rho$ on $J^r\bY$ is called a {\em Lepage $n$-form} of
order $r$ if for every $\pi_{r,0}$-vertical vector f\/ield $\xi$ on
$J^r\bY$
\begin{gather*}
h(\xi\rfloor d\rho)=0.
\end{gather*}
\end{Definition}

Since $\rho$ is an $n$-form on $J^r\bY$ its horizontal component $h\rho$ is a Lagrangian on $J^{r+1}\bY$.

Notice that
\begin{itemize}\itemsep=0pt
\item {\em every} $n$-form on $\bY$ is a Lepage $n$-form. The corresponding Lagrangian is then def\/ined on~$J^1\bY$, and it is a polynomial of degree $n$ in the f\/irst derivatives\footnote{This fact is of particular importance in Hamiltonian mechanics and f\/ield theory, or in dealing with variational forces and energy-momentum tensors \cite{HK,Kru08,olga09-2}.},
\item {\em every closed} $n$-form on $J^r\bY$ is a Lepage $n$-form. The corresponding Lagrangian is a null-Lagrangian (giving rise to the zero Euler--Lagrange form).
\end{itemize}

The structure of Lepage $n$-forms is characterized as follows:

\begin{Theorem}[Krupka~\cite{Kru73}]\label{ThLep}
The following conditions are equivalent:
\begin{enumerate}\itemsep=0pt
\item[$(1)$] $\rho$ is a Lepage $n$-form of order $r$,
\item[$(2)$] $p_1 d\rho$ is a dynamical form, i.e., $p_1d\rho = \cI(d\rho)$,
\item[$(3)$] $\pi_{r+1,r}^*d\rho= E+F$,
 where $E$ is a dynamical form, and $F$ is at least $2$-contact,
\item[$(4)$] $\pi_{r+1,r}^*\rho= \tht_{h\rho} +d\nu+ \mu = \tht_{h\rho}
+p_1d\nu+\eta$,
where $\theta_{h\rho}$ is the Cartan form of the Lagran\-gian~$h\rho$,
$\nu$ is a contact
$(n-1)$-form, and $\mu$, resp.~$\eta$ is at least $2$-contact.
\end{enumerate}
\end{Theorem}

It should be stressed that for $n \geq 2$ and $r \geq 3$ the Cartan form is generically not global. However, the above theorem states that it can be ``globalized'' by adding $p_1d\nu$. Indeed, in the decomposition of $\pi_{r+1,r}^*\rho$ in (4), $\tht_{h\rho} +p_1d\nu$ and $\eta$ are global (being the at most $1$-contact and the at least $2$-contact part of $\rho$, respectively), while, in general, the decomposition $\tht_{h\rho} + p_1d\nu$ is not invariant under changes of f\/ibred coordinates.

The name `Cartan form' refers to \'E.~Cartan, who introduced the form to the classical calculus of variations~\cite{Cart}. Concerning dif\/ferent aspects of generalization to many independent variables we refer, e.g., to \cite{Bet,Cara, DeDo,Ded,Ga74,Got91,Got90,Kru73,Kru77,Ol,Ru}. Higher-order Cartan forms have been introduced and studied by many authors; see, e.g., \cite{DeDo,deLRo,Fe84,FeFra,GaMu,HoKo,Kru83} to name just a~few.

Accounting the def\/inition and properties of the operators $\cI$ and $\cR$ we are able to f\/ind an {\em intrinsic formula for the Cartan form}:

\begin{Theorem} \label{thmCartanform}
Let $r \geq 0$, and $\rho \in \Lambda^r_n$. Then the Cartan form of the Lagrangian $\lambda = h\rho$ takes the form
\begin{gather}\label{Cartan}
\theta_{\lambda} = \lambda - p_1 \cR(d\lambda) .
\end{gather}
\end{Theorem}

Notice that, indeed, $\theta_{h\rho}$ is a Lepage $n$-form, equivalent with $\rho$ (and $h\rho$), since by formu\-las~\eqref{pir} and~\eqref{pdr}
\begin{gather*}
 p_1d\theta_{h\rho} = \cI (d\theta_{h\rho}) + p_1 dp_1\cR(d\theta_{h\rho})
 = \cI (d\theta_{h\rho}) + p_1 dp_1\cR(dh\rho) - p_1 dp_1 \cR(dp_1 \cR(dh\rho)) \\
\hphantom{p_1d\theta_{h\rho}}{}
= \cI (d\theta_{h\rho}) + p_1 dp_1\cR(dh\rho) - p_1 dp_1 \cR(dh\rho) + \cI (dp_1 \cR(dh\rho))
= \cI (d\theta_{h\rho}),
\end{gather*}
and, writing $\pi_{r+1,r}^*\rho = h\rho + \beta$ we can see that (for a proper $s$)
\begin{gather*}
\pi_{s,r}^*\rho - \theta_{h\rho} = \beta + p_1 \cR(dh\rho) \in \Theta^s_n .
\end{gather*}

The meaning of Lepage $n$-forms for the calculus of variations comes from the fact that if~$\rho$ is a Lepage $n$-form then the dynamical form $p_1 d\rho = E_{h\rho}$ is the {\em Euler--Lagrange form} of the Lagrangian~$h\rho$. Indeed, the components of $p_1d\rho$ are the {\em Euler--Lagrange expressions} of the Lagrangian $\lambda = h\rho$.

We stress that even though $\tht_{h\rho}$ need not be global, the local forms do give rise to the {\em global form $p_1d\theta_{h\rho}$}. Indeed, formula~\eqref{Cartan} yields
\begin{gather*}
p_{1}d\theta_{\lambda} = p_{1}d\lambda - p_{1}dp_{1} \cR(d\lambda) = \cI(d\lambda),
\end{gather*}
which, indeed, is a global form.
(Note that we recovered the formula for the Euler--Lagrange form of $\lambda$ in Takens representation of the variational sequence.) Moreover, since for every Lepage form $\rho$, $\pi_{r+1,r}^*p_1d\rho= p_1d\tht_{h\rho}$, we get another proof of uniqueness of the Euler--Lagrange form.

As the horizontal parts of Lepage $n$-forms are Lagrangians, we can view a Lepage $n$-form as an {\em extension of a Lagrangian by a contact form}, and we come to the concept of {\em Lepage equivalent of a Lagrangian}~\cite{Kru73}. Given a Lagrangian~$\lam$, a~Lepage equivalent of $\lambda$ is a~Lepage form~$\rho$ such that $\lam=h\rho$.

Remarkably, if $\lambda$ is of order $r$, then its Lepage equivalents are generically of order $2r-1$, and the Euler--Lagrange form is of order~$2r$.

Local existence of Lepage equivalents of $\lambda$ and their structure follows immediately from the above theorem. Moreover, it is well known that every Lagrangian admits a global Lepage equivalent. For $n=1$ and any~$r$, or for any~$n$ and $r \leq 2$ a global Lepage equivalent of~$\lambda$ is the Cartan form. If $n=1$ (mechanics and higher-order mechanics) the Cartan form $\theta_\lambda$ is the {\em unique} Lepage equivalent of~$\lambda$. For $n>1$, Lepage equivalent of~$\lambda$ is no longer unique. Remarkably, apart from the (globalized) Cartan form, there are also other distinguished global Lepage equivalents of~$\lambda$, as, e.g., the celebrated {\em Carath\'eodory form} (\cite{Cara,Ol} for $r=1$, and~\cite{Sa08} for $r=2$), which is invariant with respect to all (not only f\/ibred) coordinate transformations, and the {\em Krupka--Betounes form} (for $r=1$) \cite{Bet,Kru77}, which has the property that $d\rho = 0$ if and only if $E_{h\rho} = 0$ (i.e., the Lagrangian belongs to~$\Ker E_n$).

\begin{Example}
Let us illustrate the use of the Lepage representation of the variational sequence on an example.
Consider the quantum particle Lagrangian~(\ref{QP}) in Example~\ref{ExQP} (including notations).
Its Cartan form reads
\begin{gather*}
 \theta_\lambda=L\omega_0+\frac{\partial L}{\partial v_t}\omega^1\wedge dx+
\frac{\partial L}{\partial w_t}\omega^2\wedge dx-\frac{\partial
L}{\partial v_x}\omega^1\wedge dt-\frac{\partial L}{\partial
w_x}\omega^2\wedge dt \\
\hphantom{\theta_\lambda}{}
=\left(-\frac{\hbar^2}{4m}\big(v_x^2+w_x^2\big)-\frac{\hbar}{2}(vw_t-wv_t)\right)\omega_0 \\
\hphantom{\theta_\lambda=}{}
 +\frac{\hbar^2}{2m}v_x\omega^1\wedge dt+\frac{\hbar^2}{2m}w_x\omega^2\wedge dt
+\frac{\hbar}{2}w\omega^1\wedge dx-\frac{\hbar}{2}v\omega^2\wedge dx.
\end{gather*}
As we have seen, being a Lepage equivalent of the Lagrangian~(\ref{QP}), the form $\theta_\lambda$ quickly gives us the Euler--Lagrange form $E_\lambda$, and hence the Euler--Lagrange equations:
\begin{gather*}
E_\lambda = p_1 d\theta_\lambda
=\left(\frac{\hbar^2}{2m}v_{xx}-\hbar
w_t\right)\omega^1\wedge\omega_0+\left(
\frac{\hbar^2}{2m}w_{xx}+\hbar v_t\right)\omega^2\wedge\omega_0.
\end{gather*}
We can see that, indeed, $E_\lambda$ equals to ${\cal{I}}(d\lambda)$ obtained in Example~\ref{ExQP}; this follows from
the commutativity relation $\operatorname{Lep}_{2}\circ E_2=d\circ \operatorname{Lep}_1$.

We note that the form $\theta_\lambda$ (or any other Lepage equivalent of the Lagrangian) also gives local Noetherian conserved currents related with symmetries of the Lagrangian and the Euler--Lagrange form (see the discussion below).
\end{Example}

With help of a Lepage equivalent of a Lagrangian the non-intrinsic procedure of integration by parts in the f\/irst variation of the action function is substituted by a geometric splitting with help of the Cartan formula for decomposition of the Lie derivative. As a consequence, the integral f\/irst variation formula appears also in the following {\em intrinsic differential form}~\cite{Kru73} (for more details see, e.g.,~\cite{Kru08}). Given a Lagrangian $\lambda$ of order $r$, and a $\pi$-projectable vector f\/ield $\Xi$ on $\bY$ (a ``variation''), it holds
\begin{gather} \label{inf1vf}
L_{J^r\Xi}\lam \equiv L_{J^{2r}\Xi}h\rho= h\big(J^{2r-1}\Xi \rfloor
d\rho\big) + hd\big(J^{2r-1}\Xi \rfloor\rho\big) \equiv h L_{J^{2r-1}\Xi}\rho,
\end{gather}
where $L_{J^r\Xi}$ denotes the Lie derivative along the $r$-jet prolongation of~$\Xi$, and~$\rho$ is (any) Lepage equivalent of $\lambda$. Let us show that the f\/irst term in the sum of~(\ref{inf1vf}) is the Euler--Lagrange term (carrying information about extremals), and the second term is the Noether term, carrying information about conservations laws.

We have $\lambda = h\rho$, and the Euler--Lagrange form is def\/ined by $E_\lambda=p_1d\varrho$.
For the f\/irst dif\/ferential form on the right-hand side of~(\ref{inf1vf}) we obtain
\begin{gather*}
h\big(J^{2r-1}\Xi\rfloor d\rho\big)=h\big(J^{2r}\Xi\rfloor\pi^*_{2r,2r-1}d\rho\big)=
h\big(J^{2r}\Xi\rfloor (p_1d\rho+\mu)\big) = h\big(J^{2r}\Xi\rfloor (E_\lambda)\big),
\end{gather*}
where $\mu$ is at least 2-contact, hence $J^{2r} \Xi\rfloor\mu$ is contact and $h(J^{2r} \Xi\rfloor\mu)=0$.
Denote $E_\lambda=E_\sigma\omega^\sigma\wedge\omega_0$ (i.e., $E_\sigma$ are the Euler--Lagrange expressions of $\lambda$), $\omega_i = \frac{\partial}{\partial x^i}\rfloor \omega_0$, and $\Xi=\xi^i\frac{\partial}{\partial x^i}+\Xi^\sigma\frac{\partial}{\partial y^\sigma}$. Then{\samepage
\begin{gather*}
h\big(J^{2r-1}\Xi\rfloor d\rho\big)=h\big(J^{2r}\Xi\rfloor
(E_\sigma\omega^\sigma\wedge\omega_0)\big)
=h\big(E_\sigma\Xi^\sigma\omega_0-E_\sigma\xi^i\omega^\sigma\wedge\omega_i\big)=
E_\sigma\Xi^\sigma \omega_0,
\end{gather*}
showing that, indeed, $h(J^{2r-1}\Xi\rfloor d\rho)$ provides Euler--Lagrange equations.}

Let us turn to the second dif\/ferential form in the sum of~(\ref{inf1vf}).
Assume that $\Xi$ is a Noether symmetry of the Lagran\-gian~$\lambda$, meaning that $L_{J^r\Xi}\lambda=0$. Then
in the inf\/initesimal f\/irst varia\-tion formula~(\ref{inf1vf}) the left-hand side vanishes, and along every extremal $\gamma$ of the Lagran\-gian~$\lambda$ the $n$-form $h(J^{2r-1}\Xi\rfloor d\rho)$ vanishes, because the Euler--Lagrange equations $E_\sigma \circ J^{2r}\gamma=0$, \mbox{$1 \leq \sigma \leq m$}, hold true. Consequently, the $n$-form $hd(J^{2r-1}\Xi \rfloor\rho)$ vanishes along extremals, which is the celebrated Noether theorem.
The $(n-1)$-form
\begin{gather*}
\Phi(\Xi)=J^{2r-1}\Xi\rfloor\rho
\end{gather*}
is then the corresponding Noether current.

In the simplest case of f\/irst-order mechanics when $\rho = \theta_\lambda = Ldt+\frac{\partial L}{\partial\dot{q}^\sigma}\omega^\sigma$,
we recover the well-known formulas
\begin{gather*}
 h\big(J^{1}\Xi \rfloor d\theta_\lambda\big)=\left(\frac{\partial L}{\partial q^\sigma}-
\frac{d}{dt} \frac{\partial L}{\partial\dot{q}^\sigma}\right)\Xi^\sigma dt,
\\
\Phi(\Xi)=L\xi^0+\frac{\partial
L}{\partial\dot{q}^\sigma}(\Xi^\sigma-\dot{q}^\sigma\xi^0) = -H \xi^0 + p_\sigma \Xi^\sigma.
\end{gather*}
If $\Xi$ is a Noether symmetry of $\lambda$, then the function $\Phi(\Xi)$ is constant along extremals.

We shall return to the variational splitting of the Lie derivative in a more general context in Section~\ref{section4}.

\subsection[Lepage forms of higher degrees, and the Lepage representation of the variational sequence]{Lepage forms of higher degrees, and the Lepage representation\\ of the variational sequence}

A motivation for introducing Lepage forms of higher degrees is the following observation:
If $\dim X = 1$ then every Lagrangian $\lambda$ has a unique Lepage equivalent, the Cartan form $\theta_\lambda$. The mapping $\Lep_1\colon \Lambda^r_{1,\bX} \ni \lambda \to \theta_\lambda \in \Lambda^{2r-1}_{1}$ thus relates the Euler--Lagrange mapping to the exterior derivative through the relationship $E_\lambda = p_1 d\theta_\lambda$. This suggests the idea to extend $\Lep_1$ to dynamical forms, i.e., to $\Lep_2\colon \Lambda^{s}_{2,\bY} \ni \varepsilon \to \alpha_{\varepsilon} \in \Lambda^{s}_{2}$, so that we would have
\begin{gather*}
\begin{CD}
\lambda@> E_1 >> E_\lambda
\\
@V{\Lep_1} VV @VV{\Lep_2} V
\\
\theta_\lambda @> d >> d\theta_\lambda
\end{CD}
\end{gather*}

The concept of a {\em Lepage equivalent of a dynamical form} was introduced in \cite{olga86} for locally variational forms. In that situation one obtains an extension of the Euler--Lagrange form to a~{\em closed form} which (similarly as in the case of a~Lepage equivalent of a Lagrangian) is unique for $n=1$ and nonunique otherwise
(see \cite{Gr08,HK,olga86,olga02,KS09}). Lepage equivalents of Euler--Lagrange forms became important particularly in study of the inverse variational problem, Hamiltonian theory, and geometric integration based on symmetries of the equations \cite{Got90,Kru15,olga97,olga02,olga09-2,KP08}.

Having the variational sequence and its Takens representation allows to def\/ine the concept of Lepage equivalent for any source form, and, in this way, to link all the variational operators to the exterior derivative.
Remarkably, with Lepage forms one gets another representation of the variational sequence where variational morphisms simplify to exterior derivatives. The immediate benef\/it is the reformulation and solution of the problem of existence of Lagrangians for given dif\/ferential equations ``as they stand'', the problem of the structure of null Lagrangians, and the other inverse problems in the variational sequence: by using Lepage forms these problems are reduced to application of the Poincar\'e lemma.

The idea how to generalize the concept of Lepage form to higher degrees is suggested by property (2) (or (3)) in Theorem \ref{ThLep}, that the lowest contact component of the exterior derivative of a Lepage form $\rho$ is the canonical source form for $d\rho$:

\begin{Definition}[\cite{KrMu05,KrSe05}]
Let $k \geq 0$. A $(n+k)$-form $\rho$ on $J^r \bY$ is called {\em Lepage form} if
\begin{gather} \label{defLep}
p_{k+1}d\rho=\cI (d\rho).
\end{gather}
\end{Definition}

Lepage $(n+k)$-forms for $k > 0$ have many similar properties as have Lepage $n$-forms. First,
\begin{itemize}\itemsep=0pt
\item {\em every} $q$-form on $\bY$, $q \geq n$, is a Lepage form,
\item {\em every closed} $q$-form on $J^r\bY$, $q \geq n$, is a Lepage form.
\end{itemize}

Structure of Lepage forms can be derived by direct calculations in coordinates from the def\/inition. However, this procedure is quite lengthy and tedious. Much more advantageously, we again explore the operator~$\cR$ to solve the equation~\eqref{defLep} with respect to~$\rho$ in an intrinsic way. Then we obtain:

\begin{Theorem} \label{thmhC}
Equation \eqref{defLep} has the solution
\begin{gather} \label{Lepdeg}
\pi_{r+1,r}^*\rho = \tht_{p_k\rho} +d\nu+ \mu = \tht_{p_k\rho} +p_{k+1}d\nu+\eta,
\end{gather}
where
\begin{gather} \label{Cartandeg}
\theta_{p_k\rho} = p_k\rho - p_{k+1} \cR(dp_k\rho) ,
\end{gather}
 $\nu$ is an arbitrary at least $(k+1)$-contact $(n+ k-1)$-form, and $\mu$ $($resp.~$\eta)$ is an arbitrary at least $(k+2)$-contact form.

For every choice of $\nu$ and $\mu$ $($resp.~$\eta)$ the Lepage forms~\eqref{Lepdeg} belong to the same variational class $[\rho] \in \cV^s_{n+k}$ $($for a proper~$s)$, i.e., $\pi_{s,r}^*\rho - \theta_{p_k\rho} \in \Theta^s_{n+k}$.
\end{Theorem}

\begin{proof}
First we show that $\theta_{p_k\rho}$ given by (\ref{Cartandeg}) is a Lepage form, i.e., satisf\/ies
$p_{k+1}d\theta_{p_k\rho} = \cI (d\theta_{p_k\rho})$.
Using the same procedure as above for $k=0$, we obtain:
\begin{gather*}
 p_{k+1}d\theta_{p_k\rho} = \cI (d\theta_{p_k\rho}) + p_{k+1} dp_{k+1}\cR(d\theta_{p_k\rho}) \\
\hphantom{p_{k+1}d\theta_{p_k\rho}}{}
= \cI (d\theta_{p_k\rho}) + p_{k+1} dp_{k+1}\cR(dp_k\rho) - p_{k+1}dp_{k+1} \cR(dp_{k+1} \cR(dp_k\rho)) = \cI (d\theta_{p_k\rho}),
\end{gather*}
since
\begin{gather*}
p_{k+1}dp_{k+1} \cR(dp_{k+1} \cR (dp_k\rho)) = p_{k+1} dp_{k+1} \cR(dp_k\rho) - \cI (dp_{k+1} \cR(dp_k\rho))\\
\hphantom{p_{k+1}dp_{k+1} \cR(dp_{k+1} \cR (dp_k\rho))}{}
= p_{k+1} dp_{k+1} \cR(dp_k\rho) ,
\end{gather*}
in view of $\cI (dp_{k+1} \cR(dp_k\rho)) = \cI (p_{k+1}(dp_{k+1} \cR(dp_k\rho))) = 0$. Hence, $\theta_{p_k\rho}$
is a solution of \eqref{defLep}. Since the equation concerns $d\rho$, it is clear that if $\rho$ is a solution then also $\rho + d\nu$ is a solution, where $\nu$ is an arbitrary $(n+k-1)$-form. Moreover, equation~\eqref{defLep} allows to determine only the lowest, $(k+1)$-contact component of $d\rho$. Thus, any solution~$\rho$ is determined up to
$d\nu+\mu$, where~$\mu$ is at least $(k+2)$-contact. In this way, so far we have obtained $\rho = \theta_{p_k\rho} + d\nu + \mu$, where~$\mu$ is at least $(k+2)$-contact.
We observe that $p_k \theta_{p_k\rho} = p_k\rho$, hence $p_k d\nu$ must be equal to zero, meaning that $\nu$ is at least $(k+1)$-contact.

Finally, we have to show that whatever the choice of $\nu$ and $\mu$, $\rho$ is equivalent with $\theta_{p_k\rho}$. This is, however, easily seen: if we denote by $s$ the order of $\theta_{p_k\rho}$, and by $\beta$ the at least $(k+1)$-contact component of $\rho$, we obtain
$\pi_{s,r}^*\rho - \theta_{p_k \rho} = p_k\rho + \beta - \theta_{p_k \rho} = \beta + p_{k+1} \cR(dp_k\rho)$, which is a strongly contact $(n+k)$-form, thus belonging to the kernel $\Theta_{n+k}^s$.
\end{proof}

The $(n+k)$-form $\theta_{p_k\rho}$ is the {\em higher-degree generalization of the Cartan form}. It is completely determined by its lowest contact ($k$-contact) component. Similarly as the Cartan form of degree~$n$ it is generically not global, since it is def\/ined by means of the operator~$\cR$. However, quite similarly as in the case of Lepage $n$-forms, it can be ``globalized'' by adding a proper term~$p_{k+1}d\nu$.

And similarly as in the familiar $k=0$ situation, $p_{k+1} d\theta_{p_k\rho}$ is global, and it holds
\begin{gather*}
p_{k+1}d\theta_{p_k\rho} = \cI(dp_k\rho) .
\end{gather*}

Also in the higher degree situation, there is a distinguished case of $n=1$ (and arbitrary $r$) (mechanics). Then $\nu$ is a $k$-form and~$\mu$ is a~$(k+1)$-form, hence they both are zero and we have:

\begin{Corollary}
Let $n = \dim \bX = 1$, $r \geq 0$. Then for every $k \geq 0$, the higher degree Cartan form $\theta_{p_k\rho}$ is unique and hence global.
\end{Corollary}

The higher-degree generalization of Theorem~\ref{ThLep} now reads as follows:

\begin{Theorem} \label{ThLepdeg}
Let $k \geq 0$. The following conditions are equivalent:
\begin{enumerate}\itemsep=0pt
\item[$(1)$] $\rho$ is a Lepage $(n+k)$-form of order $r$,

\item[$(2)$] $\pi_{r+1,r}^*d\rho= \cI(d\rho) +F$, where $F$ is at least $(k+2)$-contact,

\item[$(3)$] $p_{k+1}d\cR(p_{k+1}d\rho)=0$,

\item[$(4)$] $\pi_{r+1,r}^*\rho= \tht_{p_k\rho} +d\nu+ \mu = \tht_{p_k\rho} +p_{k+1}d\nu + \eta$,
where
$
\theta_{p_k\rho} = p_k\rho - p_{k+1} \cR(dp_k\rho)$,
 $\nu$ is an at least $(k+1)$-contact $(n+ k-1)$-form, and $\mu$, resp.~$\eta$ is at least $(k+2)$-contact.
 \end{enumerate}
\end{Theorem}

Again, in the decomposition of $\pi_{r+1,r}^*\rho$ in (4), $\tht_{p_k\rho} +p_{k+1}d\nu$ and $\eta$ are global (being the at most $(k+1)$-contact and the at least $(k+2)$-contact part of $\rho$, respectively), while, in general, the decomposition $\tht_{p_k\rho} + p_{k+1}d\nu$ is not invariant under changes of f\/ibred coordinates. And $p_{k+1} d\rho$ is unique (independent upon a choice of $\rho$).

Also the concept of Lepage equivalent of a Lagrangian extends to $(n+k)$-forms. Given a~$k$-contact form $\sigma$, by a {\em Lepage equivalent of $\sigma$} we mean a Lepage form $\rho$ such that $p_k\rho = \sigma$.
Theorem~\ref{ThLepdeg} then gives the structure of Lepage equivalents, and guarantees {\em local existence}.

There remains to answer a question about existence of {\em global} Lepage equivalents of higher degrees. The result is af\/f\/irmative, and its proof is a generalization of a proof for $n$-forms~\cite{Kru83}. The idea is to show that there exists a Lepage equivalent which is a global form on a closed submanifold of~$J^{2r+1}\bY$.
Then, since the (local) form def\/ines a soft sheaf of forms it can be extended to a global form on the whole space.

Let us consider the canonical injection
$\iot_{r+1, r}\colon J^{2r+1}\bY \to J^{ r+1}(J^{r%(r+1)-1
}\bY)$
 def\/ined by
$\iot_{r+1, r}(j^{2r+1}_{x}\gam)$ $= j^{r+1}_{x}(J^{r}\gam)$.

\begin{Theorem}
Let $\tilde{\phi}$ be a $(n+k)$-form such that $p_{l}\tilde{\phi}=0$, $l\geq k+1$, e.g., $\tilde{\phi}$ is the $p_k$ component defined on $J^{1}(J^r\bY)$ of a $(n+k)$-form on $J^r\bY$. There exists a $(n+k)$-form $\tilde{\eta}$ on
$J^{ r+1}(J^{r%(r+1)-1
}\bY)$ such that
\begin{enumerate}\itemsep=0pt
\item[$(1)$] $\tilde{\eta}$ is $(\pi_r )_{r+1, 0}$-horizontal and $\tilde{\eta}$ is a $(k+1)$-contact $(n+k)$-form on $J^{ r+1}(J^{r%(r+1)-1
}\bY)$;
\item[$(2)$] for every $(\pi_r)_{r+1,0}$-vertical vector f\/ield $\bZ$ on $J^{ r+1}(J^{r%(r+1)-1
}\bY)$ $($i.e., such that the projection of~$\bZ$ on~$J^{r}\bY$ is zero$)$
\begin{gather*}
\iot_{r+1, r}^*L_{\bZ}((\pi_r )^{*}_{r+1, 1}\tilde{\phi}+\tilde{\eta})
\end{gather*}
is a $(k+1)$-contact $(n+k)$-form on $J^{ r+1}(J^{r%(r+1)-1
}\bY)$;
\item[$(3)$] the form $\iot_{r+1, r}^* ((\pi_r )^{*}_{r+1, 1}\tilde{\phi}+\tilde{\eta}) $ is a Lepage $(n+k)$-form.
\end{enumerate}
\end{Theorem}

\begin{proof}
It is enough to prove uniqueness of $\tilde{\eta}$ when restricted to the submanifold $\iot_{r+1, r}(J^{2r+1}\bY)$.
In fact, for $\rho$ a $(n+k)$ form on $J^{r}\bY$,
we can take $\tilde{\eta} = -\cI(p_{k}\rho)$ when $\tilde{\phi}= p_k\rho$. The result follows by the uniqueness of $\cI(p_{k}\rho)$ which implies that it is a globally def\/ined form on a closed subset of~$J^{2r+1}\bY$ (def\/ined by the injection above), therefore since the (local) form def\/ines a soft sheaf of forms on~$J^{2r+1}\bY$, then~$\tilde{\eta}$ can be globalized.
\end{proof}

The most important Lepage forms of degree $> n$ are {\em Lepage equivalents of canonical source forms}. By the above, if $\sigma$ is a canonical source form of degree $n+k$, $k \geq 1$, and order~$r$, i.e., $\cI(\sigma)$ is $\pi_{2r+1,r}$-projectable, and $p_k\sigma = \sigma = \cI (\sigma)$, then all its local Lepage equivalents take the form
\begin{gather} \label{Lepeqdeg}
\rho = \tht_{\sigma} +d\nu+ \mu = \tht_{\sigma} +p_{k+1}d\nu+\eta,
\end{gather}
where
\begin{gather*}
\theta_{\sigma} = \sigma - p_{k+1} \cR(d\sigma) ,
\end{gather*}
$\nu$ is an arbitrary at least $(k+1)$-contact $(n+ k-1)$-form, and $\mu$ (resp.~$\eta$) is an arbitrary at least $(k+2)$-contact form. Moreover,
\begin{gather} \label{Lepeqdeg2}
p_{k+1}d\theta_{\sigma} = \cI(d\sigma) .
\end{gather}
In particular, if $\varepsilon$ is a dynamical form then $p_{2}d\theta_{\varepsilon}$ is the canonical Helmholtz form~$H_{\varepsilon } = \cI(d\varepsilon)$.

Comparing Lepage forms of dif\/ferent degrees yields the following results:

\begin{Theorem} \label{ThRd}
Let $q \geq n$. If $\alpha$ is a Lepage equivalent of $R_q([\rho])$ then $d\alpha$ is a Lepage equivalent of $R_{q+1}([d\rho])$.
\end{Theorem}

\begin{proof}
Since $\alpha$ is a Lepage form, $d\alpha$ is Lepage, as trivially follows from the def\/inition.
Denote $\sigma = R_q([\rho])$. If $q=n$, $\sigma$ is a Lagrangian, hence $R_n([\rho]) = h\rho$, and $R_{n+1}([d\rho]) = \cI(d\rho) = \cI(dh\rho)$. Thus $\alpha = \theta_{h\rho} + d\nu + \mu$, and~$d\alpha$ is a Lepage equivalent of $p_1d\alpha$. However, by~\eqref{Lepeqdeg2}, $p_{1}d\alpha = p_{1} d\theta_{h\rho} = \cI(dh\rho) =
R_{n+1}([d\rho])$. If $q=n+k$, $k>1$, we have $R_{n+k}([\rho]) = \cI(\rho)$, and $R_{n+k+1}([d\rho]) = \cI(d\rho) = \cI(d\cI(\rho)) = \cI(d\sigma)$. The form $d\alpha$ is a Lepage equivalent of $p_{k+1}d\alpha$, which by~\eqref{Lepeqdeg2} is equal to $p_{k+1} d\theta_{\sigma} = \cI(d\sigma) = R_{n+k+1}([d\rho])$.
\end{proof}

\begin{Theorem}
A closed Lepage equivalent of a canonical source form of degree $n+k$, $k \geq 1$, is locally equal to $d\alpha$ where $\alpha$ is a Lepage equivalent of a canonical source form of degree $n+k-1$ if $k \geq 2$, resp.\ of a Lagrangian if $k=1$.
\end{Theorem}

\begin{proof}
If $\beta$ is a closed Lepage equivalent of $\sigma = p_k\beta = \cI(\beta)$, then locally $\beta = d\alpha$ where (with help of the contact homotopy operator), $\alpha = \cA\beta$, and $p_{k-1} \alpha = \cA p_k \beta = \cA\sigma$.

If $k=1$ then $p_0\alpha = h\alpha$ is a Lagrangian, and $p_1d\alpha = p_1\beta = \cI(\beta)$, so that
$\alpha= \cA \beta$ is a Lepage equivalent of $h\alpha = \cA \sigma$.

If $k \geq 2$, set $\varepsilon = \cI(\cA\sigma)$. $\varepsilon$ is a canonical source form of degree $n+k-1$
representing the class~$[\cA \sigma]$. Therefore (by the preceding theorem), if $\alpha^\prime$ is a Lepage equivalent of $\cI(\cA \sigma)$ then $d\alpha^\prime$ is a~Lepage equivalent of $\cI (d\varepsilon) = \cI(d \cA \sigma) = \cI(dp_{k-1} \alpha) = \cI(d\alpha) = \cI(\beta) = \sigma$. Now, both~$d\alpha^\prime$ and~$\beta$ are Lepage equivalents of $\sigma$, hence we have $\beta = d\alpha^\prime + d\nu +\mu$, where $\nu$ is at least $(k+1)$-contact and $\mu$ is at least $(k+2)$-contact, and since $d\beta = 0$, we get $d\mu= 0$, i.e., locally $\mu = d \tau$ and $\tau = \cA \mu$. Hence, $\beta = d(\alpha ^\prime + \nu + \tau)$, and $\alpha = \cA \beta= \alpha^\prime + \nu + \tau$ is a Lepage form (dif\/fering from the Lepage $(k-1)$-form $\alpha^\prime$ by an at least $(k+1)$-contact form), and it is a Lepage equivalent of $p_{k-1} \alpha = p_{k-1} \alpha^\prime = \cI(\cA \sigma)$, as required. Moreover, we note that since $p_{k-1} \alpha = \cA p_k \beta = \cA\sigma$, it holds
$\cI(\cA \sigma)=\cA\sigma = \cA (\cI \sigma)$.
\end{proof}

\begin{Corollary}\quad
\begin{enumerate}\itemsep=0pt
\item[$(1)$] For every canonical source form $\sigma$
\begin{gather*}
\cI(\cA \sigma)=\cA\sigma = \cA (\cI \sigma).
\end{gather*}

\item[$(2)$] If $\sigma$ is a canonical source $(n+k)$-form then $\cA \sigma$ is a canonical source $(n+k-1)$-form.
\end{enumerate}
\end{Corollary}

\begin{Corollary}\quad
\begin{enumerate}\itemsep=0pt
\item[$(1)$] If $\beta$ is a closed Lepage equivalent of a dynamical form $\varepsilon$ then~$\cA\beta$ is a Lepage equivalent of the Lagrangian $\lambda = \cA \varepsilon$.

\item[$(2)$] If $\beta$ is a closed Lepage equivalent of a canonical source form $\sigma$ then
$\cA\beta$ is a Lepage equivalent of the canonical source form~$\cA\sigma$.
\end{enumerate}
\end{Corollary}

The above properties of Lepage forms suggest a new representation of the variational sequence, based on {\em Lepage equivalents of canonical source forms}, as follows: For every $k \geq 0$ we have a mapping $\Lep_{n+k}$ assigning to every Lagrangian $\lambda \in \Lambda^r_{n,\bX}$, if $k=0$, and every canonical source form $\sigma \in \cI(\Lambda^r_{n+k})$ if $k\geq 1$, {\em the family} $\{\rho\}$ of its Lepage equivalents~\eqref{Lepeqdeg}.

With help of Lepage mappings Theorem~\ref{ThRd} can be reformulated as follows:

\medskip

\noindent
{\bf Theorem~\ref{ThRd}$'$.}
{\em If $\Lep_q(R_{q}([\rho])) = \{\alpha\}$ then $\Lep_{q+1}(R_{q+1}([d\rho])) = \{d\alpha\}$, or}
\begin{gather*}
\{d \Lep_q(R_{q}([\rho]))\} = \{\Lep_{q+1}(R_{q+1}([d\rho]))\} .
\end{gather*}

The Lepage maps then induce the following representation mappings:
\begin{gather*}
\tilde R_q= R_q \qquad \text{for} \quad 0 \leq q \leq n-1 \qquad \text{and} \qquad
\tilde R_q = \Lep_q \circ R_q \qquad \text{for} \quad q \geq n.
\end{gather*}

Note that for $q >n$ the images by $\tilde R_q$ of classes in $\cV^r_q$ are {\em subsheaves} of $\Lambda^r_q$.

With help of the representation mappings $\tilde R_q$ we construct to the variational sequence $0 \to \R_{\bY} \to \cV^r_{*}$ a representation sequence
\begin{gather*}
0 \to \R_{\bY} \to \tilde R_0(\cV^{r}_0) \to \tilde R_1(\cV^{r}_1) \to \tilde R_2(\cV^{r}_2) \to \cdots
\end{gather*}
shortly denoted by $0\to \R_{\bY} \to \tilde R_{*}(\cV^{r}_{*})$, and called the {\em Lepage representation} of the variational sequence $0 \to \R_{\bY} \to \cV^r_{*}$.

\begin{Theorem}
The representation sequence $0\to \R_{\bY} \to \tilde R_{*}(\cV^{r}_{*})$ is exact, and every morphism $ \tilde R_q(\cV^{r}_q) \to \tilde R_{q+1}(\cV^{r}_{q+1})$ for $q \geq n$, is the exterior derivative operator~$d$ $($acting on classes of Lepage forms$)$.
\end{Theorem}

\begin{proof}
Denote
\begin{gather*}
\tilde E_q\colon \ \tilde R_q(\cV^r_q) \to \tilde R_{q+1}(\cV^r_{q+1}), \qquad q \geq n .
\end{gather*}
With the representation mapping $\tilde R_q$, Theorem~\ref{ThRd} states that
 `if $\tilde R_{q}([\rho]) = \{\alpha\}$ then $\tilde R_{q+1}([d\rho])$ $= \{d\alpha\}$', in other words, the exterior derivative operator $d$ extends in an obvious way to classes of Lepage forms; denoting the exterior derivative of classes by the same symbol~$d$, we can see that
 \begin{gather*}
 \tilde E_q(\{\alpha\}) = \{d\alpha\} \equiv d\{\alpha\}.
 \end{gather*}
 This proves that $\tilde E_q = d$, $q \geq n$.

In view of this result, and since $\tilde R_q = R_q$ for $0 \leq q \leq n-1$, and the Takens representation sequence is exact, it remains to check the exactness in the arrow~$E_{n-1}$. To this end, let us write the variational sequence and its Lepage representation in form of the following diagram:
\begin{gather*}
\begin{CD}
\cdots @>\mathcal{E}_{n-2}>>\cV_{n-1}^r @>\mathcal{E}_{n-1}>>\cV_{n}^r @>\mathcal{E}_{n}>>\cV_{n+1}^r @>\mathcal{E}_{n+1}>> \cdots
\\
@. @V{R_{n-1}} VV@V{\tilde R_n} VV@V{\tilde R_{n+1}} VV
\\
\cdots @>E_{n-2} >> \Lambda_{n-1,\bX}^r@>E_{n-1} >> \tilde R_{n}(\cV^r_{n})@>d >> \tilde R_{n+1}(\cV_{n+1}^r)@> d >> \cdots
\end{CD}
\end{gather*}
Then for every $\eta \in \tilde R_{n-1}(\cV_{n-1}^r) = \Lambda^r_{n-1,\bX}$ we obtain
\begin{gather*}
(d \circ E_{n-1})(\eta) = (d \circ E_{n-1})(R_{n-1}([\rho])) = (d \circ E_{n-1} \circ R_{n-1})([\rho])
\\
\hphantom{(d \circ E_{n-1})(\eta)}{}
 = (d \circ \tilde R_n \circ \cE_{n-1})([\rho]) = d (\tilde R_n([d\rho])) = d \Lep_{n}(R_{n}([d\rho])\\
\hphantom{(d \circ E_{n-1})(\eta)}{}
= \Lep_{n+1}(R_{n+1}([dd\rho]) = 0 .\tag*{\qed}
\end{gather*}
\renewcommand{\qed}{}
\end{proof}

It is worth to write all the three sequences within one scheme as follows:
\begin{gather*}
\begin{CD}
\cdots @>\mathcal{E}_{n-1}>> \Lambda_n^r / \Theta_n^r @>\mathcal{E}_n>>
\Lambda_{n+1}^r / \Theta_{n+1}^r
@>\mathcal{E}_{n+1}>>
\Lambda_{n+2}^r / \Theta_{n+2}^r @>\mathcal{E}_{n+2}>> \cdots
\\
@. @V{R_n} VV @V{R_{n+1}} VV @V{R_{n+2}} VV
\\
\cdots @>E_{n-1} >> \Lambda_{n,\bX}^{r+1} @> E_n >> \Lambda_{n+1,1,\bY}^{r+1} @> E_{n+1} >>\cI\Lambda_{n+2}^{r+1}
@>E_{n+2} >> \cdots
\\
@. @V{\Lep_n} VV @V{\Lep_{n+1}} VV @V{\Lep_{n+2}} VV
\\
\cdots @>E_{n-1} >> \Lep(\Lambda_{n,\bX}^{r+1}) @> d >> \Lep(\Lambda_{n+1,1,\bY}^{r+1}) @> d >>\Lep(\cI\Lambda_{n+2}^{r+1} )
@>d >> \cdots
\end{CD}
\end{gather*}

As we have seen, for $n = \dim X = 1$ (mechanics), Lepage equivalent is {\em unique} for every $k \geq 0$ (recall that this is the Cartan form of degree $n+k$). This means that for every $k$,
\begin{gather*}
\Lep_{n+k} \colon \ \cI(\Lambda^*_{n+k}) \to \Lambda^*_{n+k} ,
\end{gather*}
i.e., $\Lep_{n+k}(\cI \rho)$ is an element of the sheaf $\Lambda^*_{n+k}$ (rather than a family of elements of the sheaf), and $d$ is the ``true'' exterior derivative of dif\/ferential forms. Hence ``morally'' (up to orders of the sheaves) the Lepage representation sequence becomes a {\em subsequence of the de Rham sequence}.

Finally we want to stress that in view of the two corollaries of Theorem~\ref{ThRd} the contact homotopy operator~$\cA$ restricts to canonical source forms. This means that one can obtain primitives (inverse images) not only of closed Lepage forms in the Lepage representation, but also of canonical source forms in the Takens representation. For a primitive of a variational dynamical form $\varepsilon$ (i.e., a Lagrangian) one has $\lambda = \cA \varepsilon$ which in coordinates is the celebrated Tonti Lagrangian~\cite{Ton69-I,Ton69-II}. An analogous formula holds then, indeed, also for forms of higher degrees: if $\sigma = E_{n+k}(\eta)$ then $\cA \sigma = \eta$.

\begin{Example}
Let us write down explicit examples of Cartan forms of higher degrees in mechanics ($n =1$). Denote local f\/ibred coordinates on $\bY$ by $(t,q^\sigma)$ and the associated coordinates on $J^3\bY$ by $(t,q^\sigma,\dot q^\sigma, \ddot q^\sigma, \dddot q^\sigma)$, and let
\begin{gather*}
\omega^\sigma = dq^\sigma - \dot q^\sigma dt,
\qquad
\dot \omega^\sigma = \omega^\sigma_1 = d \dot q^\sigma - \ddot q^\sigma dt,
\qquad
\ddot \omega^\sigma = \omega^\sigma_2 = d \ddot q^\sigma - \dddot q^\sigma dt.
\end{gather*}
By Theorem \ref{thmhC}, formula (\ref{Cartandeg}) we have:

$\bullet$ For $\rho$ of degree $n=1$ and such that $h\rho = \lambda = L dt$ is of order $r = 1$: $\cR(d\lambda) = - \frac{\partial L}{\partial \dot q^\sigma} \omega^\sigma$, hence we obtain
\begin{gather*}
\theta_{\lambda} = \lambda + \frac{\partial L}{\partial \dot q^\sigma} \omega^\sigma ,
\end{gather*}
which is the classical Cartan $1$-form.

$\bullet$ For $\rho$ of degree $n+1 = 2$ such that $p_1\rho$ is a second-order dynamical form, i.e., $p_1 \rho = \varepsilon = E_\sigma \omega^\sigma \land dt$, the Cartan $2$-form is $\theta_\varepsilon = \varepsilon - \cR (d\varepsilon)$. To compute $\cR(d\varepsilon)$ we need the following formula for the residual operator applied to a~second-order $2$-contact $3$-form: denote
\begin{gather*}
\alpha = \sum_{i,j = 0}^2 A_{\sigma \nu}^{ij} \omega^\sigma_i \land \omega^\nu_j \land dt, \qquad A_{\sigma \nu}^{ij} = - A_{\nu \sigma}^{ji}
\end{gather*}
Then
\begin{gather*}
 \cR(\alpha) = \frac{1}{2} \left(A_{\sigma \nu}^{1j} - A_{\nu \sigma}^{j1} - \frac{d}{dt} \big(A_{\sigma \nu}^{2j} - A_{\nu \sigma}^{j2}\big) \right) \omega^\sigma \land \omega^\nu_j
\\
\hphantom{\cR(\alpha) =}{} - \frac{1}{2} \big(A_{\sigma \nu}^{2j} - A_{\nu \sigma}^{j2} \big) \omega^\sigma \land \omega^\nu_{j+1} + \frac{1}{2} \big(A_{\sigma \nu}^{2j} - A_{\nu \sigma}^{j2} \big) \omega^\sigma_1 \land \omega^\nu_j
\end{gather*}
(summation over $j = 0,1,2$). Taking $\alpha = d \varepsilon$ we obtain
\begin{gather*}
\theta_\varepsilon = \varepsilon + \frac{1}{2} \Bigl( \frac{\partial E_\sigma}{\partial \dot q^\nu}
- \frac{d}{dt} \frac{\partial E_\sigma}{\partial \ddot q^\nu} \Bigr) \omega^\sigma \land \omega^\nu -
 \frac{\partial E_\sigma}{\partial \ddot q^\nu} \dot \omega^\sigma \land \omega^\nu.
\end{gather*}
For a locally variational form $\varepsilon$ this is the Lepage equivalent of $\varepsilon$, introduced in \cite{olga86} (see also~\cite{olga97}).

$\bullet$ For $\rho$ of degree $n+2 = 3$ such that $p_2 \rho$ is a Helmholtz-like form denote $p_2 \rho = \eta = H_{\sigma \nu}^0 \omega^\sigma \land \omega^\nu \land dt + H_{\sigma \nu}^1 \omega^\sigma \land \dot \omega^\nu \land dt +
H_{\sigma \nu}^2 \omega^\sigma \land \ddot\omega^\nu \land dt$;
the Cartan $3$-form is $\theta_\eta = \eta - \cR (d\eta)$. The formula for the residual operator applied to a general second-order $3$-contact $4$-form
\begin{gather*}
\alpha = \sum_{i,j,k = 0}^2 A_{\sigma \nu \rho}^{ijk} \omega^\sigma_i \land \omega^\nu_j \land \omega^\rho_k \land dt, \qquad A_{\sigma \nu \rho}^{ijk} = - A_{\nu \sigma \rho}^{jik} = A_{\nu \rho \sigma}^{jki} ,
\end{gather*}
reads
\begin{gather*}
 \cR(\alpha) = - \left(\big\{A_{\sigma \nu \rho }^{1jk} \big\} - \frac{d}{dt} \big\{A_{\sigma \nu \rho}^{2jk} \big\} \right) \omega^\sigma \land \omega^\nu_j \land \omega^\rho_k
+ 2 \big\{A_{\sigma \nu \rho}^{2jk} \big\} \omega^\sigma \land \omega^\nu_{j} \land \omega^\rho_{k+1} \\
\hphantom{\cR(\alpha) =}{} -\big\{A_{\sigma \nu \rho}^{2jk} \big\} \omega^\sigma_1 \land \omega^\nu_{j} \land \omega^\rho_{k} ,
\end{gather*}
where $\{A_{\sigma \nu \rho}^{ijk} \} = \frac{1}{3} (A_{\sigma \nu \rho}^{ijk} + A_{\rho \sigma \nu}^{kij} + A_{\nu \rho \sigma}^{jki} )$, and summation over $j,k = 0,1,2$ applies. Now we take $\alpha = d \eta$, so that
\begin{gather*}
\big\{A_{\sigma \nu \rho}^{000}\big\} = \left\{\frac{\partial H^0_{\sigma \nu}}{\partial q^\rho} \right\}, \qquad
\big\{A_{\sigma \nu \rho}^{001}\big\} = \left\{\frac{\partial H^0_{\sigma \nu}}{\partial \dot q^\rho} + \frac{\partial H^1_{\nu \rho}}{\partial q^\sigma} \right\},
\qquad
\big\{A_{\sigma \nu \rho}^{011}\big\} = \left\{\frac{\partial H^1_{\sigma \nu}}{\partial \dot q^\rho} \right\}, \\
\big\{A_{\sigma \nu \rho}^{002}\big\} = \left\{\frac{\partial H^0_{\sigma \nu}}{\partial \ddot q^\rho} + \frac{\partial H^2_{\nu \rho}}{\partial q^\sigma} \right\},
\qquad\!
\big\{A_{\sigma \nu \rho}^{012}\big\} = \left\{\frac{\partial H^1_{\sigma \nu}}{\partial \ddot q^\rho} - \frac{\partial H^2_{\sigma \rho}}{\partial \dot q^\nu} \right\}, \qquad\!
\big\{A_{\sigma \nu \rho}^{022}\big\} = \left\{\frac{\partial H^2_{\sigma \nu}}{\partial \ddot q^\rho} \right\},
\end{gather*}
and $\{A_{\sigma \nu \rho}^{ijk}\} = 0$ for $i,j,k \ne 0$,
and obtain the Cartan $3$-form of $\eta$ as follows:
\begin{gather*}
 \theta_\eta = \eta - \cR(d\eta) = \eta
+ \left(\big\{A_{\sigma \nu \rho}^{100} \big\} - \frac{d}{dt} \big\{A_{\sigma \nu \rho}^{200} \big\} \right) \omega^\sigma \land \omega^\nu \land \omega^\rho
\\
\hphantom{\theta_\eta =}{}
 + 2 \left( \big\{A_{\sigma \nu \rho}^{101} \big\} - \big\{A^{200}_{\sigma \nu\rho}\big\} + \frac{1}{2}\big\{A^{002}_{\sigma\nu \rho}\big\}
 - \frac{d}{dt} \big\{A_{\sigma \nu \rho}^{201}\big\} \right) \omega^\sigma \land \omega^\nu \land \dot\omega^\rho
\\
\hphantom{\theta_\eta =}{}
 + 2 \left(\big\{A_{\sigma \nu \rho}^{102} \big\} - \big\{A^{201}_{\sigma \nu\rho}\big\}
- \frac{d}{dt} \big\{A_{\sigma \nu \rho}^{202} \big\} \right) \omega^\sigma \land \omega^\nu \land \ddot\omega^\rho
- 2 \big\{A^{202}_{\sigma \nu\rho}\big\} \omega^\sigma \land \omega^\nu \land \dddot \omega^\rho
\\
\hphantom{\theta_\eta =}{}
+ 2 \left( \big\{A_{\sigma \nu \rho}^{021} \big\} - \big\{A_{\sigma \nu \rho}^{210} \big\} \right) \omega^\sigma \land \dot \omega^\nu \land \dot\omega^\rho
+ 2 \left( \big\{A_{\sigma \nu \rho}^{022} \big\} - \big\{A_{\sigma \nu \rho}^{202} \big\} \right) \omega^\sigma \land \dot \omega^\nu \land \ddot \omega^\rho
\end{gather*}
with the non-zero $A$'s as above
(cf.~\cite{KrMa10} where a corresponding formula for the variationally trivial case was obtained).
\end{Example}

\subsection{The variational order} \label{varor}

A key issue of the Krupka variational sequence for the variational calculus, making the main dif\/ference between his and other approaches ($C$-spectral sequences, the variational bicomplex and inf\/inite order variational sequences) is that the variational sequence~(\ref{vseq}) {\em fixes} (and thus {\em defines}) {\em the order} of variational problems in the sense of the following def\/inition:

\begin{Definition}
Let $\eta \in \Lambda^s_q$ be a source form, or a Lepage form of degree $q \geq 1$ def\/ined on (an open subset of) $J^s\bY$. We say that $\eta$ has the {\em variational order $r$} if it comes from the variational sequence of order~$r$ (i.e., it is a Takens or Lepage representation of a variational class $[\rho] \in \Lambda^r_q/\Theta^r_q$.
\end{Definition}

How to understand this concept? A variational class $[\rho] \in \Lambda^r_q/\Theta^r_q$ concerns local dif\/ferential forms of degree $q$ and order $r$. Takens representation (and thus also Lepage representation) is based on the horizontalization operator $h = p_0$ and the interior Euler operator ${\cal I}$ which provide a dif\/ferential form of generally {\em higher order}. For example, if $r =1$ we have (for $q = n$) the Lagrangian
$\lambda = R_n([\rho]) = h\rho$, (for $q= n+1$) the dynamical form $\varepsilon =
R_{n+1}([\rho]) = {\cal I}(\rho)$, (for $q= n+2$) the Helmholtz-like form $H = R_{n+2}([\rho]) = {\cal I}(\rho)$, etc., all of order $2$. As given by the structure of ${\cal I}$, generically, the components of the dif\/ferential $(n+k)$-forms ($k \geq 1$) arising from $\Lambda^r_{n+k}/ \Theta^r_{n+k}$ are {\em polynomials in the derivatives starting from the order $r+1$} up to the order $2r$ in the linear term. Notice that such a polynomial behaviour is known to be typical for functions known as ``variational derivatives''. However the point is that the ``variational derivatives'' are not just general polynomials: the analysis of ${\cal I}$ shows that the coef\/f\/icients of these polynomials have a determined structure of {\em derivatives} of the components of $\rho$ (which are of order $r$), with prescribed symmetrization/skew-symmetrization rules in the indices, in other words they have to satisfy certain symmetry and integrability conditions (``order reducibility conditions''). For instance, a dynamical form $\varepsilon$ of order $2$ (which def\/ines a system of {\em second order differential equations})\footnote{In the Takens representation $\varepsilon$ is an element of the sheaf of dif\/ferential forms
$\Lambda^2_{n+1}$.} can be of {\em variational order $1$ or $2$}.
In particular, if $\varepsilon$ is {\em locally variational} (comes locally from a Lagrangian, i.e., its image by the Helmholtz morphism in the Takens representation sequence vanishes) we obtain that if the variational order of $\varepsilon$ is~$1$, we have local Lagrangians arising as $\lambda = h\rho$ from f\/irst order $\rho$; thus
$\lambda$ is of order $2$, or of order $1$ if $h\rho$ is $\pi_{2,1}$-projectable. If the variational order of
$\varepsilon$ is nontrivially $2$, we have local $3$rd order Lagrangians arising as $\lambda = h\rho$ from second order $\rho$, or second order Lagrangians if $h\rho$ is $\pi_{3,2}$-projectable (which, indeed, {\em cannot be reduced} to the f\/irst order by extracting a variationally trivial Lagrangian). The same order discussion concerns {\em global Lagrangians} if $H^n_{\rm dR}(\bY) = \{0\}$.

\begin{Example}
Lagrangians, dynamical forms and Helmholtz-like forms of {\em variational order one}:

$\bullet$ Let $n = \dim \bX = 1$. Any $\rho \in \Lambda^1_1/ \Theta^1_1$ takes the form
\begin{gather*}
\rho = A dt + F_\sigma \omega^\sigma + C_\sigma d\dot q^\sigma \sim
A dt + C_\sigma d\dot q^\sigma ,
\end{gather*}
so that
\begin{gather*}
\lambda = h\rho = L dt, \qquad \text{where} \quad L = A + C_\sigma \ddot q^\sigma.
\end{gather*}
We can see that a generic Lagrangian of variational order one in mechanics is of {\em second order, affine in the second derivatives}. Then, in particular, reducible second-order Lagrangians (and, of course, f\/irst-order Lagrangians) correspond to those classes $[\rho]$ which are generated by {\em $\pi_{1,0}$-horizontal forms}.

If $n = \dim \bX > 1$, we denote $\omega_{i_1} = \der/\der x^{i_1} \rfloor \omega_0$,
$\omega_{i_1 i_2} = \der/\der x^{i_2} \rfloor \omega_{i_1}$, etc.; we have
\begin{gather*}
\rho \sim A \omega_0
+ \sum_{k = 1}^n C_{\sigma_1 \dots \sigma_k}^{j_1 \dots j_k,i_1 \dots i_k} dy^{\sigma_1}_{j_1} \land \dots \land dy^{\sigma_k}_{j_k} \land \omega_{i_1 \dots i_k} ,
\end{gather*}
where the coef\/f\/icients $C_{\sigma_1 \dots \sigma_k}^{j_1 \dots j_k,i_1 \dots i_k}$ are totally antisymmetric in the indices $i_1 \dots i_k$, and totally symmetric in the pairs $(\sigma_1,j_1) \dots (\sigma_k,j_k)$. Now,
$\lambda = h \rho = L \omega_0$, where
\begin{gather*}
L = A + \bar C_{\sigma_1}^{j_1,i_1} y^{\sigma_1}_{j_1i_1} + \bar C_{\sigma_1 \sigma_2}^{j_1 j_2,i_1 i_2} y^{\sigma_1}_{j_1i_1} y^{\sigma_2}_{j_2i_2} +
\cdots + \bar C_{\sigma_1 \dots \sigma_n}^{j_1 \dots j_n,i_1 \dots i_n} y^{\sigma_1}_{j_1 i_1} \cdots y^{\sigma_n}_{j_n i_n}
\end{gather*}
(here the $\bar C$'s are just corresponding multipliers of the~$C$'s).
Thus, in f\/ield theory, a generic Lagrangian of variational order one is of second order, polynomial of degree $n$ in the second derivatives, with coef\/f\/icients obeying the symmetry-antisymmetry relations described above. In particular, reducible (and f\/irst-order) Lagrangians correspond to generating forms $\rho$ that are $\pi_{1,0}$-horizontal. The most famous Lagrangian of this kind is the scalar curvature, giving rise to the Einstein equations.

$\bullet$ Let again $n = \dim \bX = 1$. Any $\rho \in \Lambda^1_2/ \Theta^1_2$ takes the form
\begin{gather*}
\rho \sim A^1_\sigma \omega^\sigma \land dt + A^2_\sigma d \omega^\sigma +
 B^1_{\sigma \nu} \omega^\sigma \land d\dot q^\nu + C_{\sigma \nu} d \dot q^\sigma \land d \dot q^\nu
 \\
\hphantom{\rho}{} \sim \omega^\sigma \land (A_\sigma dt + B_{\sigma \nu} d\dot q^\nu)
 + C_{\sigma \nu} d \dot q^\sigma \land d \dot q^\nu,
\end{gather*}
where $C_{\sigma \nu} = - C_{\nu \sigma}$. In Takens representation we get the dynamical form $\cI (\rho) = \varepsilon = E_\sigma \omega^\sigma \land dt$,
where
\begin{gather*}
E_\sigma = A_\sigma + B_{\sigma \nu} \ddot q^\nu - \frac{d}{dt} (2C_{\sigma \nu}\ddot q^\nu )
= A_\sigma + \bar B_{\sigma \nu} \ddot q^\nu - 2 \frac{\partial C_{\sigma \nu}}{\partial \dot q^\rho} \ddot q^\nu \ddot q^\rho - 2 C_{\sigma \nu} \dddot q^\nu.
\end{gather*}
Thus a generic dynamical form of variational order one in mechanics is of order {\em three, affine in the third derivatives, and polynomial of degree two in the second derivatives $($with a special form of coefficients as above$)$}. {\em Second}-order dynamical forms of variational order one correspond to generating forms $\rho$ such that $C_{\sigma \nu} = 0$, meaning that the corresponding variational class is represented by a {\em $\omega^\sigma$-generated form} $\rho$ and the dif\/ferential equations are {\em affine in second derivatives}:
\begin{gather*}
\rho \sim A_\sigma \omega^\sigma \land dt + B_{\sigma \nu} \omega^\sigma \land d\dot q^\nu,
\qquad \text{i.e.,} \qquad
E_\sigma = A_\sigma + B_{\sigma \nu} \ddot q^\nu.
\end{gather*}
In particular, if variational ($[d \rho] = 0$, i.e., $[\rho] = [d \lambda]$), they come from second-order Lagrangians of the form
\begin{gather*}
L = A + C_\sigma \ddot q^\sigma, \qquad \frac{\partial C_\sigma}{\partial \dot q^\nu} = \frac{\partial C_\nu}{\partial \dot q^\sigma} ,
\end{gather*}
that, indeed, are reducible (equivalent with f\/irst-order Lagrangians).

Finally, {\em all} f\/irst-order dynamical forms are of variational order one and they arise from the class $[\rho] = [A_\sigma \omega^\sigma \land dt]$.

$\bullet$ Any $\rho \in \Lambda^1_3/ \Theta^1_3$ takes the form
\begin{gather*}
 \rho \sim A^1_{\sigma \nu} \omega^\sigma \land \omega^\nu \land dt + A^2_{\sigma \nu} \omega^\sigma \land d \omega^\nu + A^3_{\sigma \nu} d \dot q^\sigma \land d \omega^\nu
\\
\hphantom{\rho \sim}{} + B^1_{\sigma \nu \rho} \omega^\sigma \land \omega^\nu \land d \dot q^\rho + B^2_{\sigma \nu \rho} \omega^\sigma \land d \dot q^\nu \land d \dot q^\rho + C_{\sigma \nu \rho} d \dot q^\sigma \land d\dot q^\nu \land d \dot q^\rho,
\end{gather*}
where $A^1_{\sigma \nu} = - A^1_{\nu \sigma}$, $A^3_{\sigma \nu} = - A^3_{\nu \sigma}$, $B^1_{\sigma \nu \rho} = - B^1_{\nu \sigma \rho}$, $B^2_{\sigma \nu \rho} = - B^2_{\sigma \rho \nu}$, and the $C$'s are totally antisymmetric.
Then the corresponding Helmholtz-like form becomes $\cI (\rho) = \eta = H^1_{\sigma \nu} \omega^\sigma \land \omega^\nu \land dt + H^2_{\sigma \nu} \omega^\sigma \land \dot \omega^\nu \land dt + H^3_{\sigma \nu} \omega^\sigma \land \ddot \omega^\nu \land dt$, where, as expected, the components have a polynomial structure in second and higher derivatives:
\begin{gather*}
 H^1_{\sigma\nu} = A^1_{\sigma \nu} + B^1_{\sigma \nu \rho} \ddot q^\rho +
\frac{1}{4} \frac{d}{dt} \bigl( A^2_{\sigma \nu} - A^2_{\nu \sigma} - 2 \big( B^2_{\sigma \nu \rho} - B^2_{\nu \sigma \rho} \big) \ddot q^\rho \bigr)
- \frac{1}{2} \frac{d^2}{dt^2} \big( A^3_{\sigma \nu} - 3 C_{\sigma \nu \rho} \ddot q^\rho\big)
\\
\hphantom{H^1_{\sigma\nu}}{}
 = \alpha^1_{\sigma \nu} + \alpha^2_{\sigma \nu \rho} \ddot q^\rho + \alpha^3_{\sigma \nu \rho \lambda} \ddot q ^\rho \ddot q^\lambda + \alpha^4_{\sigma \nu \rho \lambda \kappa} \ddot q ^\rho \ddot q^\lambda \ddot q^\kappa +
\alpha^5_{\sigma \nu \rho} \dddot q^\rho
 +\alpha^6_{\sigma \nu \rho \lambda} \ddot q^\rho \dddot q^\lambda + \frac{3}{2} C_{\sigma \nu \rho} q^\rho_4 ,
\\
 H^2_{\sigma \nu} = - \frac{1}{2} \big(A^2_{\sigma \nu} + A^2_{\nu \sigma}\big) + \big(B^2_{\sigma \nu \rho} + B^2_{\nu \sigma \rho}\big) \ddot q^\rho ,
 \qquad
 H^3_{\sigma \nu} = A^3_{\sigma \nu} - 3 C_{\sigma \nu \rho} \ddot q^\rho .
 \end{gather*}
 The requirement that the class $[\rho]$ is $\omega^\sigma$-generated gives $A^3_{\sigma \nu} = 0$ and $C_{\sigma \nu \rho} = 0$, and simplif\/ies the Helmholtz-like forms to $H^3_{\sigma \nu}= 0$, and
 \begin{gather*}
H^1_{\sigma\nu} = \alpha^1_{\sigma \nu} + \alpha^2_{\sigma \nu \rho} \ddot q^\rho + \alpha^3_{\sigma \nu \rho \lambda} \ddot q ^\rho \ddot q^\lambda +
\alpha^5_{\sigma \nu \rho} \dddot q^\rho .
\end{gather*}
If, moreover, $\eta$ is a Helmholtz form (i.e., $[d\rho] = 0$), related with a $\omega^\sigma$-generated class then it corresponds to a~second-order dynamical form $\varepsilon = E_\sigma \omega^\sigma \land dt$, where
$E_\sigma = A_\sigma + B_{\sigma \nu} \ddot q^\nu$, and $B_{\sigma \nu} = B_{\nu \sigma}$.
\end{Example}

The above discussion illustrates that the range of applications of the variational sequence is wider than that of corresponding {\em infinite order} constructions or even of some other known {\em finite order} constructions. Remarkably, it gives solution of the

Inverse order problem (order reduction problem):
{\em If a differential $q$-form $\eta$ of order $s$ is $($locally/globally$)$ trivial, or $($locally/globally$)$ variational $($its image under the corresponding morphism~$E_q$ vanishes$)$ what is the `minimal order' of the $($local/global$)$ preimage $(q-1)$-form $\mu$ such that
$\eta = E_{q-1}(\mu)$?}

Note that this question is answered by the variational sequence, but is {\em not} answered, e.g., by the variational bicomplex, or even by Takens representation(!), since $\eta$ is an element of a sheaf of dif\/ferential forms of order $s$ which contains forms of {\em different variational orders}.

\begin{Example}
Consider a dynamical form $\varepsilon$ (nontrivially) of order three over a f\/ibred manifold with $n = \dim \bX = 1$. $\varepsilon$ belongs to $\Lambda^3_2$ and represents a system of 3rd order ordinary dif\/ferential equations. Takens sequence gives us that if locally variational, $\varepsilon$ comes from a local Lagrangian of order~$3$. However, from the variational sequence we can get information about the variational order of $\varepsilon$, and hence a f\/iner result about the order of the local Lagrangians. Namely, the following cases may arise:
\begin{itemize}\itemsep=0pt
\item The variational order of $\varepsilon$ is (nontrivially) three. This means that in the variational sequence $\varepsilon$ is represented by a class $[\rho] \in \Lambda^3_2/ \Theta^3_2$. If $\varepsilon$ is locally variational, any minimal order Lagrangian is of order three (not further reducible).

\item The variational order of $\varepsilon$ is two. This means that in the variational sequence $\varepsilon$ is represented by a class $[\rho] \in \Lambda^2_2/ \Theta^2_2$. If locally variational, it has a (local) second-order Lagrangian.

\item The variational order of $\varepsilon$ is one. This means that in the variational sequence $\varepsilon$ is represented by a class $[\rho] \in \Lambda^1_2/ \Theta^1_2$. As follows from the previous example, if $\varepsilon$ is locally variational then the Tonti Lagrangian (which is of order three) is locally reducible to a~second-order Lagrangian af\/f\/ine in the second derivatives.
\end{itemize}
\end{Example}

\section{Lie derivative in the variational sequence}\label{section4}

One of the most important features of the geometric formulation of the calculus of variations
in jet bundles is the fact that {\em variations can be described by Lie derivatives of differential
forms with respect to prolongations of projectable vector fields}. A famous example
is the f\/irst variation formula which takes the integral form~\cite{Kru73} (cf.~\cite{GoSt73})
\begin{gather*}
\int_ {\Omega} J^r \gamma^* L_{J^r\Xi}\lambda = \int_ {\Omega} J^{2r-1} \gamma^*\big(J^{2r-1}\Xi \rfloor
d\rho\big) + \int_ {\Omega} J^{2r-1} \gamma^* d\big(J^{2r-1}\Xi \rfloor\rho\big) ,
\end{gather*}
where $\rho$ is a Lepage equivalent of $\lambda$, $\Xi$ is a projectable vector f\/ield on $\bY$ (``variation vector f\/ield''), and $\Omega \subset \bX$ is a compact connected $n$-dimensional submanifold with boundary.
It is, however, even more interesting (and fundamental) that with use of the dif\/ferential calculus in jet bundles, this formula can be (as mentioned in Section~\ref{secLep}), invariantly expressed in the ``inf\/initesimal form'' as a {\em formula for differential forms} on the manifold~$J^{2r}\bY$~\cite{Kru73}
\begin{gather} \label{1vf}
L_{J^r\Xi}\lam = h\big(J^{2r-1}\Xi \rfloor
d\rho\big) + hd\big(J^{2r-1}\Xi \rfloor\rho\big) ,
\end{gather}
and (with the account of $\lambda = h\rho$) it yields a remarkable property of the Lie derivative in the calculus of variations:
\begin{gather} \label{hL}
 L_{J^{2r}\Xi}h\rho= h L_{J^{2r-1}\Xi}\rho .
\end{gather}
Moreover, as discovered in \cite{KT}, the Lie derivative commutes with the Euler--Lagrange operator, as the Lie derivative of the Euler--Lagrange form of $\lambda$ is the Euler--Lagrange form of the transformed Lagrangian:
\begin{gather} \label{LieE}
L_{J^{2r}\Xi} E_\lambda = E_{L_{J^r\Xi}\lam} .
\end{gather}
Such properties of the Lie derivative provide an elegant intrinsic formulation (and an easy proof) of Noether theorem, Noether--Bessel-Hagen theorem \cite{BeHa21, Noe18}, conservation laws, and relationship between symmetries of a Lagrangian and its Euler--Lagrange form (for any order and any number of independent variables)~\cite{Kru83,Kru15}.

Having in mind the variational sequence and its representations, there arises a question about generalization of these properties to variational forms of any degree.

First we recall results obtained in \cite{KrKr08, KKPS05,KKPS07}. Let ${\cal C}_r$ denote the contact ideal of order~$r$.

\begin{Definition}
A vector f\/ield $Z$ on $J^r\bY$ is called a {\em contact symmetry} if it is a symmetry of the contact ideal ${\cal C}_r$ (meaning that for every contact form $\omega$ the Lie derivative $L_Z \omega$ is a contact form.
\end{Definition}

Contact symmetries were characterized in~\cite{KrKr08}. It is worth notice that for any projectable vector f\/ield~$\Xi$ on~$\bY$ the $r$-jet prolongation $Z = J^r\Xi$ is a contact symmetry, and conversely, if a~contact symmetry $Z$ on $J^rY$ is projectable onto $\bX$ then $Z = J^r\Xi$ for some $\Xi$ on $\bY$.

Let $k \geq 0$. As an immediate consequence of the def\/inition it follows that the Lie derivative of a $k$-contact form by a~contact symmetry is an {\em at least $k$-contact form}. This means, however, that the Lie derivatives by contact symmetries preserve the sheaves $\Theta^r_q$.

\begin{Proposition}
Let $Z$ be a contact symmetry defined on $J^r\bY$. Then for all $q \geq 1$, \mbox{$L_Z \Theta^r_q \subset \Theta^r_q$}.
\end{Proposition}

Equivalently we can claim that for any two forms $\rho_1$, $\rho_2$ belonging to the same class $[\rho]$ in the variational sequence, the Lie derivatives $L_Z \rho_1$, $L_Z \rho_2$ by any contact symmetry $Z$ are also equivalent. Thus, following~\cite{KKPS05} we can def\/ine the {\em Lie derivative} ${\cal L}_Z$ {\em of a~class}
\begin{gather} \label{Licl}
{\cal L}_Z [\rho] = [L_Z \rho].
\end{gather}

Now one can easily prove that the property (\ref{LieE}) extends for any contact symmetry, and through the whole variational sequence, for any morphism ${\cal E}_k\colon \Lambda^r_k/ \Theta^r_k \to \Lambda^r_{k+1}/ \Theta^r_{k+1}$ \cite{KKPS05} (see also~\cite{KKPS07}):

\begin{Theorem} \label{thmLie}
Let a vector field $Z$ on $J^r\bY$ be a contact symmetry. Then the Lie derivative $L_Z$ commutes with all morphisms in the variational sequence. More explicitly, for all $k \geq 1$
\begin{gather} \label{Limor}
{\cal L}_Z {\cal E}_k [\rho] = {\cal E}_k ({\cal L}_Z [\rho]) = {\cal E}_k [ Z \rfloor d\rho ] .
\end{gather}
\end{Theorem}

\begin{proof}	
Since the Lie derivative commutes with the exterior derivative, we have for any $k$-form~$\rho$ on~$J^r\bY$
\begin{gather*}
{\cal L}_Z [d\rho] = [L_Z d \rho] = [d L_Z \rho] = [ d Z \rfloor d\rho ] .
\end{gather*}
Writing this formula in terms of the morphism ${\cal E}_k$ we get (\ref{Limor}).
\end{proof}

From now on, we shall restrict to {\em $\pi_s$-projectable}
 contact symmetries (recall that then we have $Z = J^s\Xi$ for a projectable vector f\/ield $\Xi$ on $\bY$). From the def\/inition it is clear that the Lie derivative then preserves even individual contact components. Namely, if $Z$ is a projectable contact symmetry on $J^{r+1}\bY$ then $Z = J^{r+1}\Xi$ and
\begin{gather*} %\label{contLi}
L_{J^{r+1}\Xi} p_k\rho = p_k L_{J^{r}\Xi} \rho, \qquad k \geq 0,
\end{gather*}
generalizing (\ref{hL}) to variational forms of any degree. Note that we may write
\begin{gather*}
L_{J^{r+1}\Xi} \pi_{r+1,r}^* \rho = L_{J^{r+1}\Xi} \left( \sum_{k=0}^q p_k \rho \right) = \sum_{k=0}^q p_k (L_{J^{r}\Xi} \rho) .
\end{gather*}

Let us turn back to formula (\ref{Licl}). It means that the Lie derivative of classes of forms, i.e., {\em variational Lie derivative}, can be correctly def\/ined as the equivalence class of the standard Lie derivative of forms, and thus {\em represented by forms}. Such a point of view opens a possibility to express the inf\/initesimal f\/irst variation formula~(\ref{1vf}) in terms of the {\em variational morphisms}, and, in this way, to obtain a generalization of the f\/irst variation formula for {\em any degree} of forms (with higher degree analogs of Noether theorems, indeed).

Some results in this direction have been already achieved, namely explicit formulae for the quotient Lie derivative operators
were provided, as well as corresponding versions of Noether theorems interpreted in terms of
conserved currents for Lagrangians and Euler--Lagrange morphisms. However, only classes of forms up to degree $n+2$,
were considered (the latter assumed to be exact) \cite{FPV02}. The representation made use of intrinsic decomposition formulae for vertical morphisms due to Kol\'a\v{r} \cite{Kol83}, which express geometrically the integration by part procedure, however, in a dif\/ferent way compared to that based on the splitting of the Cartan form \cite{GoSt73,Kru73}. The decomposition formulae introduce local objects such as {\em momenta} which could be glo\-ba\-li\-zed by means of connections, and besides the usual momentum associated with a Lagrangian, a~`generalized' momentum is associated with an Euler--Lagrange type morphism. Its interpretation in the calculus of variations has not yet been exhaustively exploited; as conjectured in~\cite{PaVi01}, generalized momenta could play a r\^ole within the multisymplectic framework for f\/ield theories.

In the sequel of this section we propose an approach to the problem of representation of the Lie derivative of classes of forms which encompasses previous results and generalizes to any degree of forms.
Actually, representation of (non exact) classes of forms of any degree appears
relevant, especially in degree $n+3$, for the study of nonvariational dif\/ferential equations arising from Helmholtz-like forms (related with the problem of the existence of closed $n+2$ forms)~\cite{KrMa10, Ma09}. We shall exploit the {\em relation between the interior Euler operator and the Cartan formula for the Lie derivative of differential forms}. The representation of the variational Lie derivative will provide, in a quite simple and immediate way, `any degree' generalizations of~(\ref{1vf}), together with the generalized Noether theorems, as `{\em quotient Cartan formulae}'.

These results, which have an intrinsic importance from a theoretical point of view, as shown in Section $5$, have also various concrete applications; besides the ones presented here, it is worth mentioning that applications of iterated variational Lie derivatives in variational problems on gauge-natural bundles produced physically relevant results concerning existence and globality of conserved quantities and on existence of Higgs f\/ields on spinor bundles (see, e.g., \cite{FFPW08,FFPW11,PaWi03,PaWi07,PaWi08,PaWi13a}).

Let $\Xi$ be a projectable vector f\/ield on $\bY$.
We def\/ine the interior product of a prolongation of $\Xi$ with the equivalence class of $\rho$ by the formula
\begin{gather} \label{cocl}
J^r\Xi \rfloor [\rho] = \big[J^s\Xi \rfloor R_q([\rho])\big]
\end{gather}
(the equivalence class of the interior product of the vector f\/ield with
the representation of the equivalence class of~$\rho$). This
def\/inition is well posed since the representation $R_q$ is unique.

Given a $q$-form $\rho$
def\/ined (locally) on $J^r\bY$, we have a commutative diagram def\/ining an operator~$\hat{R}_q$
\begin{gather*}
\hat{R}_q (\cL_{J^r\Xi} [\rho] ) \equiv \hat{R}_q ( [
L_{J^r\Xi}\rho] ) = L_{J^s\Xi}R_q ([\rho]).
\end{gather*}
This
operator is uniquely def\/ined and is equal, respectively, to the
following expressions:
\begin{alignat*}{4}
& L_{J^s\Xi}h\rho, \qquad && 0\leq q\leq n, \qquad && s=r+1,&\\
& L_{J^s\Xi} \cI(\rho), \qquad && n+1\leq q\leq P, \qquad && s=2r+1,&\\
& L_{J^s\Xi} \rho, \qquad && q > P+1, \qquad&& s=r,&
\end{alignat*}
where $P$ denotes the corank of the Cartan distribution on $J^r\bY$.

This def\/inition enables us to deal with ordinary
Lie derivatives of forms on $\For^s_q$, thus we can apply the standard Cartan
formula.

We recall a result in \cite[Theorem~III.11]{Kr02}, see
also~\cite{KrMu05}.
\begin{Lemma} \label{Krbek}
Let $\Xi$ be a $\pi$-vertical vector field on
$\bY$ and $\rho$ a differential $q$-form on $J^r\bY$. Then the
following holds true for $i=1,\ldots , q$
\begin{gather*}
J^{r+2}\Xi\rfloor
p_i d p_i \rho = - p_{i-1}d \big(J^{r+1}\Xi \rfloor p_i\rho\big),
\end{gather*}
and
\begin{gather*}
L_{J^{r+2}\Xi}(\pi_{r+2,r+1})^* p_i\rho =
J^{r+2}\Xi\rfloor p_{i+1} d p_i \rho + p_{i}d \big(J^{r+1}\Xi \rfloor
p_i\rho\big).
\end{gather*}
\end{Lemma}

As above, we denote $\omega_0 = dx^1 \land dx^2 \land \cdots \land dx^n$. Next, let us denote
\begin{gather*}
\omega_i = \frac{\partial}{\partial x^i} \rfloor \omega_0 , \qquad
\omega_{ij} = \frac{\partial}{\partial x^j} \rfloor \omega_i.
\end{gather*}

It is well known that in order to obtain a representation of classes of degree $n+1$ in the variational sequence by dynamical forms the following integration formula is used \cite{Kru90}
\begin{gather*}
\ome^{\alp}_{J i}\wed \ome_0 = - d\ome^{\alp}_{J}\wed \ome_i ,
\end{gather*}
and the corresponding representation is obtained by taking the~$p_1$ component obtained by ite\-ra\-ted integrations by parts.
In order to integrate by parts $(p+k)$-forms with $p<n$, we need to generalize as
\begin{gather*}
\gamma_{\alp}^{J [i j] } \omega^{\alp}_{J [ i } \wed \ome_{j]}
= - \gam_{\alp}^{J [i j] } d \ome^{\alp}_J \wed \ome_{i j} .
\end{gather*}
This enables us to generalize results given in \cite{KrMu05}.

\begin{Lemma}\label{integration by part}
Let $\rho \in \For^r_{p+k}$, $1 \leq p \leq n$. Let $p_{k} \rho=\sum\limits^{r}_{|J |=0} \ome^{\alp}_{J} \wed \eta^{J}_\alp$, with $\eta^{J}_\alp$ $(k-1)$-contact $(p+k-1)$-forms.
Then we have the decomposition
\begin{gather*}
p_k\rho= \mathfrak{I}(\rho)+p_{k}dp_{k} \mathfrak{R}(\rho) ,
\end{gather*}
where $\mathfrak{R}(\rho)$ is a local $k$-contact $(p+k-1)$-form such that
\begin{gather*}
 \mathfrak{I}(\rho)+p_{k}dp_{k} \mathfrak{R}(\rho) = \ome^{\alp}\wed \sum^{r}_{ |J |= 0}(-1)^{| J |}d_{J} \eta^{J}_\alp+ \sum^{r}_{| I |=1}d_{I}\big(\ome^{\alp}\wed \zet^{I}_{\alp}\big),
\end{gather*}
with
\begin{gather*}
\zet^{I}_{\alp} = \sum^{r- | I |}_{|J |=0}(-1)^{J} \binom{| I |+ |J |}{|J |}
d_{J} \eta^{J I}_{\alp}.
\end{gather*}
\end{Lemma}

Notice that $d_{J} \eta^{J}_\alp$ are also $(k-1)$-contact $p$-horizontal (i.e., containing the wedge of $p$ of the forms $dx^i$) $(p+k-1)$-forms.
It is therefore well def\/ined a local splitting of $p_k\rho$, $\rho \in \For^r_{p+k}$, $1 \leq p \leq n$ and
$k > 1$.

Of course, $\mathfrak{I}=\cI$ and $\mathfrak{R}=\cR$ in the case $p = n$.

\begin{Example}
Let $p=n-1$.
We can write $\ome^{\alp}\wed \zet^{I}_{\alp}= \chi ^{I l} \wed \ome_l$, where $\chi ^{I l}$ are some local $k$-contact $k$-forms on $J^{2r}\bY$. Therefore
\begin{gather*}
 \sum^{r}_{| I |=1}d_{I}\big(\ome^{\alp}\wed \zet^{I}_{\alp}\big) = \sum^{r}_{| I |=1}d_{I}\chi^{I l}\wed \ome_{l}
 = d_i\sum^{r-1}_{| I |=0}d_{I}\chi^{I li}\wed \ome_{l}=
d_i\sum^{r-1}_{|I |=0}d_{I}\chi^{I li}\wed dx^j \wed \ome_{lj} \\
= (-1)^{(n-1)+k-1}d_i\left[\sum^{r-1}_{|I |=0}(-1)^k d_{I}\chi^{{I}{[lj]}}\wed \ome_{lj}\right]\wed dx^i =
d_H \left[\sum^{r-1}_{| I |=0}(-1)^k d_{I}\chi^{I [lj]}\wed \ome_{lj}\right] .
\end{gather*}
We denote by $\mathfrak{R}(\rho)=\sum\limits^{r-1}_{|I |=0}(-1)^k d_{I}\chi^{I [lj]}\wed \ome_{lj}$ so that
\begin{gather*}
\sum^{r}_{|I |=1}d_{I}\big(\ome^{\alp}\wed \zet^{I}_{\alp}\big)=p_{k}dp_{k}\mathfrak{R}(\rho) .
\end{gather*}
\end{Example}

In the following we shall use such a splitting in order to split the vertical dif\/ferential of a~$p$-density and def\/ine a corresponding `momentum form'.

\begin{Theorem}\label{1}
Let $0\leq q\leq n$, $\Xi$ a $\pi$-projectable vector
field on $\bY$ and $\rho$ a differential $q$-form on $J^r\bY$. We
have
\begin{enumerate}\itemsep=0pt
\item[$1.$] For $0\leq q\leq n-1$
\begin{gather*}
 \hat{R}_q(\cL_{ J^r \Xi}[\rho] )\equiv L_{ J^{r+1}\Xi}h\rho \\
\hphantom{\hat{R}_q(\cL_{ J^r \Xi}[\rho] )}{} = J^{r+2} \Xi_H \rfloor d_H h\rho + J^{r+2}\Xi_{V} \rfloor \mathfrak{I }(d\rho)
 +d_H\big(J^{r+1}\Xi_V\rfloor \tilde{p}_{d_V h\rho} +J^{r+1} \Xi_H \rfloor h\rho\big),
\end{gather*}
where a `generalized momentum' is defined by
\begin{gather} \label{mom}
\tilde{p}_{d_V h\rho} = - p_1\mathfrak{R}(d\rho) .
\end{gather}

\item[$2.$] For $q=n$
\begin{gather*}
 \hat{R}_n(\cL_{ J^r \Xi}[\rho] ) \equiv L_{ J^{r+1}\Xi}h\rho = J^{r+2} \Xi_V \rfloor E_n
(h\rho) + d_H \big(J^{r+1}\Xi_V\rfloor p_{d_V h\rho} + J^{r+1}\Xi_H \rfloor h\rho\big),
\end{gather*}
where
\begin{gather} \label{momn}
p_{d_V h\rho} = - p_1 \cR(d\rho).
\end{gather}
\end{enumerate}
\end{Theorem}

Notice, that the latter formula, written as
\begin{gather*}
L_{ J^{r+1}\Xi} \lambda = J^{r+2} \Xi_V \rfloor E_\lambda + d_H \phi_{\lambda, \Xi} ,
\end{gather*}
where
\begin{gather*}
\phi_{\lambda, \Xi} = J^{r+1}\Xi_V\rfloor p_{d_V h\rho} + J^{r+1}\Xi_H \rfloor h\rho ,
\end{gather*}
is obviously the inf\/initesimal f\/irst variation formula (\ref{1vf}). We observe that $p_{d_V h\rho} $ are momentum forms, and $\phi_{\lambda, \Xi}$ is the Noether current related with $\Xi$ and $\lambda$.

\begin{proof}
For all cases
by the Cartan formula
\begin{gather*}
L_{ J^{r+1}\Xi}h\rho = J^{r+1}\Xi \rfloor d h\rho +
d\big( J^{r+1}\Xi \rfloor h\rho\big)
\end{gather*}
we have
\begin{gather*}
 L_{ J^{r+1}\Xi} h\rho= \big(J^{r+2}\Xi_H + J^{r+2}\Xi_V\big)\rfloor (d_V h\rho + d_H h\rho) + (d_V+d_H)\big(\big(J^{r+1}\Xi_H + J^{r+1}\Xi_V\big)\rfloor h\rho\big) \\
\hphantom{L_{ J^{r+1}\Xi} h\rho}{}
 = d_{H}\big( J^{r+1}\Xi_{H} \rfloor h\rho\big)+J^{r+2}\Xi_{H} \rfloor d_{H}h\rho
+ J^{r+2}\Xi_{V} \rfloor d_{V}h\rho.
\end{gather*}
This expression can be further characterized in more detail.

Let $0 \leq q \leq n-1$. Since $d_V h\rho = p_1d\rho = \mathfrak{I }(d\rho) +
p_1 d p_1\mathfrak{R}(d\rho)$, and since by Lemma~\ref{Krbek},
\begin{gather*}
J^{r+2}\Xi_V\rfloor p_1dp_1 \mathfrak{R}(d\rho) = - p_0
d\big( J^{r+1}\Xi_V \rfloor p_1 \mathfrak{R}(d\rho) \big) = - d_H\big( J^{r+1}\Xi_V \rfloor p_1 \mathfrak{R}(d\rho) \big),
\end{gather*}
we can def\/ine a {\em momentum form associated with a density of degree $<n$} by (\ref{mom}) and obtain
\begin{gather*}
 L_{ J^{r+1}\Xi} h\rho=
d_{H}\big( J^{r+1}\Xi_{H} \rfloor h\rho\big)+J^{r+2}\Xi_{H} \rfloor d_{H}h\rho
+ J^{r+2}\Xi_{V} \rfloor (\mathfrak{I }(d\rho) +
p_1 d p_1\mathfrak{R}(d\rho)) \\
\hphantom{L_{ J^{r+1}\Xi} h\rho}{} = J^{r+2}\Xi_{V} \rfloor \mathfrak{I }(d\rho) + J^{r+2} \Xi_H \rfloor d_H
h\rho +d_H(J^{r+1}\Xi_V\rfloor \tilde{p}_{d_V h\rho} + J^{r+1}\Xi_H \rfloor
h\rho).
\end{gather*}

Let $q=n$. Then $d_Hh\rho = 0$, hence $J^{r+2} \Xi_H \rfloor d_{H}h\rho=0$. Next, again since $d_V h\rho = p_1d\rho = \cI (d\rho) + p_1dp_1\cR(d\rho)$ and since by Lemma~\ref{Krbek}, $J^{r+2}\Xi_V\rfloor p_1dp_1 \cR(d\rho) = - d_H( J^{r+1}\Xi_V \rfloor p_1 \cR(d\rho) )$, we def\/ine a local momentum form by~(\ref{momn}), and obtain
\begin{gather*}
 L_{ J^{r+1}\Xi} h\rho = J^{r+2}\Xi_{V} \rfloor (\cI (d\rho) - d_H\big( J^{r+1}\Xi_V \rfloor p_1 \cR(d\rho) \big) + d_{H}( J^{r+1}\Xi_{H} \rfloor h\rho)
\\
\hphantom{L_{ J^{r+1}\Xi} h\rho}{}
= J^{r+2}\Xi_{V} \rfloor \cI (d\rho) + d_{H}\big( J^{r+1}\Xi_{V} \rfloor p_{d_{V}h\rho}+ J^{r+1}\Xi_H \rfloor h\rho\big)\\
\hphantom{L_{ J^{r+1}\Xi} h\rho}{}
= J^{r+2}\Xi_{V} \rfloor \cI (d\rho) + d_{H} \phi_{\lambda, \Xi} .
\end{gather*}
However, $\phi_{\lambda, \Xi}$ is a horizontal form, and in Takens representation, $\cI (d\rho) = E_n(h\rho)$
and $d_{H} \phi_{\lambda, \Xi} = hd\phi_{\lambda, \Xi} = E_{n-1}( \phi_{\lambda, \Xi})$,
proving the theorem.
\end{proof}

\begin{Theorem}\label{HigherLieDer0}
Let $q= n + k$, $ k \geq 1$. Let $ \Xi$
be a $\pi$-vertical vector field on $\bY$. We have
\begin{gather*}
\hat{R}_{n+k}(\cL_{ J^r \Xi}[\rho]) = J^{s+1}\Xi \rfloor p_{k+1}
d(\cI(\rho) ) + p_k d \big( J^{s}\Xi \rfloor \cI( \rho)\big) .
\end{gather*}
\end{Theorem}

\begin{proof}
For simplicity of notation, in what follows we omit the projections (so we write $p_k \rho$ instead of $\pi_{s,r+1}^*p_k \rho$, etc.).

Since $R_{n+k}([\rho])= \cI(\rho)= p_k \rho - p_k dp_k \cR(\rho)$, we have
\begin{gather*}
 \hat{R}_{n+k}(\cL_{ J^r\Xi}[\rho])= L_{ J^{s} \Xi}\cI(\rho) =
L_{ J^{s}\Xi} p_k \rho - L_{ J^{s}\Xi}p_k dp_k\cR(\rho) .
\end{gather*}

By applying Lemma \ref{Krbek} to $dp_k\cR(\rho)$ we have
\begin{gather*}
L_{ J^{s+1} \Xi} p_k dp_k\cR(\rho)= J^{s+1}
\Xi \rfloor p_{k+1} d p_{k} d p_{k} \cR(\rho) + p_k d\big( J^{s}\Xi
\rfloor p_k d p_k \cR(\rho)\big) .
\end{gather*}
Simple manipulations show that
\begin{gather*}
p_k d( J^{s} \Xi \rfloor p_k d p_k \cR(\rho)) = - J^{s+1} \Xi \rfloor p_{k+1} d p_{k} d p_{k} \cR(\rho),
\end{gather*}
so that
\begin{gather*}
L_{ J^{s+1} \Xi} p_k dp_k\cR(\rho)=0 .
\end{gather*}
On the other hand
\begin{gather*}
 L_{ J^{s}\Xi} p_k \rho= J^{s+1}\Xi \rfloor p_{k+1} d p_k \rho +
 p_k d \big( J^{s}\Xi\rfloor p_k \rho\big),
\end{gather*}
but, since $p_k\rho = \cI(\rho) +p_k dp_k\cR(\rho)$, by substituting
we get
\begin{gather*}
 J^{s+1}\Xi \rfloor p_{k+1} d p_k \rho= J^{s+1}\Xi \rfloor p_{k+1}d(\cI(\rho) +p_k dp_k\cR(\rho)) \\
\hphantom{J^{s+1}\Xi \rfloor p_{k+1} d p_k \rho}{}= J^{s+1}\Xi \rfloor p_{k+1} d(\cI(\rho) ) + J^{s+1}\Xi \rfloor p_{k+1} d(p_k dp_k\cR(\rho)) \\
\hphantom{J^{s+1}\Xi \rfloor p_{k+1} d p_k \rho}{}= J^{s+1}\Xi \rfloor p_{k+1} d(\cI(\rho) ) - p_k d \big(J^{s}\Xi \rfloor p_k dp_k\cR(\rho)\big).
\end{gather*}
Again by the same substitution in $p_k d ( J^{s}\Xi\rfloor p_k \rho)$, we get
\begin{gather*}
 p_k d \big( J^{s}\Xi \rfloor \cI( \rho)\big) + p_k d \big( J^{s}\Xi\rfloor p_k dp_k\cR( \rho)\big).
\end{gather*}
Summing up we f\/inally obtain
\begin{gather*}
 \hat{R}_{n+k}(\cL_{ J^{r}\Xi}[\rho])= L_{ J^{s}\Xi} p_k \rho= J^{s+1}\Xi \rfloor p_{k+1} d(\cI(\rho) ) +
p_k d ( J^{s}\Xi \rfloor \cI( \rho)) .\tag*{\qed}
\end{gather*}
\renewcommand{\qed}{}
\end{proof}

By the above theorem we get a generalization of the famous {\em Noether equation} and {\em Noether--Bessel-Hagen equation} (i.e., equation for symmetries of a~Lagrangian and of the Euler--Lagrange form) to a~{\em source form of any degree}. Indeed, a projectable vector f\/ield~$\Xi$ on~$\bY$ is a symmetry of $\cI(\rho)$, i.e., $L_{ J^s\Xi} \cI(\rho)= 0$ if and only if $L_{ J^{s+1}\Xi_V} \cI(\rho)= 0$, that is if and only if
\begin{gather*}
J^{s+1}\Xi _V \rfloor p_{k+1} d(\cI(\rho) ) + p_k d \big( J^{s}\Xi _V\rfloor \cI(\rho)\big) = 0 .
\end{gather*}

Theorems~\ref{1} and \ref{HigherLieDer0} are the desired generalization of the inf\/initesimal f\/irst variation formula~(\ref{1vf}) to forms of any degree. The `boundary term' on the right-hand side in Theorem~\ref{1} contains `generalized momenta', and provides a lower-degree analog of Noether theorem in the same way as in the classical case $q = n$ $(k = 0)$.

\begin{Lemma}
Let $q= n + k$, $ k \geq 1$, and let $\Xi$ be a $\pi$-projectable vector field on~$\bY$. We have
for every class~$[\rho]$ of degree~$n+k$
\begin{gather*}
\cL_{ J^{r+1} \Xi_{H} } [\rho] = 0.
\end{gather*}
\end{Lemma}

\begin{proof}
In fact, $ J^{s+1} \Xi_H \rfloor d _H \cI(\rho) = 0$ because $d_H \cI(\rho)$ is
the horizontal dif\/ferential of an $n$-horizontal form. On the
other hand since $d_V \cI(\rho) = p_{k+1 }d (\cI(\rho))$ and $p_{k+1
}d \circ p_{k+1 }d = 0$, we have $ J^{s+1} \Xi_H \rfloor d_V \cI(\rho)$ $=$ $J^{s+1} \Xi_H \rfloor p_{k+1 }d (\cI(\rho))$ (a $(k+1)$-contact and
$(n-1)$-horizontal $(n+k)$-form); therefore $\hat{R}_{n+k} (\cL_{ J^{r+1}
\Xi_{H} } [\rho] ) $ $=$ $J^{s+1} \Xi_H \rfloor p_{k+1 }d (\cI(\rho))) +
 p_k d(J^{s+1} \Xi_H \rfloor \cI(\rho)) \in \Thd_{n+k}^{s+1}$.
\end{proof}

Since $p_{k+1} d(\cI(\rho) )= \cI(d(\cI(\rho) )) =R_{q+1} (\cE_{q}([\rho] ))$, taking into account Theorem~\ref{ThRd}, we f\/inally obtain a `quotient' Cartan formula for the variational Lie derivative of $n+k$ classes, with $k\geq 1$.
\begin{Theorem}\label{HigherLieDer}
Let $ k \geq 1$. Let $\Xi$ be a $\pi$-projectable vector field on $\bY$.
For every class $[\rho] \in \Lambda^r_{n+k} / \Theta^r_{n+k}$ we have $($up to pull-backs$)$
\begin{gather*}
\cL_{ J^r \Xi}[\rho] = J^{r+1}\Xi_V \rfloor \cE_{n+k}([\rho] ) +
\cE_{n+k-1}\big( J^{r+1} \Xi_V \rfloor [\rho] \big).
\end{gather*}
\end{Theorem}

\begin{proof}
Since $\Xi$ is assumed to be a $\pi$-projectable vector f\/ield on~$\bY$, up to projections,
\begin{gather*}
\hat{R}_{n+k}(\cL_{ J^{r}\Xi}[\rho]) = \hat{R}_{n+k}(\cL_{ J^{r+1}\Xi_V} [\rho]) + \hat{R}_{n+k}(\cL_{ J^{r+1}\Xi_H }[\rho])
\end{gather*}
holds true and the result follows from the above lemma and theorem (the latter applied for the vertical part of~$J^r \Xi$).
\end{proof}

\section{Some applications}\label{section5}

\subsection{Helmholtz forms in mechanics} \label{Hel}

As mentioned above, for $k \geq 2$, a class $[\rho] \in \cV^r_{n+k}$ in the variational sequence can be represented by dif\/ferent source forms. We shall illustrate this feature on the image $\cE_{n+1}([\varepsilon]) = [d \varepsilon] \in \cV^2_{n+2}$, where $\varepsilon$ is a second-order dynamical form, in the case when $n = 1$ (mechanics).

Let $n = \dim X = 1$; as above, denote local f\/ibred coordinates on $\bY$ by $(t,q^\sigma)$ and the associated coordinates on $J^3\bY$ by $(t,q^\sigma,\dot q^\sigma, \ddot q^\sigma, \dddot q^\sigma)$, and let
\begin{gather*}
\omega^\sigma = dq^\sigma - \dot q^\sigma dt,
\qquad
\dot \omega^\sigma = d \dot q^\sigma - \ddot q^\sigma dt,
\qquad
\ddot \omega^\sigma = d \ddot q^\sigma - \dddot q^\sigma dt.
\end{gather*}
Consider a dynamical form $\varepsilon \in \Lambda^2_2$, $\varepsilon = E_\sigma \omega^\sigma \land dt$. Then $[d\varepsilon]$ is represented by the following canonical source form (called {\em canonical Helmholtz form}), obtained by a straightforward calculation:
\begin{gather*}
H_{\varepsilon} = \cI(d \varepsilon) = \frac{1}{2} \left( \left(\frac{\der E_\sigma}{\der q^\nu} - \frac{\der E_\nu}{\der q^\sigma} - \frac{1}{2} \frac{d}{dt} \left(\frac{\der E_\sigma}{\der \dot q^\nu} - \frac{\der E_\nu}{\der \dot q^\sigma}\right) +
\frac{1}{2} \frac{d^2}{dt^2} \left(\frac{\der E_\sigma}{\der \ddot q^\nu} - \frac{\der E_\nu}{\der \ddot q^\sigma}\right) \right) \omega^\nu\right.
\\
\left.\hphantom{H_{\varepsilon} =}{}
+ \left(\frac{\der E_\sigma}{\der \dot q^\nu} + \frac{\der E_\nu}{\der \dot q^\sigma} - 2 \frac{d}{dt} \frac{\der E_\nu}{\der \ddot q^\sigma}\right) \dot \omega^\nu +
\left(\frac{\der E_\sigma}{\der \ddot q^\nu} - \frac{\der E_\nu}{\der \ddot q^\sigma} \right) \ddot \omega^\nu
\right) \land \omega^\sigma \land dt .
\end{gather*}
By construction $H_\varepsilon \in \Lambda^5_3$, however, it is projectable to $J^4\bY$.

From the exactness of the representation sequence we get that the condition $E_2(\varepsilon) = H_\varepsilon = 0$ is necessary and suf\/f\/icient for existence of a local Lagrangian $\lambda$ such that $\varepsilon = E_1(\lambda) = E_\lambda$. Expressed in terms of the components of $H_\varepsilon$ the condition $H_\varepsilon = 0$ gives the celebrated {\em Helmholtz conditions} of local variationality for dynamical forms, in the context of the variational sequence called also {\em source constraints}.

The canonical Helmholtz form is not the unique source form providing Helmholtz conditions. Another distinguished Helmholtz form is~\cite{Kru04}
\begin{gather*}
\bar H_\varepsilon = \frac{1}{2} \left( \left(\frac{\der E_\sigma}{\der q^\nu} - \frac{\der E_\nu}{\der q^\sigma} - \frac{1}{2} \frac{d}{dt} \left(\frac{\der E_\sigma}{\der \dot q^\nu} - \frac{\der E_\nu}{\der \dot q^\sigma}\right) \right) \omega^\nu\right.
\\
\left. \hphantom{\bar H_\varepsilon =}{}
+ \left(\frac{\der E_\sigma}{\der \dot q^\nu} + \frac{\der E_\nu}{\der \dot q^\sigma} - \frac{d}{dt} \left(\frac{\der E_\sigma}{\der \ddot q^\nu} + \frac{\der E_\nu}{\der \ddot q^\sigma} \right) \right) \dot \omega^\nu +
\left(\frac{\der E_\sigma}{\der \ddot q^\nu} - \frac{\der E_\nu}{\der \ddot q^\sigma} \right) \ddot \omega^\nu
\right) \land \omega^\sigma \land dt .
\end{gather*}
$\bar H_\varepsilon$ is of order $3$, and is, indeed, equivalent with $H_\varepsilon$. Note that it provides Helmholtz conditions in an equivalent and more simple form. The relationship between the canonical and the ``reduced'' Helmholtz form is as follows~\cite{KP07}:
\begin{gather*}
\bar H_\varepsilon = H_\varepsilon + p_2d\eta,
\end{gather*}
where
\begin{gather*}
\eta = - \frac{1}{4} \frac{d}{dt}\left(\frac{\der E_\sigma}{\der \ddot q^\nu} - \frac{\der E_\nu}{\der \ddot q^\sigma} \right) \omega^\nu \land \omega^\sigma .
\end{gather*}

It is worth note that in the most frequent case when the components of $\varepsilon$ are af\/f\/ine in the second derivatives with {\em symmetric coefficients}, i.e.,
\begin{gather*}
E_\sigma = A_\sigma + B_{\sigma \nu} \ddot q^\nu, \qquad B = B^T ,
\end{gather*}
the canonical Helmholtz form $H_\varepsilon$ is projectable to $J^3\bY$ and coincides with the ``reduced'' form~$\bar H_\varepsilon$.

Asking about Helmholtz forms of {\em variational order~$1$}, that is, coming from the variational sequence of order~$1$, we get the following result: {\em The Helmholtz form $H_\varepsilon$ has variational order~$1$ if and only if it is projectable onto~$J^2\bY$, and this is if and only if the following conditions are satisfied:}
\begin{gather} \label{HB}
B = B^T , \qquad
\frac{\partial B_{\sigma \nu}}{\partial \dot q^\rho} = \frac{ \partial B_{\sigma \rho}}{\partial \dot q^\nu} .
\end{gather}
Note that:
\begin{itemize}\itemsep=0pt
\item Conditions (\ref{HB}) are a part of the Helmholtz conditions, i.e., they are necessary for $\varepsilon$ be locally variational.

\item Conditions (\ref{HB}) are integrability conditions, guaranteeing existence of a function $L$ such that
\begin{gather*}
B_{\sigma \nu} = - \frac{\partial^2L}{\partial \dot q^\sigma \partial \dot q^\nu}
\end{gather*}
(the minus sign is chosen just for convenience). $L$ is not generally a Lagrangian for $\varepsilon$ (indeed, the form $\varepsilon$ need not be variational, $H_\varepsilon$ need not be equal to~$0$). However, with help of the Euler--Lagrange form $E_\lambda$ of $\lambda = L dt$ we get $\varepsilon = E_\lambda - F$, so that in coordinates equations represented by $\varepsilon$ take the form
\begin{gather*}
\frac{\partial L}{\partial y^\sigma} - \frac{d}{dt} \frac{\partial L}{\partial \dot y^\sigma} = F_\sigma .
\end{gather*}
We can conclude that the requirement that~$H_\varepsilon$ has variational order one means that the corresponding equations are so-called {\em Lagrange equations of the second kind}.

\item It can be easily shown that the above equations are variational, i.e., $H_\varepsilon = 0$ if and only if the force $F$ has the form of a Lorentz-like force (as, e.g., Lorentz force, Coriolis force, etc.).
\end{itemize}

\subsection{Inverse problems in the variational sequence}

Perhaps the best known problems which are solved with help of variational sequences and bicomplexes are the {\em local and global inverse problem of the calculus of variations}, and the problem of {\em local and global triviality of Lagrangians}. Let us show how these problems appear and are solved with help of the variational sequence and the representation sequences.

$\bullet$ {\em Variationally trivial Lagrangians}, also called {\em null-Lagrangians} are Lagrangians which give rise to identically zero Euler--Lagrange expressions (meaning that the Euler--Lagrange form $E_\lambda$ identically vanishes). This problem concerns the classes of degree~$n-1$, $n$ and~$n+1$ and the corresponding morphisms ${\cal E}_{n-1}\colon {\cal V}^r_{n-1} \to {\cal V}^r_n$, and
${\cal E}_{n}\colon {\cal V}^r_{n} \to {\cal V}^r_{n+1}$ (the Euler--Lagrange mapping).

Consider the variational sequence $0 \to \R_{\bY} \to {\cal V}^r_*$, and assume that a class $[\rho] \in {\cal V}^r_{n}$ satisf\/ies ${\cal E}_n([\rho]) = 0$. Then by exactness of the variational sequence there exists a class $[\eta] \in {\cal V}^r_{n-1}$ such that ${\cal E}_{n-1}([\eta]) = [\rho] = [d\eta]$. The condition ${\cal E}_n([\rho]) = 0$ is the local variational triviality condition and determines the structure of Lagrangians whose
Euler--Lagrange forms vanish identically: every variationally trivial Lagrangian locally is a closed $n$-form $d\eta$ modulo contact forms, in other words, $\lambda = hd\eta = d_H \eta$. In coordinates, if $\eta= f^i \omega_i + \cdots$ then $L = d_i f^i$. If, in addition, $H^{n}_{\rm dR}\bY= \{0\}$ then $\eta$ may be chosen globally def\/ined on
$J^r\bY$.

In Takens representation the same result follows by the following arguments: If $[\rho] \in {\cal V}^r_{n}$ satisf\/ies ${\cal E}_n([\rho]) = 0$ then $\lambda = R_n([\rho]) = h\rho$ satisf\/ies the constraint on local variational triviality $E_n(h\rho) = 0$, which means that locally $h\rho = R_{n-1}[(d\eta]) = hd\eta$. A corresponding local Lagrangian is then $\lambda = R_n({[\eta]}) = h\eta$. Condition $H^{n}_{\rm dR}\bY= \{0\}$ then guarantees existence of a global $\eta$, hence of a global Lagrangian~$h\eta$.

In Lepage representation, the local triviality condition ${\cal E}_n([\rho]) = 0$ means that (any) Lepage $n$-form
$\rho \in [\rho] $ (Lepage equivalent of $\lambda = h\rho$) is closed. Hence locally $\rho = d\eta$, i.e., $\lambda = hd\eta =( d_H \eta =) hdh\eta = d_if^i \omega_0$, and $\eta$, resp.\ the relevant part $h\eta$, is constructed from $\rho$ with help of the Poincar\'e contact homotopy operator as
$\eta = {\cal A} \rho$, resp.\ $h\eta = f^i \omega_i = {\cal A} \lambda$.

$\bullet$ By the {\em inverse problem of the calculus of variations}, here one understands the question about local and global variationality of a dynamical form $\varepsilon$. This means the problem to determine constraints on variationality (Helmholtz conditions), and for a locally variational form, to construct a local Lagrangian. This problem concerns the classes of degree $n$, $n+1$ and $n+2$ in the variational sequence, and the corresponding morphisms ${\cal E}_{n}\colon {\cal V}^r_{n} \to {\cal V}^r_{n+1}$ (the Euler--Lagrange morphism) and ${\cal E}_{n+1}\colon {\cal V}^r_{n+1} \to {\cal V}^r_{n+2}$ (the Helmholtz morphism).

Consider the variational sequence $0 \to \R_{\bY} \to {\cal V}^r_*$, and assume that a class $[\rho] \in {\cal V}^r_{n+1}$ satisf\/ies ${\cal E}_{n+1}([\rho]) = 0$. Then by exactness of the variational sequence there exists a class $[\eta] \in {\cal V}^r_{n}$ such that ${\cal E}_{n}([\eta]) = [\rho] = [d\eta]$. The condition ${\cal E}_{n+1}([\rho]) = 0$ is the local variationality condition and determines the structure of locally variational forms (giving rise to variational equations): Every locally variational form $\varepsilon$ of {\em variational order} $r$ is a closed $(n+1)$-form $d\eta$ modulo strongly contact forms and exterior derivatives of contact forms of order $r$. Remarkably, this factorization procedure determines the structure of locally variational forms of the {\em variational order} $r$, i.e., produces constraints for a locally variational form of order $s$ to come from a local Lagrangian of order $r$ ({\em order reduction constraints}). If, in addition, $H^{n+1}_{\rm dR}\bY= \{0\}$ then $\varepsilon$ is globally variational.

In Takens representation, if $[\rho] \in {\cal V}^r_{n+1}$ satisf\/ies ${\cal E}_{n+1}([\rho]) = 0$ then the dynamical form $\varepsilon = R_{n+1}([\rho] )= {\cal I}\rho$ satisf\/ies $E_{n+1}(\varepsilon) = 0$ ({\em Helmholtz conditions}). Hence we can say:

{\em A dynamical form $\varepsilon$ is locally variational if and only if its Helmholtz form $H_\varepsilon = E_{n+1} (\varepsilon)$ vanishes.}

Then locally $\varepsilon = R_{n}[(d\eta]) = {\cal I}d\eta$, providing a Lagrangian for $\varepsilon$.

In Lepage representation, the local variationality condition ${\cal E}_{n+1}([\rho]) = 0$ means that (any) Lepage form $\rho \in [\rho] $ (Lepage equivalent of $\varepsilon = {\cal I} \rho = p_1d\rho$) is closed. Hence locally $\rho = d\eta$, i.e., $\varepsilon = p_1 \rho = p_1d\eta$, and $\eta$, respectively, a Lagrangian $\lambda$ for $\varepsilon$ is constructed from $\rho$ with help of the Poincar\'e contact homotopy operator as $\eta = A\rho$, respectively, $\lambda = {\cal A}\varepsilon$ ({\em Tonti Lagrangian}).

In \cite{KrMa10,Ma09} the same questions concerning $(n+2)$-forms have been studied for the case of mechanics ($\dim \bX = 1$). Among others, explicit formulae for constraints on local `variational triviality' of Helmholtz-like forms, $E_{n+2}(\rho) = 0$, generalizing Helmholtz conditions to forms of degree $n+2$, as well as Tonti-like formulas for the corresponding (nonvariational) dynamical forms, have been found.

In the same way, the inverse problems, the equivalence problems, and the corresponding order reduction problems in the local and global form extend to {\em each column of the variational sequence}, and are solved by the exactness of the sequence. A practical question, however, is to f\/ind {\em explicit formulae for source constraints}, i.e., constraints on a source form (of any degree) to be variationally trivial (i.e., to belong to the kernel of the corresponding morphism in the variational sequence). The problem in its full generality is solved easily with help of the Lepage representation of the variational sequence, where it is transformed to the Poincar\'e lemma:

\begin{Theorem}[general inverse variational problem] \label{GHC}
Let $\varepsilon$ be a source form of degree $n+k$, $k \geq 1$.
The following conditions are equivalent:
\begin{enumerate}\itemsep=0pt
\item[$(1)$] $\varepsilon$ is locally variationally trivial, i.e., belongs to $\Ker E_{n+k}$,

\item[$(2)$] $\varepsilon$ has a closed Lepage equivalent,

\item[$(3)$] every Lepage equivalent $\alpha$ of $\varepsilon$ satisfies $p_{k+1}d\alpha = 0$.
\end{enumerate}

If $\varepsilon$ is locally variationally trivial then a corresponding local primitive source $(n+k-1)$-form~$\eta$ $($satisfying $E_{n+k-1}(\eta) = \varepsilon)$ is $\eta = {\cal A} \varepsilon$.

If $H_{\rm dR}^{n+k} \bY = \{0\}$ then $\varepsilon$ is globally variationally trivial, with a global primitive of order $\leq 2r+1$, where $r$ is the variational order of $\varepsilon$.
\end{Theorem}

\begin{proof}
Let $\varepsilon \in \Ker E_{n+k}$. This means that $\varepsilon = {\cal I}(\rho)$ where $[\rho] \in \Ker {\cal E}_{n+k}$. Hence $[\rho] = [d\nu] = {\cal E}_{n+k-1}([\nu])$, and we have $\{\Lep_{n+k}(\varepsilon)\} = \tilde R_{n+k}(\rho) = \tilde R_{n+k}(d\nu) = \{\Lep_{n+k} (R_{n+k}(d\nu) ) \} = \{0\}$, proving~(2).

Next, assume $d\{\Lep_{n+k}(\varepsilon)\} = \{0\}$. This means that every Lepage equivalent $\alpha$ of $\varepsilon$ satisf\/ies $d\alpha = d\mu$ where $\mu$ is at least $(k+2)$-contact. Hence $p_{k+1}d\alpha = 0$, as desired.

Finally, assume (3). Since $\alpha = \theta + d\nu + \mu$ where $\mu$ is at least $(k+2)$-contact, we can see that $p_{k+1} d\theta = 0$. Hence $\{\alpha\} = \tilde R({[\rho]})$ where $\{d\alpha\} = \tilde R({[d\rho]}) = \{0\}$.
This means that $[d\rho] = [0]$, i.e., $[\rho] \in \Ker {\cal E}_{n+k}$, and $\varepsilon = {\cal I}(\rho) \in \Ker E_{n+k}$.

The remaining assertions follow immediately from the properties of ${\cal A}$, ${\cal I}$, and the variational sequence.
\end{proof}

Condition~(3) in Theorem~\ref{GHC} can be also expressed by means of the Cartan $(n+k)$-form $\theta_\varepsilon$ of the source $(n+k)$-form $\varepsilon$ as
\begin{gather*} %\label{GHC2}
p_{k+1}d\theta_\varepsilon = 0.
\end{gather*}
It represents necessary and suf\/f\/icient {\em source constraints} (constraints on a source form to be locally variationally trivial) for source forms of any degree $q > n$, and of any order. For $q = n+1$ this is exactly an intrinsic form of the celebrated {\em Helmholtz conditions} mentioned above.

\subsection{Bosonic string}

A bosonic string is described as a minimal immersion of a $2$-dimensional surface in a $4$-dimen\-sio\-nal pseudo-Riemannian manifold~$(M,g)$. Hence the `motion equation' is the minimal surface equation arising from the area functional (known as {\em Nambu--Goto action}~\cite{Goto}).

The corresponding f\/ibred manifold is $\pi\colon \bY \to \bX$ where $\bX = \R^2$ and $\bY = \R^2 \times M$, so that $\dim \bY = 6$. Base coordinates $\tau^0$ and $\tau^1$ represent the time evolution parameter and the position on the string, respectively, f\/ibred coordinates $(x^\mu)$, $\mu=0,1,2,3$, are local coordinates on~$M$. We denote by $(\tau^i,x^\mu,x^\mu_i)$, where $i = 0,1$ and $\mu=0,1,2,3$, associated local coordinates on $J^1\bY$, and put $\omega^\mu = dx^\mu - x^\mu_i d\tau^i$.
If $\iota\colon \R^2 \to \Sigma \subset M$ is an immersion of $\R^2$ to~$M$, and $g=g_{\mu\nu}dx^\mu\otimes dx^\nu$ is a (symmetric regular pseudo-Riemannian) metric on $M$, we have the induced metric on $\Sigma$
\begin{gather*}
\iota^*g \equiv h_{ij} d \tau^i\otimes d\tau^j = (g_{\mu\nu} \circ \iota) \frac{\partial x^\mu}{\partial \tau^i} \frac{\partial x^\nu}{\partial \tau^j} d\tau^i \otimes d\tau^j
\end{gather*}
and the area element
\begin{gather*}
d\Sigma =\sqrt{-\det (h_{ij})}d\tau^0\wedge d\tau^1 = \sqrt{-\det \left( (g_{\mu\nu} \circ \iota) \frac{\partial x^\mu}{\partial \tau^i} \frac{\partial x^\nu}{\partial \tau^j} \right) } d\tau^0\wedge d\tau^1,
\end{gather*}
hence the Nambu--Goto action reads
(in units with light speed $c=1$)
\begin{gather*}
S = -T \int_\Sigma d\Sigma,
\end{gather*}
where $T$ is the (constant) string tension. In the f\/ibred setting, we write the action as
$S = \int_\Omega J^1\gamma^*\lambda$, where~$\Omega$ is a domain in~$\R^2$, $\gamma = (\id_{\R^2}, \iota)$ is a~local section of~$\pi$, and the Lagrangian (now considered as a horizontal form on~$J^1\bY$) is
\begin{gather*}
\lambda = -T \sqrt{-\det \big(g_{\mu\nu} x^\mu_i x^\nu_j \big) } d\tau^0\wedge d\tau^1
\end{gather*}
(note that the action is reparametrization invariant).
Let us denote
\begin{gather*}
L\big(\tau^i,x^\mu,x^\mu_i\big) = -T\sqrt{-D} = -T \sqrt{-\det \big(g_{\mu\nu} x^\mu_i x^\nu_j \big) }
 = -T\sqrt{(g_{\alpha\beta} g_{\mu\nu}-
g_{\alpha\mu} g_{\beta\nu}) x^\alpha_0x^\mu_0x^\beta_1x^\nu_1}.
\end{gather*}
The Cartan form of the Lagrangian $\lambda$ is the $2$-form
\begin{gather*}
 \theta_\lambda = L d\tau^0\wedge d\tau^1+
 \frac{\partial L}{\partial x^\mu_0} \omega^\mu\wedge d \tau_1 -
\frac{\partial L}{\partial x^\mu_1} \omega^\mu\wedge d \tau_0 \\
\hphantom{\theta_\lambda}{}
= T\sqrt{-D} d\tau^0\wedge d\tau^1
 + \frac{T}{\sqrt{- D}} (g_{\alpha\mu} g_{\beta\nu}-g_{\alpha\beta} g_{\mu\nu})
d x^\mu\wedge \big(x^\alpha_0 x^\beta_0 x^\nu_1 d\tau^0-x^\alpha_1 x^\beta_1 x^\nu_0 d\tau^1\big).
\end{gather*}
We introduce momenta
\begin{gather*}
 p_\mu^0 = \frac{\partial L}{\partial
x^\mu_0}=-\frac{T}{\sqrt{-D}}(g_{\alpha\beta}g_{\mu\nu}-g_{\alpha\mu}g_{\beta\nu})
x^\alpha_0x^\beta_1x^\nu_1, \\
 p_\mu^1= \frac{\partial L}{\partial
x^\mu_1}=-\frac{T}{\sqrt{-D}}(g_{\alpha\beta}g_{\mu\nu}-g_{\alpha\mu}g_{\beta\nu})
x^\alpha_1x^\beta_0x^\nu_0,
\end{gather*}
and notice that $p^0_\mu x^\mu_0 = p^1_\mu x^\mu_1=-T\sqrt{-D}$, and
$p^0_\mu x^\mu_1 = p^1_\mu x^\mu_0 = 0$.

As we have seen, a projectable vector f\/ield $\Xi$ on $\bY$ is a symmetry of $\lambda$ if $\Xi$ satisf\/ies
the Noether equation $L_{J^1\Xi}\lambda=0$. The corresponding Noether current is then the horizontal $1$-form $\Psi=h (J^1\Xi\rfloor\theta_\lambda)$, and, by Noether theorem, it is closed along critical sections (on every extremal surface $\Sigma$): $d(\Psi \circ J^1\gamma) = 0$.

To make explicit computations, denote
\begin{gather*}
\Xi=\xi^i\frac{\partial}{\partial
\tau^i}+\Xi^\mu\frac{\partial}{\partial x^\alpha} , \qquad
J^1\Xi=\xi^i\frac{\partial}{\partial
\tau^i}+\Xi^\mu\frac{\partial}{\partial
x^\mu}+\left(d_i\Xi^\mu-x^\mu_j\frac{\partial\xi^j}{\partial \tau^i}
\right)\frac{\partial}{\partial x^\mu_i} .
\end{gather*}
Contracting $\theta_\lambda$ and taking the horizontal part we obtain
\begin{gather*}
\Psi = \Xi^\mu \big(p^0_\mu d\tau^1 -p^1_\mu d\tau^0\big) .
\end{gather*}
This means that to every symmetry $\Xi$ of $\lambda$ we have on shell a conservation law
\begin{gather*}
\Xi^\mu \left( \frac{\partial p^0_\mu}{\partial \tau^0} + \frac{\partial p^1_\mu}{\partial \tau^1} \right) = 0 .
\end{gather*}

Symmetries of our Lagrangian can be explicitly determined from the above Noether equation. Since $D$ does not depend on $\tau^0$ and $\tau^1$, with help of momenta the Noether equation yields
\begin{gather*}
 T \frac{\partial\sqrt{-D}}{\partial
x^\mu} \Xi^\mu - p_\mu^i d_i\Xi^\mu =0 .
\end{gather*}
We can see that the equation does not contain the projection $\xi$ of $\Xi$ (which is a vector f\/ield on the base $\bX = \R^2$). This means that {\em every} vector f\/ield on $\R^2$ is a symmetry of~$\lambda$, in other words, we recover the fact that the action is reparametrization invariant. Now se can take into account only {\em vertical} symmetries (which are vector f\/ields on the space-time~$(M,g)$). If, moreover,
$(M,g)$ is the Minkowski space-time, i.e., $\R^4$ with the metric $\operatorname{diag} g=(1, -1, -1, -1)$, then~$D$ does not depend on the $x^\mu$'s, and we immediately obtain that
the Poincar\'{e} group is a symmetry group of~$\lambda$. The symmetries and the corresponding
Noether currents are ($s=1, 2, 3$, $\mu=0, 1, 2, 3$)
\begin{alignat*}{3}
& \Xi=\frac{\partial}{\partial x^\mu}, \qquad &&
\Psi=-p^1_\mu d\tau^0+p^0_\mu d\tau^1, &\\
& \Xi=x^s\frac{\partial}{\partial x^0}+x^0\frac{\partial}{\partial x^s}, \qquad &&
\Psi=\big({-}p_0^1x^s-p^1_sx^0\big) d\tau^0+\big(p^0_0x^s+p^0_sx^0\big) d\tau^1, &\\
 & \Xi=x^2\frac{\partial}{\partial x^1}-x^1\frac{\partial}{\partial x^2}\; \mbox{cycl}, \qquad &&
 \Psi=\big({-}p^1_1x^2+p^1_2x^1\big) d\tau^0+\big(p^0_1x^2-p^0_2x^1\big) d\tau^1 .&
\end{alignat*}

\subsection{Variational Lie derivative and cohomology}

The variational Lie derivative can be seen as a local
dif\/ferential operator which acts on cohomology classes trivializing them \cite{PaWi11}.
This fact has important consequences for symmetries and conservation
laws associated with {\em local variational problems} generating global
Euler--Lagrange expressions; see, e.g., \cite{BFFP03,FePaWi11,FrPaWi13}.

As an immediate example of that feature, let $\del(\cdot)$ denote a cohomology class def\/ined by a closed class in the variational sequence, and let $\varepsilon$ be a dynamical form locally variational: $E_{n+1} (\varepsilon) = 0$ (i.e., locally, on an open set $\bU_\iota$, $\varepsilon\equiv E_{\lam_{\iota}} = E_n (\lam_\iota)$). In the case $q=n+1$, Theorem \ref{HigherLieDer} implies that for every projectable vector
f\/ield $\Xi$ on $\bY$, $L_{J^s\Xi} \varepsilon =E_{n} (J^s \Xi_V \rfloor
\varepsilon)$, that is, the Euler--Lagrange form of the {\em global} Lagrangian $ \lambda = J^s \Xi_V \rfloor
\varepsilon$, which implies $\del (L_{J^s\Xi}E_{\lam_{\iota}})=
\del(E_{L_{J^r\Xi}\lam_{\iota}}) = 0$ even if the cohomology class
of $E_{\lam_{\iota}}$ is nontrivial. In particular, since
$E_{L_{J^r\Xi}\lam_{\iota}} \equiv E_n (L_{J^r\Xi}\lam_{\iota})= E_{n} (J^s\Xi_V \rfloor E_{\lam_{\iota}})$ we
see that Euler--Lagrange equations of the local problem def\/ined by
$L_{J^r\Xi}\lam_{\iota}$ are the same as the Euler--Lagrange equations of the
global problem def\/ined by $\lambda = J^s\Xi_V \rfloor E_{\lam_{\iota}}$. Summarizing,
a problem def\/ined by a family of local Lagrangians $L_{J^r\Xi}\lam_{\iota}$ is variationally equivalent to a global one~\cite{PaWi11,PaWiGa11}.

Analogous results can be obtained at any higher degree in the variational sequence. Another case of particular importance, the degree $n+2$, will be discussed below.

\subsection{Symmetries of Helmholtz forms and variational dynamical forms}

Invariance properties of classes in the variational sequence suggested to Krupka et al.~\cite{KrKr08,KKPS05,KKPS07} the idea that there should be a close relationship between the notions of local
variational triviality of a dif\/ferential form and invariance of its exterior derivative. Namely, `translating' Theorem~\ref{thmLie} to the Takens representation, we get~\cite{KKPS05,KKPS07}:

\begin{Theorem}
The Lie derivative of a source $(n+k)$-form $\varepsilon$ by a contact symmetry~$Z$ is locally variationally trivial if and only if the flow of~$Z$ leaves invariant the source $(n+k+1)$-form $E_{n+k}(\varepsilon)$.
\end{Theorem}

Applying this fact to $(n+1)$-forms we then get a rather surprising assertion:

\begin{Corollary}[\cite{KKPS05,KKPS07}]
Contact symmetries of Helmholtz forms transfer nonvariational dynamical forms to locally variational ones.
\end{Corollary}

This result has an interesting application: namely, to a system of equations which is not variational as it stands, and even does not possess any variational multiplier one can f\/ind a~va\-riatio\-nal system related by a contact transformation
Examples showing this interesting feature are elaborated in the papers~\cite{KKPS05,KKPS07}.

Furthermore, one can obtain interesting results in study of {\em global properties} of variational problems. Let us consider a variationally trivial Helmholtz-like form~$\eta$, i.e., let $E_{n+2} (\eta) =0$, and denote the corresponding {\em local} system of dyna\-mi\-cal forms by~$\varepsilon_\iota$, i.e., let $\eta = E_{n+1}(\varepsilon_\iota)$. If a~projectable vector f\/ield $\Xi$ on~$Y$ is a symmetry of $\eta$ then, by the above, each of the local dyna\-mi\-cal forms $\bar \varepsilon_\iota = L_{J^s\Xi} \varepsilon_\iota$ is locally variational (meaning that locally it is the Euler--Lagrange form of a~Lagrangian). In fact, applying Theorem~\ref{HigherLieDer}, case $k=2$, using the Takens representation with $\eta = R_{n+2}([\rho])$, and formula~(\ref{cocl}), we obtain
\begin{gather*}
 0 = L_{ J^s\Xi} \eta = L_{ J^s\Xi} E_{n+1}(\varepsilon_\iota) = E_{n+1} (L_{ J^s\Xi} \varepsilon_\iota) = 0 +
R_{n+2}\cE_{n+1}( J^{s} \Xi_V \rfloor [\rho] ) \\
 \hphantom{0}= E_{n+1} R_{n+1}(J^{s}\Xi_V \rfloor [\rho] ) = E_{n+1} R_{n+1} ( [J^{s}\Xi_V \rfloor \eta ] ) =
E_{n+1} {\cal I} (J^{s}\Xi_V \rfloor \eta) ,
\end{gather*}
so that $E_{n+1} (L_{ J^s\Xi} \varepsilon_\iota) = E_{n+1} {\cal I} (J^{s}\Xi_V \rfloor \eta) = 0$,
meaning that the system of local Euler--Lagrange forms $L_{J^s\Xi}\varepsilon_\iota$ is variationally equivalent (in the sense that they have the same Helmholtz form) with the {\em global dynamical form}
$ \varepsilon = {\cal I} (J^{s}\Xi_V \rfloor \eta)$. {\em Moreover}, $\varepsilon$ is locally variational, and therefore locally it is the Euler--Lagrange form of a local Lagrangian~\cite{PaWi12}.

We conclude that the variational problem def\/ined by a system of local Euler--Lagrange forms~$\bar \varepsilon_\iota$ is variationally equivalent to the global variational problem def\/ined by the Euler--Lagrange form~$\varepsilon$. Remarkably, in mechanics this `globalization procedure' means to add to the local equations a variational force, `sewing the equations together' on overlaps of their domains.

\subsection{Generalized symmetries generating Noether currents}

Representing the variational sequence by dif\/ferential forms is of fundamental importance for understanding the geometry of Noether theorems.

Let $\varepsilon = E_{\lambda_\iota}$ be a locally variational form on $J^s\bY$, and let a projectable vector f\/ield $\Xi$ on $\bY$ be a generalized symmetry, meaning that $L_{J^s\Xi} \varepsilon = 0$. Recall that by the properties of the Lie derivative this means that $E_n(L_{J^r\Xi} \lambda_{\iota}) = 0$, hence {\em locally} $L_{J^r\Xi} \lam_{\iota} = d_H \beta_{\iota}$. On the other hand, by Theorem \ref{HigherLieDer}, $L_{J^s\Xi} \varepsilon =E_{n} (J^s \Xi_V \rfloor
\varepsilon)$, hence $E_n (J^s \Xi_V\rfloor\varepsilon) = 0$ , and from here
we have {\em locally} $J^s\Xi_V\rfloor \varepsilon =d_H\nu_\iota$, where $\nu_\iota$ are local
currents which are conserved on-shell (i.e., along critical sections). Denoting by
$\eps_\iota$ the usual {\em canonical} Noether current def\/ined by Proposition \ref{1}, case $q=n$, i.e., $\eps_\iota =
J^s\Xi_V\rfloor p_{d_V\lam_\iota} + J^s \Xi_H \rfloor \lam_\iota$, and putting $\beta_\iota = \nu_\iota + \eps_\iota$,
we can write $J^s\Xi_{V} \rfloor \varepsilon + d_{H}( \eps_\iota - \beta_{\iota} )$
$=$ $0$. We call the
(local) current $\eps_{\iota} - \beta_{\iota}$ a {\em Noether--Bessel-Hagen current}.
Thus Noether--Bessel-Hagen currents $\eps_{\iota} - \beta_{\iota}$ are local currents
associated with a generalized symmetry (conserved along critical
sections); in \cite{FrPaWi13} it was proved that a
Noether--Bessel-Hagen current is variationally equivalent to a
global current if and only if $0 =\delta_{n-1} (\Xi_{V} \rfloor
\cE_{n}(\lam_{i})) \in H^{1}(\bY, \ker d_H)$.

It is of interest to investigate under which conditions a Noether--Bessel-Hagen current
is variationally equivalent to a Noether conserved current for a
suitable Lagrangian. A Noether--Bessel-Hagen current $\eps_{\lam_\iota} -\bet_\iota$ associated
with a generalized symmetry of $\varepsilon = E_{\lam_\iota}$ is a Noether
conserved current if it is of the form $\eps_{\lam_\iota} - L_{J^s\Xi}
\mu_\iota $, with $\mu_\iota$ a current satisfying
$L_{J^s\Xi}(\lam_\iota - d_H\mu_\iota)=0$ (i.e., $\Xi$ is a symmetry of the Lagrangian ${\bar \lambda}_\iota = \lam_\iota - d_H\mu_\iota$ for~$\varepsilon$)~\cite{PaWi14}. The reader can check that taking for an invariant Lagrangian instead of the canonical current ${J^s \Xi} \rfloor\theta$ a current, related to a Lepage equivalent dif\/ferent than the Cartan form $\theta$, we get a current which provides the same conservation law, see~\cite{olga09-2}.
We have the on shell conservation laws $d_H(\eps_{\lam_\iota} -\bet_\iota )=0$. It can be proved that $\eps_{\lam_\iota} - L_{J^s\Xi}
\mu_\iota= \eps_{\lam_\iota-d_H \mu_\iota}$ is not only closed but also exact on shell (for a detailed discussion see, e.g.,~\cite{CaPaWi16}). In the following we shall give an example of application concerning symmetries of Jacobi type equations for a Chern--Simons $3D$ gauge theory.

\begin{Example}
Let us then consider the $3$-dimensional Chern--Simons Lagrangian
\begin{gather*}
\lambda_{\rm CS}(A)=\frac {\kappa}{4\pi}\epsilon^{ijm} \operatorname{Tr} \left(A_i d_j A_m+\frac{2}{3} A_i A_j A_m \right) \ome_0 ,
\end{gather*}
where $\ome_0$ is a $3$-dimensional volume density, $\kappa$ a constant, while $A_i=A^\mu_i J_\mu$ are the coef\/f\/icients of the connection $1$-form $A=A_i dx^i$ taking their values in any Lie algebra $\mathfrak g$ with genera\-tors~$J_\mu$. By f\/ixing $\mathfrak g={\mathfrak{sl}}(2,\mathbb{R})$ and choosing the generators
$J_\nu=\frac{1}{2}\sig_\nu$, whith $\sig_\nu$ Pauli matrices, we have $[J_\mu,J_\nu]=\eta^{\sig\rho} \epsilon_{\rho\mu\nu} J_\sig$ and $\operatorname{Tr}(J_\mu J_\nu)=\frac{1}{2}\eta_{\mu\nu}$, with $\eta=\operatorname{diag} (-1,1,1)$ and $\epsilon_{012}=1$. Hence, we can explicitly write
$\lambda_{\rm CS}(A)$ $=$
$\frac{\kappa}{16\pi}\epsilon^{ijk} (
\eta_{\mu\nu} F^\mu_{ij} A^\nu_\rho- \frac{1}{3}\epsilon_{\mu\nu\rho}
A^\mu_i A^\nu_j A^\rho_k )ds$,
where $ F^\mu_{ij}=d_i A^\mu_j-d_j A^\mu_i+ \epsilon^\mu_{\nu\rho}A^\nu_i A^\rho_j$ is the so-called f\/ield strength.
Note that, while Chern--Simons equations of motion are covariant with respect to spacetime dif\/feomorphism and with respect to gauge transformations, the Chern--Simons Lagrangian, instead, is not gauge invariant. Indeed for a~pure gauge tranformation generated by a vertical vector f\/ield $\xi$ -- acting on the gauge potential as
$L_{\xi} A^{\mu}_{i} $ $=$ $ \nabla_i \xi^{\mu} $, where $ \nabla_i $ is the covariant derivative with respect to a principal connection~$\ome^\mu$ def\/ined by~$A^\mu$ in a usual way through the introduction of new variables~\cite{AFR03}~-- we have
$ \cL_{\xi } \lambda_{\rm CS}(A)$ $=$ $d_i \big(\frac{\kappa}{8\pi}\epsilon^{ijk} \eta_{\mu\nu} A^{\mu}_j d_k \xi^{\nu} \big)$.

Let $\nu_\iot + \eps_\iot$ be a $0$-cocycle of strong Noether currents for the Chern--Simons Lagrangian (here~$\iot$ refers to an open set $U_\iot$ in a good covering and in the following we shall use the identif\/ication $\lambda_{\rm CS}(A)=\lambda_{\rm CS}(\iot)$ and use them indif\/ferently); it is known~\cite{FrPaWi12} that
if $j^s\Xi$ is a symmetry of the Chern--Simons dynamical form,
the global conserved current
\begin{gather*}
\Xi_H \rfloor \cL_{\Xi} \lam_{{\rm CS}_{\iot}} + j^s\Xi_{V} \rfloor p_{d_{V}\cL_{\Xi} \lam_{{\rm CS}_{\iot}}} ,
\end{gather*}
is associated with the invariance of the Chern--Simons equations and it is
{\em variationally equivalent} to the variation of the strong Noether currents
$\nu_\iot+\eps_\iot$.
Generators of such a global current lie in the kernel of the second variational derivative and are symmetries of the variationally trivial Lagrangian $\cL_{j^s \Xi_V} \lambda_{\rm CS}(\iot)$.

Let now require $j^s \Xi$ to be a symmetry of the Euler--Lagrange class of the `deformed' Lag\-rangian $\cL_{j^s\Xi} \lam_{\rm CS}(\iot)$, but not of the Chern--Simons class $E_n(\lambda_{\rm CS}(\iot))$ (this means that~$j^s\Xi$ {\em is not} a gauge-natural lift).
Note that in this case $\cL_{j^s\Xi} \lam_{\rm CS}(\iot)$ needs not to be a local divergence, in general. By naturality and by Theorem \ref{HigherLieDer} this means that we require
\begin{gather}\label{higher Bessel}
\cL_{j^s\Xi} E_n (\Xi_V \rfloor E_n(\lambda_{\rm CS}(\iot)) =0,
\end{gather}
i.e., $j^s\Xi$ is a symmetry of a Jacobi type dynamical form (or Euler--Lagrange class of the deformed Lagrangian). Let now $j^s\Xi_V$ be the generator of a gauge transformation.
Since we know that $\cL_{j^s\Xi_V} \lambda_{\rm CS}(\iot) = d_H \gam$, it is clear that such a $j^s\Xi$ is a Noether symmetry of the Lagrangian
\begin{gather*}
\Xi_V \rfloor E_n(\lambda_{\rm CS}(\iot)- d_H\big(\gam - j^s\Xi_V\rfloor p_{d_{V}\lam_{\rm CS}}\big),
\end{gather*}
this comes from equation~\eqref{higher Bessel} and from
\begin{gather*}
\cL_{ j^s\Xi}(\cL_{j^s\Xi} \lambda_{\rm CS}(\iot) - d_H\eps_{\lam_{\rm CS}})= \cL_{j^s\Xi}( d_H\gam - d_H (\Xi_V\rfloor p_{d_{V}\lam_{\rm CS}})).
\end{gather*}
By applying the above results concerning Noether--Bessel-Hagen currents to the `deformed' Lag\-rangian we can then state that along sections pulling back to zero the Jacobi type dynamical form (note that solutions of Chern--Simons equations are among such sections) we have that
\begin{gather*}
d_H \big[j \Xi_{V} \rfloor p_{d_{V}(\cL_{\Xi_V} \lam_{{\rm CS}_{\iot}}-d_H \Xi_V\rfloor p_{d_{V}\lam_{\rm CS}})} \big]\equiv 0 ,
\end{gather*}
i.e., the corresponding Noether--Bessel-Hagen current must be exact on shell which implies a~constraint for the connection $1$-form $A$.
\end{Example}

\subsection[Sharp obstruction to the existence of global solutions in 3D Chern--Simons theories]{Sharp obstruction to the existence of global solutions\\ in 3D Chern--Simons theories}

We now explore the relationship between the existence of global critical sections and the existence of global Noether--Bessel-Hagen currents.
We denote by $[[ \cdot ]]_{\rm dR}$ a class in de Rham cohomology. As already recalled, the Noether--Bessel-Hagen current is variationally equivalent to a global current if and only if $0 =\delta_{n-1} (\Xi_{V} \rfloor
\varepsilon)\sim [[\Xi_{V} \rfloor \varepsilon]]_{\rm dR} \in H^{n}_{\rm dR}(\bY)$.

In the following we shall need to see certain cohomology classes in $H^{n}_{\rm dR}(\bY)$ as the pull-back of cohomology classes in $H^{n}_{\rm dR}(\bX)$, i.e., def\/ined by {\em closed differential forms on the base manifold}.

Let us then assume $\sig$ be a {\em global section} of $\pi\colon \bY\to\bX$; for any global section we have of course $\sig^{*}\circ \pi^{*} =\operatorname{id}_{H^{n}_{\rm dR}(\bX)}$. Let furthermore $H^{n}_{\rm dR}(\bY) \sim \pi^{*} (H^{n}_{\rm dR}(\bX))$, i.e., {\em the pull-back $\pi^{*}\colon H^{n}_{\rm dR}(\bX)$ $ \to $ $H^{n}_{\rm dR}(\bY)$ is an isomorphism of cohomology groups}.

We need to compare the cohomology of the base manifold with the cohomology of the total space, and while $H^{n}_{\rm dR}(\bY)\simeq H^{n}_{VS}(\bY)$, in principle we do not have a similar isomorphism for the cohomology groups $H^{n}_{\rm dR}(\bX)$. This is overcome by the peculiar structure of quotient space in the variational sequence: since for all sections
$J^r \sig^{*} (\Thd_{*}^{r})=0$, we can correctly def\/ine $J^r \sig^{*} [[\Xi_{V} \rfloor\varepsilon]]_{\rm dR} \doteq [[ J^r\sig^{*} (\Xi_{V} \rfloor\varepsilon) ]]_{\rm dR}$; in other words,
since $\Thd_{n}^{r}$ consists of contact $n$-forms the pull-back $\sig ^{*}$ factorizes through $\For_{n}^{r}/\Thd_{n}^{r} $.
Furthermore, there always exists a closed form $\alp \in \For_{n}^{r}$, e.g., a~Lepage equivalent, which represents the cohomology class $[[ \Xi_{V} \rfloor \varepsilon]]_{\rm dR}$ and which projects onto $\Xi_{V} \rfloor \varepsilon \in \For_{n}^{r} / \Thd_{n}^{r} $. Thus, if $\sig$ is a critical section, $J^r \sigma ^{*} (\alpha) = 0$ and the corresponding class vanishes in $H^{n}_{\rm dR}(\bX)$; therefore, by our assumption on the $n$th cohomology groups $[[\Xi_V \rfloor \varepsilon ]]_{\rm dR} = 0$ and there exist global Noether--Bessel-Hagen conserved currents for all generalized symmet\-ries~$\Xi$~\cite{FrPaWi13}.

We can use the result above also conversely to def\/ine a cohomological obstruction to the existence of global solutions for a given problem.
Let then $\sigma$ be a section of $\bY$ over $\bX$. Let $\alp$ be a~closed dif\/ferential form which represents
$[[ \Xi_{V} \rfloor \varepsilon ]]_{\rm dR}$ $\in$ $H^{n}_{\rm dR}(\bY)$, the obstruction to the existence of global conserved quantities for conservation laws associated with the problem $\varepsilon$ and its symmet\-ry~$\Xi$. Since $\sig$ (as well as any other section) def\/ines a projection $\sig^{*}\colon H^{n}_{\rm dR}(\bY) \mapsto H^{n}_{\rm dR}(\bX)$, if the class $J^r\sig^{*} [[\alp ]]_{\rm dR}\in H^{n}_{\rm dR}(\bX)$ does not vanish, then there are no global solutions.

In the case we drop the above condition on the isomorphism of cohomology groups, if the pull-back with a~given section of the cohomology class $[[\Xi_{V} \rfloor \varepsilon ]]_{\rm dR}$ is not zero, no other {\em homotopic} section can be a~critical section.
Note that, in this case, the obstruction {\em does depend} on the section; $J^r \sig^{*} [[\alp]]_{\rm dR}$ may vanish for sections from one homotopy class, but not for those from another~\cite{FrPaWi13}.

It is also important to note that this obstruction is not identically zero, in general.

As an example, let us consider the classical $3$-dimensional Chern--Simons theory. It is a~classical f\/ield theory for principal connections on an arbitrary principal bundle $\bP$ over a $3$-dimensional manifold~$\bM$.
It is well known that
the equations can be written $F = 0$ ($F$~the curvature);
these equations admit solutions if and only if the cohomology class $[[\operatorname{tr} F_{A} ]]_{\rm dR}$ (``f\/irst Chern class'') vanishes for an arbitrary (and, thus, for every) connection~$A$.
Thus, we will consider only principal bundles $\bP$ with group ${\rm U}(1)$ or $\C^*$ (taking the trace of the curvature corresponds to ``taking the determinant of $\bP$'').

To apply our result, we have to formulate Chern--Simons theory in terms of the bundle of connections:
 the bundle of connections is def\/ined as
 \begin{gather*}
 \mathcal{C} := J^{1}\bP/\bG \mapsto \bP/\bG \sim \bX ,
 \end{gather*}
 the principal connections on $\bP$ are in one to one correspondence with the (global) sections of the bundle $\mathcal{C} \mapsto \bX$;
Chern--Simons theory can then be formulated as a variational theory on~$J^{1}\mathcal{C}$.
In this case is it advantageous to lift the dynamical form of the problem to the de Rham complex:
we can f\/ind a dif\/ferential form $\Sigma$ on $J^{1}\mathcal{C}$ such that a section $\sigma$ of
$\mathcal{C} \mapsto \bX$ is critical if and only if
\begin{gather*}
J^{1}\sigma^{*}(\Xi \rfloor \Sigma)=0,
\end{gather*}
for all vector f\/ields $\Xi$ on $J^{1}\mathcal{C}$.

We have that the contact structure of $J^{1}\bP$ def\/ines a connection on the principal bundle $J^{1}\bP \mapsto \mathcal{C}$; its curvature is given by $\mathcal{F} = dA_{j} \wedge dx^{j}$ (structure group ${\rm U}(1)$ or $\C^*$).
This connection is universal in the sense that every principal connection on $\bP \mapsto \bX$ is ``induced'' by it via a~section $\sigma\colon \bX \mapsto \mathcal{C}$, in particular $F_{\sigma} = \sigma^{*}\mathcal{F}$.

 We can set now
\begin{gather*}
 \Sigma := \pi_{1,0}^{*} \mathcal{F} \wedge \mathcal{F},
\end{gather*}
where $\Sigma$ is, of course, the exterior dif\/ferential of the Poincar\'e--Cartan form.
Notice that a~projectable vector f\/ield $\Xi$ is a symmetry of $\Sigma$ if and only if $d(\Xi \rfloor \Sigma)=0$.
The corresponding class to $[[\Xi_{V} \rfloor \varepsilon]]_{\rm dR}$ for a generic $\varepsilon$ is now
 $[[\Xi \rfloor \Sigma]]_{\rm dR}= 2\pi_{1,0}^{*} ((\Xi \rfloor \mathcal{F}) \wedge \mathcal{F})$.

If the Euler--Lagrange equations admit global solutions, then $\mathcal{F} = d\gamma$ and, thus,{\samepage
\begin{gather*}
[[\Xi \rfloor \Sigma]]_{\rm dR} = 2 \pi_{1,0}^{*} ((\Xi \rfloor \mathcal{F}) \wedge d\gamma) = 2d(\pi_{1,0}^{*} (\Xi \rfloor \mathcal{F} \wedge \gamma).
\end{gather*}
If the equations do not admit global solutions then $[[\mathcal{F}]]_{\rm dR} \neq 0$.}

Now, if $\bX$ is a closed manifold, then there is a closed cohomologically nontrivial $1$-form~$\alpha$ such that $\alpha \wedge F_{\sigma}$ and, hence, $(\pi^{*} (\alpha)) \wedge \mathcal{F}$ are cohomologically non trivial (Poincar\'e duality).

Since $\bX$ is parallelizable, we can easily f\/ind a {\em vertical} vector f\/ield $\Phi$ such that $\Phi\rfloor \mathcal{F} = \pi^{*} (\alpha)$ and
\begin{gather*}
[[\Phi\rfloor \Sigma]]_{\rm dR} = 2 \pi_{1,0}^{*} (\pi^{*} (\alpha) \wedge \mathcal{F}) .
\end{gather*}
Thus, for closed $3$-manifolds Chern--Simons theories admit solutions if and only if
\begin{gather*}
[[\Phi \rfloor \Sigma]]_{\rm dR}=0 .
\end{gather*}

For non compact manifolds this is, of course, no longer true; thus, for non compact $3$-manifolds the question, whether the obstruction is sharp, depends on the existence of a suitable compactif\/ication, see also~\cite{PaWi15}.

\subsection*{Acknowledgements}

Research supported by Department of Mathematics, University of Torino, research projects {\em Geometric methods in mathematical physics and applications} 2014 and 2015, GNSAGA-INdAM and by grant no 14-02476S {\em Variations, Geometry and Physics} of the Czech Science Foundation. MP and EW are grateful to
M.~Ferraris for useful discussions. OR also wishes to acknowledge travel support by the IRSES project LIE-DIFF-GEOM (EU FP7, nr~317721) and the hospitality of the Department of Mathematics, Ghent University, Belgium.

%\cite{BVH,Gr99}

\pdfbookmark[1]{References}{ref}
\LastPageEnding

\end{document}